%% file: equilibrium_world_model.tex
\documentclass[11pt, fleqn]{article}

\usepackage[margin=1in]{geometry}
\usepackage{soul}
\usepackage{setspace}
\usepackage{amsmath,amssymb,amsthm,mathtools}
\usepackage{bm}
\usepackage[T1]{fontenc}
\usepackage{newpxtext,newpxmath}
\usepackage{booktabs}
\usepackage{array}
\usepackage{graphicx}
\usepackage[dvipsnames]{xcolor}
\usepackage{tikz}
\usetikzlibrary{decorations.pathreplacing,arrows.meta,positioning}
\usepackage{enumitem}
\usepackage{float}
\usepackage{caption}
\usepackage{algorithm}
\usepackage{algpseudocode}
\usepackage[authoryear,round]{natbib}
\usepackage{hyperref}
\hypersetup{colorlinks=true,linkcolor=blue,citecolor=blue,urlcolor=blue}

\newtheorem{definition}{Definition}
\newtheorem{assumption}{Assumption}
\newtheorem{theorem}{Theorem}
\newtheorem{proposition}{Proposition}
\newtheorem{lemma}{Lemma}

\newtheorem*{remark}{Remark}

\newcommand{\R}{\mathbb{R}}
\newcommand{\E}{\mathbb{E}}

\newcommand{\X}{\mathcal{X}}
\newcommand{\A}{\mathcal{A}}
\newcommand{\calE}{\mathcal{E}}
\newcommand{\calR}{\mathcal{R}}
\newcommand{\calB}{\mathcal{B}}
\newcommand{\calF}{\mathcal{F}}
\newcommand{\calK}{\mathcal{K}}

\newcommand{\calT}{\mathcal{T}}
\newcommand{\muK}{\mu_\kappa}
\newcommand{\muerg}{\mu_{\pi_\theta}}
\newcommand{\nnp}{\pi_\theta}
\newcommand{\nnet}{\mathcal{N\!N}}
\newcommand{\wm}{\widehat{W}_\psi}
\newcommand{\Qpi}{Q^{\pi_\theta}}

\newcommand{\rwm}{R}
\newcommand{\TV}{\mathrm{TV}}
\newcommand{\sg}{\mathrm{sg}}
\DeclareMathOperator*{\argmin}{arg\,min}

\newcommand{\SCE}{\mathrm{SCE}}
\newcommand{\RE}{\mathrm{RE}}
\newcommand{\Normal}{\mathcal{N}}

\makeatletter
\renewcommand\paragraph{\@startsection{paragraph}{4}{\z@}
  {1ex \@plus .25ex \@minus .15ex}
  {-1em}
  {\normalfont\normalsize\bfseries}}
\makeatother

\title{\Large \textbf{Equilibrium World Models}}

\author{
\begin{minipage}[t]{0.46\textwidth}
\centering
\textsc{Andreas Schaab}\\[0.4em]
\small
Department of Economics\\
University of California, Berkeley, USA\\
\texttt{schaab@berkeley.edu}
\end{minipage}
\hfill
\begin{minipage}[t]{0.46\textwidth}
\centering
\textsc{Simon Scheidegger}\thanks{\footnotesize
Simon Scheidegger wrote part of this paper while visiting the Bank for
International Settlements (BIS), and thanks the BIS for its hospitality. We thank
Marlon Azinovic-Yang, Oliver Surbek, and Jan {\v Z}emli{\v c}ka for helpful comments and
discussions.}\\[0.4em]
\small
Department of Economics\\
University of Lausanne, Switzerland\\
Grantham Research Institute on Climate Change and the Environment, LSE, UK\\
\texttt{simon.scheidegger@unil.ch}
\end{minipage}
}

\date{\today}

\input{bewley_ledger_values.tex}

\input{bewley_ledger_stubs.tex}

\begin{document}
\setlength{\parskip}{0pt}
\setlength{\parindent}{15pt}
\begin{singlespace}

\maketitle
\vspace{-.4in}

\begin{abstract}
We introduce \emph{Equilibrium World Models} (EWMs), a deep-learning method for globally solving
dynamic stochastic models that feature rare disasters, binding constraints, and counterfactual states. Standard
unsupervised neural-network-based solvers impose equilibrium conditions only on states generated by their own
simulated policy. Their solutions can therefore be self-confirming: accurate on the simulated path,
but untested off it, sensitive to initialization, and costly when expectations must be recomputed at
each step. EWMs change the computational representation, not the economics. They enforce the model's
exact equilibrium conditions on a broader, model-generated distribution of ordinary, rare, stressed,
and counterfactual states. They carry the continuation with a learned surrogate, but certify the
resulting policy strictly against the true equilibrium conditions. We provide an error decomposition,
an off-path residual bound, and a convergence result linking self-confirming solutions to
rational-expectations equilibria. We demonstrate EWMs through a sequence of test cases that isolate the
main pathologies of classical deep-learning solvers and then scale them to richer economies. In a rare-disaster
Brock--Mirman laboratory, coverage reduces disaster-region residuals by an order of magnitude. In a
high-dimensional international real-business-cycle model, classical deep-learning solvers fail from all random
starts, whereas EWMs converge from nearly all and evaluate continuations up to two orders of magnitude
less often. When actions move transition measures, EWMs use action-conditioned continuations to
recover the relevant policy margin. In a heterogeneous-agent economy with aggregate risk, EWMs compress
the numerical representation of the wealth distribution by at least $25\times$ while imposing exact
full-distribution rational-expectations conditions.
\end{abstract}



\clearpage

\section{Introduction}\label{sec:intro}

\paragraph{Motivation and challenge.}
Many central questions in macroeconomics and finance are decided far from the deterministic steady
state. Monetary policy is most informative when constraints bind unevenly across households; financial
models are most revealing when leverage, asset prices, and default risk interact away from normal times;
and climate and disaster models are read precisely in the tails, where rare events dominate welfare and
asset prices. These are the regions in which a local approximation discards the economics and a fixed
grid cannot reach, so answering such questions calls for a global solution\footnote{We use the term
\emph{global solution} for a solution computed from a model's equilibrium conditions at many points across
its state space, in contrast to a \emph{local solution}, which rests on a local approximation around a
steady state; a method that delivers one is a \emph{global solution method}. This sense of ``global'' is
not the one in ``global optimization,'' which refers to finding a global optimum.} of a dynamic
stochastic model on a high-dimensional state space with occasionally binding constraints and pronounced
nonlinearities \citep{bellman1961adaptive}. Over the past decade, unsupervised deep-learning solvers have
become a leading tool for exactly this task. Deep Equilibrium Nets \citep{azinovicDEEPEQUILIBRIUMNETS2022}
and deep-learning Euler-residual methods \citep{maliarwinant2021} represent the equilibrium policy by a
neural network and train it without labelled data or a precomputed reference solution; the only training
signal is the model's own equilibrium conditions, evaluated from the structural equations and driven
toward zero.\footnote{See \citet{fernandezvillaverde2024taming} for a recent review and
\citet{scheidegger2026deeplearning} for a textbook treatment.}

The difficulty is not that these solvers abandon economic structure. By construction they impose the
model's defining equilibrium equations, so what is in doubt is not the economics they enforce but the
domain on which the returned solution is certified. A pathwise solver imposes those equations only on the
states its own simulation visits, so a candidate solution can be made accurate on the policy's typical
path while remaining untested off it: after rare shocks, near binding constraints, and along the
counterfactual responses that policy analysis traces. In this sense the computation is self-confirming. It
is certified on the data its own policy generates and left silently uncertified precisely where it will be
read. This is not merely extrapolation into a rare region but a failure of certification, and closing it,
by imposing the model's exact equilibrium conditions on the wider region where the solution will be used
and not only on the path that confirms it, is the problem this paper takes up.

\paragraph{A structural world model.}
We introduce \emph{Equilibrium World Models} (EWMs), an unsupervised, simulation-based deep-learning
method that makes the certificate follow the region in which the solution will be used. The idea is a
disciplined dreamer. In artificial intelligence, a world model lets an agent improve by rehearsing
imagined experience inside a model of its environment rather than by acting in the world itself
\citep{Ha2018WorldModels,hafner2020dream,hafner2023dreamerv3}; EWMs give a deep equilibrium solver the
same ability to rehearse states its own simulation would rarely encounter. But the analogy stops at
rehearsal. Every imagined state is generated or advanced through the model's maintained law of motion, and
every candidate policy is judged by the model's own equilibrium residual, so the construction preserves
the structural discipline emphasized by the Lucas critique \citep{Lucas1976Critique} rather than fitting a
reduced-form simulator in its place. EWMs change the computational representation, not the economics: the
transition law, the equilibrium conditions, and the reported accuracy criterion remain the model's exact
structural objects, and what changes is the support on which those conditions are enforced and the way the
forward-looking continuation is carried. To the best of our knowledge, EWMs are the first structural use
of the world-model idea for computing rational-expectations equilibria in dynamic stochastic economic
models.

\paragraph{Coverage and audited world components.}
The method has two components. First, EWMs replace the purely pathwise training distribution by a broader,
model-generated coverage measure that contains the policy's ordinary simulated states together with
structurally admissible rare, stressed, locally perturbed, and counterfactual states generated or advanced
through the model's transition; the same equilibrium conditions are then enforced not only where the
policy usually goes but also where the solution will be read. Second, EWMs carry the forward-looking
continuation with audited learned objects that make this broader enforcement affordable: in
finite-dimensional economies a continuation surrogate\footnote{A surrogate, also called an emulator or
response surface, is a cheap statistical approximation to an expensive computation; here it stands in for
the continuation value, replacing the costly per-state expectation the equilibrium conditions would
otherwise require \citep{chen2025Deep}.} amortizes the conditional expectation entering the Euler equation;
when an action itself moves the law of future states, the continuation is conditioned on the action as
well as the state, recovering a policy margin a state-only continuation would average away; and when the
aggregate state is a distribution, a learned encoder compresses the cross-section into a tractable
representation. These learned objects are instruments for solving the model, not replacements for it: the
transition law is never learned, the equilibrium residual is never replaced, and the reported accuracy is
always the exact equilibrium residual on held-out states the policy never trained on. The certificate is
therefore explicit about its own domain: it validates the policy where that held-out residual is evaluated
and makes no claim for regions outside the coverage design.

\paragraph{From self-confirmation to rational expectations.}
Enlarging coverage is not importance sampling under another name. Drawing more states from the same
simulated path only estimates the same pathwise object more accurately; it imposes no condition where that
path never goes. Coverage instead changes the population objective itself, the set of states on which the
candidate equilibrium must satisfy the model's restrictions, so each reach defines a distinct fixed point
rather than a sharper estimate of one target. At zero additional reach the object is the computational
analogue of a self-confirming equilibrium \citep{fudenberg1993selfconfirming,sargent1999conquest},
certified only on the policy's own simulated experience; as coverage expands and the approximation,
continuation, and optimization errors vanish, the same residual conditions come to hold on the
coverage-certified region, that is, on the part of the reachable state space the audit deliberately
covers. Coverage therefore moves the computed object from self-confirmation toward rational expectations
on the audited region, a conditional limit the theory below makes precise.

\paragraph{Contributions.}
The paper makes three contributions. The first is conceptual: it identifies pathwise residual
minimization as a form of computational self-confirmation, a failure of certification driven by the
endogenous support of the training distribution rather than by rare-event extrapolation alone, and so
clarifies why low on-path residuals do not by themselves license counterfactual analysis. The second is
methodological: a structural world model that widens the support of equilibrium enforcement while using
learned continuations and learned distributional representations only as audited computational devices,
under an unchanged economic model and an unchanged exact-residual audit. The third is theory and evidence:
an error decomposition, an off-path counterfactual bound, and a conditional convergence result establish
the coverage-indexed family as a constructive route from self-confirmation toward rational expectations on
the audited region, with the change in support shown to deliver off-path accuracy, seed robustness,
affordability, and distributional scalability in models that demand global nonlinear solutions.

The method is not designed to improve on-path precision for its own sake: on its typical path a flexible
solver can already drive the residual to tolerance. The binding constraint is computational robustness,
whether the solver converges from arbitrary starts, whether the solution is certified where it will be
used, and whether the certificate remains affordable as the model grows.

\paragraph{Three microscopes and a diagnostic variant.}
We develop and test the method in three discrete-time economies, chosen not as showcases but as
microscopes. Each is the simplest setting in which one demand on the method can be isolated and audited by
the same exact held-out residual, and together they raise the aggregate state from a single capital stock
to a full wealth distribution.

The first is a rare-disaster version of the textbook Brock--Mirman growth model with irreversible
investment. It is small enough that the coverage gap can be seen directly, yet it already contains the
failure that matters: after a rare productivity collapse the irreversibility constraint binds, precisely
where a pathwise policy has least reason to be accurate. Imposing the equilibrium conditions on a
model-generated set of disaster states cuts the residual there by about an order of magnitude, and a
variant that widens coverage without the learned component captures essentially all of that gain,
identifying the expansion of support, not the learned shortcut, as the lever.

The second is an international real business cycle model with rare disasters and irreversible investment,
and it is the headline economy. The number of countries sets the dimension of the state, the rare disaster
carves out a region the simulated path almost never visits, and the irreversibility constraint adds a kink
a smooth on-path policy cannot extrapolate to. The result is reliability, and it is starkest in
convergence. In the reported protocol the pathwise solver reaches a verified solution from none of its
starts, with disaster-region errors an order of magnitude larger than on-path errors. EWMs converge from
nearly all starts and stay certified in the disaster region, while reducing per-query continuation
evaluations by up to two orders of magnitude, a saving that grows with the number of countries and lets
the method keep converging as the state dimension rises. The difference is economic, not merely numerical:
through the disaster the certified policies price disaster risk consistently, while the pathwise policy,
uncertified there, misprices it.

A diagnostic variant isolates a separate failure mode, one orthogonal to dimension. When an action changes
the probability law of future states rather than the next state directly, a state-only continuation
evaluates only the action the policy already prescribes and averages away the margin that action controls.
The continuation must then be conditioned on the action as well as the state, the analogue of a Q-function
in dynamic programming. In a Brock--Mirman economy with an endogenous protection choice that lowers
disaster exposure, a state-only continuation collapses the protection margin to zero, while the
action-conditioned continuation recovers it and passes the held-out exact-residual audit. Coverage decides
where the equilibrium conditions are imposed; this variant shows that the world component itself must
carry the policy margin the model requires.

The third is a heterogeneous-agent economy with aggregate risk, in which the state is the entire
cross-sectional distribution of wealth. This is where the structural-world-model view does its most
interesting work, because it meets a long-standing tension head-on. To solve such a model under rational
expectations a household must in principle forecast the behavior of a whole distribution, an
informational burden that has made the rational-expectations benchmark itself contentious here
\citep{Moll2026EJ}, even as the Lucas--Sargent tradition prizes exactly the cross-equation discipline
that rational expectations imposes \citep{Lucas1976Critique,Sargent2024MacroAfterLucas}. EWMs separate how
the population is represented from how the equilibrium is disciplined. The policy reads only a
low-dimensional learned summary of the cross-section, an encoder trained as a Joint Embedding Predictive
Architecture\footnote{An encoder is a learned map that compresses a high-dimensional object, here the
wealth distribution, into a few numbers that keep what matters for decisions; trained as a Joint Embedding
Predictive Architecture (JEPA), it is paired with a predictor of next period's embedding
\citep{LeCun2022AMI,MaesLeLidec2026LeWM}.} a representation-learning analogue of the moment-based
tradition, from the hand-chosen moments of \citet{krusell1998} to the learned generalized moments of
\citet{han2023deepham}, while the population is still advanced by the model's exact distributional law and
the policy is still held to the exact, full-distribution rational-expectations conditions. The numerical
representation of the distribution is compressed by at least $25\times$; at the reported budget, among the
raw histogram, the classical hand-chosen moments, and the learned summary, only the learned summary
recovers the decision-relevant shape of the cross-section. This provides a computational compromise within
the maintained model: the policy uses a low-dimensional perception of the distribution, while the audit
remains the exact full-distribution rational-expectations residual.

\paragraph{Organization of the paper.}
Section~\ref{sec:related} reviews the related literature. Section~\ref{sec:framework} develops the
Equilibrium World Model object, the coverage gap it addresses, the audited continuation component, and the
coverage-indexed objective, illustrated by a Brock--Mirman laboratory. Section~\ref{sec:theory} proves the
error decomposition, the off-path counterfactual bound, and the conditional convergence from
self-confirmation to rational expectations as coverage expands. Section~\ref{sec:exp} reports the numerical
experiments, residual audits, robustness statistics, and compute accounting across the three economies.
Section~\ref{sec:conclusion} concludes, with proofs, the action-conditioned protection variant, and
implementation details in the appendices.

\section{Related Literature}\label{sec:related}

This paper is closest to five literatures: (i) unsupervised deep-learning solvers for
rational-expectations equilibria; (ii) parameterized-expectations and surrogate methods; (iii)
self-confirming and learning equilibria; (iv) heterogeneous-agent and mean-field economies; and (v)
world models in artificial intelligence. These neighboring literatures can be read along the same
self-confirming-to-rational-expectations axis the coverage sieve traces: adaptive-learning and
self-confirming approaches at one end, global rational-expectations solvers at the other, and
restricted-representation or forecasting methods in between. The organizing distinction is
simple. EWM leaves the structural economy unchanged: the transition law, equilibrium residual, and
final accuracy audit are the modeler's exact objects. It changes the support on which those exact
conditions are enforced, replacing the policy's own simulated path by a model-generated coverage
measure that includes rare, stressed, and locally perturbed states. Learned components enter only to
make this wider enforcement affordable, or to compress high-dimensional states; they do not define
or certify equilibrium.

\paragraph{Deep-learning-based solution techniques.}
The most direct benchmark is the class of unsupervised neural network solvers for dynamic stochastic
equilibria. Four features make their global solution demanding
\citep{azinovicDEEPEQUILIBRIUMNETS2022}: (a)~stochasticity, which forces conditional expectations
over future shocks; (b)~high-dimensional state spaces, where the curse of dimensionality
\citep{bellman1961adaptive} obstructs grid-based representations; (c)~kinks and strong
nonlinearities from occasionally binding constraints; and (d)~irregular ergodic geometry that makes
uniform grids wasteful or infeasible. Deep Equilibrium Nets train a policy network by minimizing squared structural residuals
along simulated paths \citep{azinovicDEEPEQUILIBRIUMNETS2022}; \citet{maliarwinant2021} unify
Euler, Bellman, and ``all-in-one'' residual training; and DeepHAM extends the approach to
heterogeneous-agent economies through learned generalized moments \citep{han2023deepham}. Related
methods now appear in continuous-time finance \citep{duarte2024ml,gopalakrishna2024aliens},
high-dimensional dynamic programming \citep{kahou2021exploiting}, heterogeneous-agent macroeconomics
\citep{FernandezVillaverdeHurtadoNuno2023,azinoviczemlicka_2024}, climate economics
\citep{Folini_2021}, search and matching \citep{payne2024deepsam}, monetary policy
\citep{nuno2024monetary}, and production networks \citep{carvalho2025planning}.

EWM keeps the same unsupervised residual-learning principle and the same equilibrium equations, but
changes where the equations are imposed. A pathwise solver certifies the policy on the states the
policy itself visits. EWM instead trains and audits on a coverage measure generated from the model's
own transition, thereby testing the rare and counterfactual states that matter for disasters,
binding constraints, and prices. Relative to this literature, the novelty is not a new residual or a
larger network. It is to make the enforcement distribution an object of the method, to identify the
coverage term in the equilibrium error, and to bound off-path accuracy by quantities observable on
the coverage set. This is complementary to sequence-space neural network solvers, which stabilize the input
distribution by conditioning on exogenous shock histories \citep{azinovicyangzemlicka2025sequencespace};
EWM instead changes the support on which the endogenous-state solution is certified. We do not claim
EWM dominates these neighboring solvers on the problems they were built for: a
parameterized-expectations method is fast and stable on a low-dimensional system with a known
approximation class, and a well-tuned pathwise solver is already accurate on its own ergodic path.
The contribution is the narrower one of certifying the off-path region those methods leave untested.

\paragraph{Parameterized expectations and surrogate methods.}
EWM also connects to parameterized-expectations methods \citep{marcet:88,denhaanmarcet1990,denhaan2010},
their early neural network implementations \citep{duffy2001approximating}, recent machine-learning
projection approaches \citep{valaitisvilla2024}, and the broader use of Gaussian-process or
neural-network-based surrogates for expensive economic computations
\citep{Rasmussen:2005:GPM:1162254,Scheidegger2019JCS,kase2022estimating,friedl2023deepuq,kublerscheideggersurbek2025carbontax,chen2025Deep}.
The mechanical overlap is clear: EWM may replace repeated costly continuation evaluations by a
cheap learned approximation.

The economic role of that approximation is different. In a parameterized-expectations method, the
fitted expectation is the object whose fixed point defines the computed policy, and accuracy is
usually assessed on the simulated paths that the policy generates. In EWM, the learned continuation
is an audited computational device. The structural transition is not learned, the residual remains
the exact equilibrium restriction, and the reported certificate is the held-out exact residual on the
coverage set. This separation has empirical content: the coverage-only version, with no continuation
surrogate, already delivers the off-path gains in the benchmark model; when the surrogate helps, it
helps by amortizing cost or improving conditioning, not by changing the equilibrium concept. The same
logic distinguishes coverage from standard active learning. Active learning samples points to
approximate a fixed target efficiently \citep{mackay1992information,snoek2012practical}; EWM chooses
support for an endogenous fixed-point constraint, so changing the support changes which
self-confirming object is ruled out.

\paragraph{Self-confirming, restricted-perceptions, and learning equilibria.}
A third strand studies equilibria in which beliefs are correct only in a limited sense.
Self-confirming equilibria require beliefs to be correct on histories generated by the agent's own
behavior, while off-path beliefs remain unrestricted \citep{fudenberg1993selfconfirming,sargent1999conquest}.
A closely related notion is Berk--Nash equilibrium, in which an agent's beliefs are optimal within a
possibly misspecified model and confirmed by the data her own actions generate \citep{espondapouzo2016}.
Restricted-perceptions, adaptive-learning, and self-justified equilibria instead restrict the
forecasting rules, subjective models, or learning processes available to agents
\citep{MarcetSargent1989,EvansHonkapohja2001,BranchEvans2006RPE,KUBLER2023105707,10.1093/jeea/jvaf062}.
Related internal-rationality and bounded-rationality approaches include
\citet{bray1982,sargent1993bounded,eusepipreston2011,adammarcetnicolini2016}.

EWM uses this literature diagnostically rather than behaviorally. A flexible pathwise neural network solver
is not mainly restricted by the class of forecasts it can represent; on its own simulated support it
can make the residual small. It is restricted by observation: the residual is never imposed at states
outside that support. The resulting fixed point is therefore self-confirming in a computational
sense. EWM replaces the single pathwise fixed point by a coverage-indexed family. At zero additional
coverage it coincides with the pathwise, self-confirming object. As coverage expands over the
reachable state space, and approximation, surrogate, and optimization errors vanish, the exact
residual holds on all reachable states, yielding rational expectations. Relative to this strand, the
novelty is a constructive sieve from computational self-confirmation to rational expectations, with
the off-path residual measuring the distance still left to cover. We use these notions as a
computational analogue, not an economic identification: there is no agent whose misspecified belief
is confirmed by her own feedback, only a solver whose residual is confirmed on the states its own
iterate visits. The equilibrium we ultimately certify is the standard rational-expectations one, on
the reachable states the coverage measure attains.

\paragraph{Heterogeneous-agent and mean-field economies.}
In heterogeneous-agent economies with aggregate risk, the aggregate state is a distribution. The
classic Krusell--Smith approach summarizes that distribution by a few moments \citep{krusell1998};
recent neural network methods learn richer generalized moments \citep{kahou2021exploiting,han2023deepham};
and continuous-time and mean-field formulations characterize values and distributions through
coupled Hamilton--Jacobi--Bellman, Kolmogorov-forward, and master equations
\citep{lasry2007mean,cardaliaguet2019master,achdou2022,gu2024masterequations,payne2024deepsam}.
The informational burden of rational expectations in such economies is substantial
\citep{Moll2026EJ}, even though the Lucas--Sargent tradition emphasizes the cross-equation
discipline that rational expectations provides \citep{Lucas1976Critique,Sargent2024MacroAfterLucas}.

EWM contributes on this computational margin. Its distributional world component is an encoder of
the cross-section: the policy reads a compact learned summary rather than the full population state.
The novelty is the discipline imposed on that summary. The population is still pushed forward by the
exact structural transition, the exact rational-expectations residual is evaluated at off-equilibrium
populations the realized economy does not generate, and the held-out exact residual is the audit.
Thus the encoder replaces hand-chosen moments as a state representation, not the law of motion of the
distribution. Relative to structural reinforcement learning that learns low-dimensional price laws
along realized population paths \citep{YangWangSchaabMoll2026SRL}, EWM keeps the
rational-expectations residual as the target and imposes it off path. The contribution is a way to
make a full-distribution rational-expectations benchmark operational while the policy itself uses
only a low-dimensional perception of the distribution.

\paragraph{World models and model-based learning.}
Finally, EWM borrows language and one design principle from world models in artificial intelligence.
Model-based reinforcement learning learns a predictive model of an unknown environment and improves a
policy by simulating inside that learned model, from early differentiable world models to World
Models, Dreamer agents, and joint-embedding predictive architectures
\citep{Schmidhuber1990MakingDifferentiable,Ha2018WorldModels,hafner2020dream,LeCun2022AMI,hafner2023dreamerv3,MaesLeLidec2026LeWM}.
Related ideas have entered finance through learned market simulators and model-based trading
environments \citep{Huang2022MBRLTrading,LiLiu2025MarS}.

The novelty of EWM is to make the world model structural. In artificial intelligence, the transition
is part of what must be learned. In EWM, the transition is the economist's model and is known. The
imagined states are generated by that transition, and the policy is trained to satisfy the model's
optimality, complementarity, and market-clearing conditions at those states. The learned component is
therefore not a simulator of an unknown environment; it is a device for expanding and amortizing the
enforcement of known equilibrium restrictions. To our knowledge, this is the first use of the
world-model idea to compute, rather than relax or simulate around, a rational-expectations
equilibrium.

\section{Equilibrium World Models}\label{sec:framework}

This section introduces EWM systematically, defining the method one object at a time and then making
it concrete on a transparent worked example. We take as the baseline the unsupervised pathwise
residual solver, a well-understood and operational class for which DEQN
\citep{azinovicDEEPEQUILIBRIUMNETS2022} are the canonical representative; we build on it precisely
because it is a strong, representative solver, not because of any deficiency.

EWM modifies that baseline in exactly two places. First, it enlarges the support on which the exact
structural residual is imposed, from the policy's own ergodic measure $\muerg$ to a broader,
model-generated coverage measure $\muK$ that reaches the rare, stressed, and counterfactual states
the policy never visits. Second, on that enlarged support it amortizes the expensive per-state
quadrature $\Qpi$ for the forward-looking continuation with an audited learned surrogate $\wm$.
Neither change is a black-box efficiency device. Enlarging the support is not importance sampling: it
changes which equilibrium is solved, and certified, rather than reweighting an estimator of a fixed
one, so that changing the support changes which self-confirming solution is ruled out. And the
surrogate is not a parameterized-expectations method: it is trained against the exact continuation
and never defines the equilibrium, whose error is always reported by the exact held-out residual, so
it amortizes cost without changing the equilibrium concept (Section~\ref{sec:related}). The economy,
the exact transition $\Gamma$, the structural residual $\rwm$, and the reported accuracy metric are
left untouched. Formally, at coverage reach $\kappa$ an equilibrium world model is the
tuple
\begin{equation}
\mathrm{EWM}_\kappa \;=\; \bigl(\,\nnp,\ \Gamma,\ \rwm,\ \muK,\ \wm,\ \mathcal{A}_H\,\bigr),
\end{equation}
collecting the equilibrium policy $\nnp$; the exact, known transition $\Gamma$; the exact
equilibrium residual $\rwm$; the coverage measure $\muK$ on which that residual is imposed; the
learned world component $\wm$ that carries the continuation (a surrogate for the
conditional-expectation term $\Qpi$ in the Euler equation, the expected discounted marginal value of
next period); and the held-out exact-residual audit
$\mathcal{A}_H$. The transition is never learned and the residual never
replaced; the world component only makes imposing the exact residual on a deliberately enlarged set
of states computationally feasible.

The section is organized components first: we build up the parts of EWM before illustrating the
assembled method. Section~\ref{sec:gap} specifies the baseline solver and the coverage gap it leaves;
Section~\ref{sub:econ} defines EWM together with its exact and surrogate residual evaluations and its
training losses; Section~\ref{sub:components} separates the three components that are fixed in every
model (the policy, the exact transition, and the exact residual) from the three an EWM adds and a
modeler adapts to the application (the coverage measure, the world component, and the audit);
Section~\ref{sub:coverage} develops the coverage measure and the sieve it induces from a
self-confirming equilibrium toward rational expectations; and Section~\ref{sub:discipline} gives the
training algorithm. We then illustrate EWM through variants of the Brock--Mirman model, first a benign
version (Section~\ref{sub:worked}) and then one with a rare disaster that makes the coverage gap
visible (Section~\ref{sub:bmdc}); Section~\ref{sub:variants} catalogues the world-component
variants, including an action-conditioned world model, whose continuation depends on the action and
not only the state, needed when the action itself shifts the distribution of next-period states (the
endogenous-protection case). These Brock--Mirman studies are illustrative, meant to make the
construction concrete and visible; the careful, comprehensive evaluation on the larger
international real business cycle model and the heterogeneous-agent Bewley economy follows in
Section~\ref{sec:exp}.

\subsection{The baseline pathwise solver and the coverage gap}\label{sec:gap}

We write a recursive stochastic economy as a tuple $(\X,\A,\Gamma,\rwm)$. The state $x\in\X\subseteq\R^n$ evolves through an exact, known
transition
\begin{equation}
x' \;=\; \Gamma(x,a,\varepsilon'), \qquad \varepsilon'\sim\nu,
\label{eqn:transition}
\end{equation}
where $a=\nnp(x)\in\A\subseteq\R^p$ stacks the agents' choices and Lagrange multipliers and is the
output of a neural network with trainable parameters $\theta$, and $\varepsilon'\in\R^{k}$ is an
exogenous shock drawn from a known distribution $\nu$ on $\R^{k}$, in general multivariate and a
product of heterogeneous marginals: its components may follow different laws, for instance Gaussian
innovations of AR(1) productivity processes alongside discrete regime or disaster shocks. For both
solvers we study, the pathwise baseline (DEQN) and EWM, and for the theoretical analysis, we fix a
quadrature rule and identify the exact model with this finite discretization, so the expectation in
\eqref{eqn:residual} is the weighted node sum and $\varepsilon'$ ranges over the finite node set
$\calE$ (Assumption~\ref{ass:regular}); the integration error of the rule relative to the
continuous-shock law $\nu$ is a separate, refinable approximation error, not part of the equilibrium
statements. The equilibrium is characterized by a structural residual, the equilibrium
conditions of the economic model written as deviations from zero and called structural because each
entry is one of the model's own optimality, complementarity, or market-clearing conditions rather
than a statistical error term,
\begin{equation}
\rwm\bigl(x,a,Q\bigr)=0,\qquad
Q=\Qpi(x):=\E\bigl[g\bigl(x',\nnp(x')\bigr)\,\big|\,x,a\bigr],
\label{eqn:residual}
\end{equation}
where $\rwm=(\rwm_1,\dots,\rwm_K)$ stacks the $K$ structural conditions, the Euler equations, the Karush--Kuhn--Tucker (KKT) complementarity conditions in Fischer--Burmeister form,\footnote{A complementarity condition is a set of inequalities: a nonnegative multiplier $\mu\ge0$, a nonnegative constraint slack $s\ge0$, and the complementary-slackness requirement $\mu s=0$ (at most one of the two is positive). Inequalities cannot be entered directly as equality residuals. The Fischer--Burmeister function $\phi(a,b)=a+b-\sqrt{a^2+b^2}$ collapses the three requirements into a single equality, since $\phi(a,b)=0$ holds if and only if $a\ge0$, $b\ge0$, and $ab=0$. Replacing each complementarity pair $(\mu,s)$ by the single equality residual $\phi(\mu,s)=0$ (continuous and globally Lipschitz, though non-differentiable at the origin; see Proposition~\ref{prop:fb}) therefore turns the KKT conditions into equalities that enter the structural residual $\rwm$ on the same footing as the Euler and market-clearing equations, so the entire equilibrium system reduces to a single residual driven to zero.} and the market-clearing conditions, into a single vector, each block $\rwm_i$ vanishing exactly at equilibrium. The continuation $\Qpi(x)$ is the conditional expectation of next period's marginal utilities, returns, prices, and multipliers the current decision depends on, and the kernel $g$ assembles the marginal-utility or market-clearing terms that enter $\rwm$.

An unsupervised deep-learning solver trains $\nnp$ with no external label, driving the exact
residual to zero on the states it visits. We follow the standard recipe for residual-learning
solvers, of which DEQN \citep{azinovicDEEPEQUILIBRIUMNETS2022} and the
deep-learning Euler-residual method of \citet{maliarwinant2021} are the canonical instances. It
has four components, which we state abstractly here and instantiate on Brock--Mirman
(Section~\ref{sub:worked}): a parametric policy approximation, a loss measuring violation of the
equilibrium conditions, an update rule, and the measure on which the loss is imposed. Only the last,
the choice of enforcement measure, is what EWM later changes; the other three it inherits unchanged.

\paragraph{(i) Function approximator.} Representing the policy is a computational choice, separate
from the economics. We approximate the equilibrium policy $\pi$ by a parametric family
$\nnp=\nnet_\theta$ with trainable parameters $\theta$, realized in our experiments by a feedforward
neural network $\nnet_\theta$, so that the action is $a=\nnp(x)=\nnet_\theta(x)$;\footnote{A densely
connected feedforward neural network with $K$ layers encodes the map
$\nnet_\theta(x)=\sigma_K\!\bigl(W_K\cdots\sigma_2(W_2\,\sigma_1(W_1 x+b_1)+b_2)\cdots+b_K\bigr)$, a
sequence of affine maps $z\mapsto W_i z+b_i$ with weight matrices $W_i$ and bias vectors $b_i$, each
followed by an elementwise activation $\sigma_i$; the parameter vector $\theta$ collects all weights
and biases. This densely connected feedforward form is the simplest such architecture; more
expressive ones (deeper, residual, convolutional, or attention-based networks) are of course
possible.} the approach is agnostic to the class,
and a polynomial or spline basis would serve equally. A neural network is a universal approximator
\citep{hornik1989multilayer} able to resolve the local, nonlinear features of equilibrium policies.

\paragraph{(ii) Loss.} The loss is the squared structural residual \eqref{eqn:residual} at a state,
with the continuation $\Qpi$ computed by quadrature on the exact $\Gamma$ at the next-period nodes
$x'=\Gamma(x,\nnp(x),\varepsilon')$. The per-state loss is the weighted sum of the squared blocks,
\begin{equation}
\mathcal{L}(x)\;=\;\bigl\|\rwm\bigl(x,\nnp(x),\Qpi(x)\bigr)\bigr\|^2
\;=\;\sum_{i=1}^{K} w_i\,\rwm_i\bigl(x,\nnp(x),\Qpi(x)\bigr)^2,
\label{eqn:Rsum}
\end{equation}
with fixed weights $w_i>0$ (unity unless stated, the case of the squared norm $\|\rwm\|^2$ used
throughout below), so that one scalar objective drives every equilibrium condition to zero at once.
There is no label and no target policy: the only training signal is the model's own equilibrium
condition.

\paragraph{(iii) Updating.} The parameters are updated by mini-batch stochastic gradient descent
(SGD). On a mini-batch $\mathcal{B}$ of sampled states they take one step down the gradient of the
mean squared residual,
\begin{equation}
  \theta \;\leftarrow\; \theta - \eta\,\nabla_\theta\,\frac{1}{\lvert\mathcal{B}\rvert}
  \sum_{x\in\mathcal{B}}\bigl\|\rwm\bigl(x,\nnp(x),\Qpi(x)\bigr)\bigr\|^2,
\label{eqn:sgd_update}
\end{equation}
with learning rate $\eta>0$; in practice we use the Adam variant \citep{kingma2015adam}, one
gradient step per simulated batch.

\paragraph{(iv) Sampling.} The states are drawn by simulating the economy forward under the current
policy, so each iteration draws a batch from the policy-induced invariant (ergodic) measure $\muerg$
and re-simulates: the sampling measure moves with the policy. The defining choice is where the
residual is imposed; at a fixed point the policy minimizes
\begin{equation}
\E_{x\sim\muerg}\bigl[\|\rwm(x,\nnp(x),\Qpi(x))\|^2\bigr],
\label{eqn:pathloss}
\end{equation}
the residual on the very measure the policy induces. Algorithm~\ref{alg:deqn} assembles the four
components: simulate a fresh batch, take one gradient step per state, and re-simulate under the
updated policy.

\begin{algorithm}[t]
\caption{Pathwise residual solver (the DEQN / deep-learning Euler-residual baseline), assembling the
four components of Section~\ref{sec:gap}. Algorithm~\ref{alg:ewm} is its equilibrium-world-model
counterpart, differing only in the sampling measure (line~\ref{algd:sample}) and an added world arm.}
\label{alg:deqn}
\begin{algorithmic}[1]
\Require exact transition $\Gamma$, residual $\rwm$, exact continuation $\Qpi(\cdot)$ via quadrature
on $\Gamma$; policy class $\Pi_m$.
\State \textbf{(i)~approximator:} initialize the policy network $\nnp$ at random \Comment{ab initio: no reference solution}
\For{episode $=1,2,\dots$}
  \State \textbf{(iv)~sample:} simulate $\Gamma$ under $\nnp$ to draw states $x\sim\muerg$ \label{algd:sample}
  \State \textbf{(ii)~loss:} compute $\Qpi(x)$ by quadrature on $\Gamma$ and the residual $\rwm(x,\nnp(x),\Qpi(x))$
  \State \textbf{(iii)~update:} mini-batch Adam step on $\theta$ descending the path loss \eqref{eqn:pathloss}
\EndFor
\State \Return $\nnp$; report the out-of-sample exact residual $R^{H}$ (defined in \eqref{eqn:RH}) on fresh simulated states
\end{algorithmic}
\end{algorithm}

The sampler feeds the approximator, whose policy determines the sampler: a closed loop whose fixed
points are configurations in which the policy is consistent on its own ergodic set and need not be
consistent anywhere else. This is the computational image of a self-confirming equilibrium
\citep[Ch.~6]{sargent1999conquest}: beliefs are confirmed by the data the agent's own behavior
generates, and off-path behavior is never tested. A pathwise solver imposes the residual only on its own path; states it does not visit are
simply left uncertified, neither tested nor confirmed. The
empirical signatures are a low on-path residual coexisting with a large off-ergodic residual (rare
regimes, generalized impulse-response functions, tails) and strong dependence on the random seed,
distinct seeds settling into distinct self-confirming basins; this is also the closed loop the
literature on these solvers recognizes as a source of training instability, since a large policy update
can shift the state distribution into sparsely sampled regions and inflate gradients
\citep{azinovicDEEPEQUILIBRIUMNETS2022,azinovicyangzemlicka2025sequencespace}. The standard remedies in
current practice (widening the initial sampling cloud, freezing intermediate iterates, multi-start
simulation, feasibility penalties) are ex-ante engineering fixes that treat the symptom without
addressing its structural cause, the endogenous simulator itself; EWM instead makes the enforcement
distribution an explicit object of the method. We call this the
coverage gap: the exact residual certifies the solution only on the support it is imposed on, and
on-path training makes that support the policy's own simulated experience.

\paragraph{Why the coverage gap cannot be dodged.} Three objections answer themselves once the use of the solution is kept in view. (i)~That the ergodic simulation rarely visits these states. But rare visits are no reason to leave them uncertified: counterfactuals, impulse responses, stress tests, and welfare integrals all query precisely these regions, so a regime of small but positive probability is an object of study, not a curiosity to be averaged away. (ii)~That the gap could be dissolved by enlarging the training box. It cannot: the reachable region is a thin, curved set carved out by the model's own dynamics, so a hypercube drawn around it has volume that explodes with dimension, spending almost all of its exact evaluations where the policy is never queried while still under-resolving the slender region that matters. EWM instead rolls the exact transition forward from stressed and rare seeds, concentrating training where the economics happens; a placebo control draws the same budget of states from a uniform box rather than the exact transition, isolating whether the gain comes from the structural dynamics or merely from sampling more states. Consistent with the volume argument, a box can still blanket a low-dimensional benchmark, so whether it substitutes is governed by the dimension. (iii)~That coverage is importance sampling under another
name. It is not: the off-path error of the residual decomposition below is a
distance between measures with no sample-size rate, so unlimited path data
leave it untouched, and the measure enters the definition of the computed
equilibrium itself; the argument is completed where coverage is defined
(Section~\ref{sub:coverage}).

\subsection{The equilibrium world model: a perceived law of motion for continuation objects}\label{sub:econ}

Relative to the pathwise residual loop of Algorithm~\ref{alg:deqn}, EWM changes two things. The first change is
where the residual is imposed: on a deliberately
enlarged coverage measure
\begin{equation}
\muK \;=\; \rho_1\,\muerg \;+\; \rho_2\,\mu^{\mathrm{stress}}_\kappa
\;+\; \rho_3\,\mu^{\mathrm{local}}_\kappa, \qquad\textstyle\sum_i\rho_i=1,
\label{eqn:muK}
\end{equation}
which mixes the on-policy ergodic component with a stress component that reaches the rare,
off-ergodic regimes the path almost never visits and a local component that fills admissible
neighborhoods of the path, every component generated by the exact transition $\Gamma$ rather than an
arbitrary hypercube (Figure~\ref{fig:coverage}). The reach $\kappa$ indexes how far off the ergodic
set the residual is imposed;\footnote{In plain terms, $\kappa$ is a dial for how far beyond the
model's typical (ergodic) behavior we insist the equilibrium conditions hold: at $\kappa=0$ only the
policy's own simulated path, and as $\kappa$ grows, ever rarer and more stressed states, so a larger
$\kappa$ enforces the same conditions on a wider set.} its schedule and the bridge it opens are
developed in Section~\ref{sub:coverage}. The second change is how the forward-looking continuation is
evaluated. In the pathwise loop (Algorithm~\ref{alg:deqn}) the continuation $\Qpi$, an expectation
over next period that itself evaluates the policy network at every quadrature node, is recomputed by
exact quadrature at every visited state, the dominant per-step cost. EWM may instead carry it with a
cheap, small neural-network surrogate $\wm(x)\approx\Qpi(x)$, fit to the exact continuation and
audited against it; being a smooth, differentiable function rather than a quadrature sum, the
surrogate also removes the per-step re-quadrature and supplies clean gradients, a numerical advantage
over the exact rule.

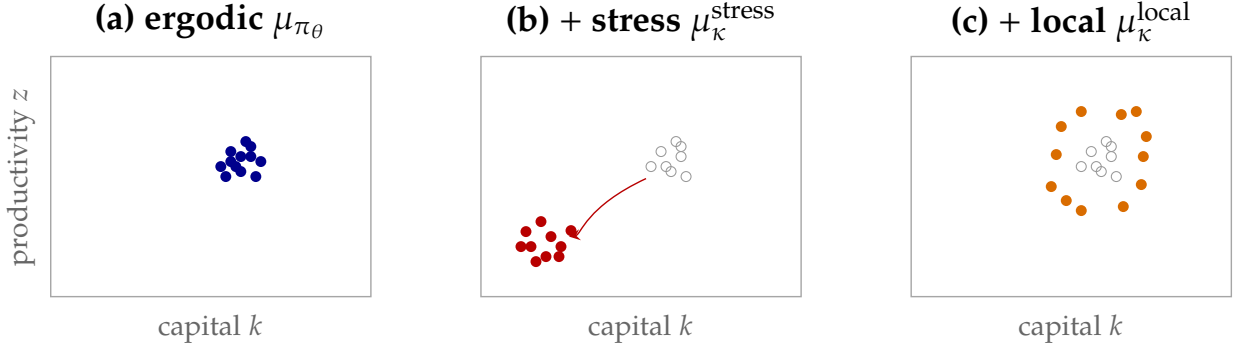
\begin{figure}[t]
\centering
\resizebox{\textwidth}{!}{
\begin{tikzpicture}[
  panel/.style={draw=black!35, thin},
  ergdot/.style={circle, fill=blue!55!black, inner sep=1.1pt},
  ergfaint/.style={circle, draw=black!40, fill=none, line width=0.3pt, inner sep=1.0pt},
  strdot/.style={circle, fill=red!72!black, inner sep=1.1pt},
  locdot/.style={circle, fill=orange!85!black, inner sep=1.1pt},
  lbl/.style={font=\small\bfseries},
  ax/.style={font=\scriptsize, black!60}]
\begin{scope}[xshift=0cm]
  \draw[panel] (0,0) rectangle (3.2,2.4);
  \node[lbl] at (1.6,2.75) {(a) ergodic $\muerg$};
  \node[ax] at (1.6,-0.32) {capital $k$};
  \node[ax,rotate=90] at (-0.30,1.2) {productivity $z$};
  \foreach \x/\y in {1.7/1.3,1.9/1.25,2.0/1.4,1.8/1.45,2.05/1.2,1.95/1.55,1.75/1.2,2.1/1.35,1.85/1.3,2.0/1.5,1.9/1.4,1.8/1.35}
    \node[ergdot] at (\x,\y) {};
\end{scope}
\begin{scope}[xshift=4.3cm]
  \draw[panel] (0,0) rectangle (3.2,2.4);
  \node[lbl] at (1.6,2.75) {(b) $+$ stress $\mu^{\mathrm{stress}}_\kappa$};
  \node[ax] at (1.6,-0.32) {capital $k$};
  \foreach \x/\y in {1.7/1.3,1.9/1.25,2.0/1.4,1.8/1.45,2.05/1.2,1.95/1.55,1.85/1.3,2.0/1.5}
    \node[ergfaint] at (\x,\y) {};
  \draw[->,>=stealth,red!72!black,thin] (1.65,1.18) .. controls (1.05,0.9) and (1.0,0.62) .. (0.88,0.6);
  \foreach \x/\y in {0.5/0.5,0.65/0.4,0.45/0.65,0.7/0.6,0.55/0.35,0.8/0.5,0.6/0.75,0.4/0.5,0.78/0.4,0.9/0.66}
    \node[strdot] at (\x,\y) {};
\end{scope}
\begin{scope}[xshift=8.6cm]
  \draw[panel] (0,0) rectangle (3.2,2.4);
  \node[lbl] at (1.6,2.75) {(c) $+$ local $\mu^{\mathrm{local}}_\kappa$};
  \node[ax] at (1.6,-0.32) {capital $k$};
  \foreach \x/\y in {1.7/1.3,1.9/1.25,2.0/1.4,1.8/1.45,2.05/1.2,1.95/1.55,1.85/1.3,2.0/1.5}
    \node[ergfaint] at (\x,\y) {};
  \foreach \x/\y in {1.4/1.1,1.5/1.7,2.3/1.12,2.35/1.6,1.45/1.42,2.32/1.4,1.7/0.86,2.12/0.9,1.7/1.85,2.1/1.82,1.55/0.96,2.25/1.85}
    \node[locdot] at (\x,\y) {};
\end{scope}
\end{tikzpicture}}
\caption{Stylized schematic of the coverage measure \eqref{eqn:muK},
$\muK=\rho_1\muerg+\rho_2\mu^{\mathrm{stress}}_\kappa+\rho_3\mu^{\mathrm{local}}_\kappa$, drawn for a
two-dimensional state $(k,z)$ of capital $k$ (horizontal axis) and productivity $z$ (vertical), as
in the Brock--Mirman economy with a rare disaster of Section~\ref{sub:bmdc}. (a)~The ergodic
component $\muerg$ is the policy's own simulated path, a thin on-path cloud at high capital. (b)~The
stress component places mass in the rare, off-ergodic region the path almost never visits, here
drawn-down capital under the disaster regime, reached by rolling stressed seeds forward through the
exact transition $\Gamma$ (arrow). (c)~The local component fills admissible neighborhoods of the
path. The exact residual is imposed on the union, so the solution is certified precisely where the
policy's own simulation places little or no weight. The point clouds are illustrative, not sampled
data.}
\label{fig:coverage}
\end{figure}

These two evaluations of the residual \eqref{eqn:residual}, distinguished only by how the
continuation $\Qpi$ inside it is computed, we name explicitly:
\begin{align}
\text{exact (high-fidelity):}\quad
   & R^{H}(x)\;:=\;\rwm\bigl(x,\nnp(x),\Qpi(x)\bigr), \label{eqn:RH}\\[2pt]
\text{surrogate (low-fidelity):}\quad
   & R^{L}(x)\;:=\;\rwm\bigl(x,\nnp(x),\wm(x)\bigr). \label{eqn:RL}
\end{align}
The exact residual $R^{H}$ evaluates the conditional expectation $\Qpi$ by exact quadrature on
the known transition $\Gamma$; the surrogate residual $R^{L}$ replaces $\Qpi$ by the learned
continuation $\wm(x)\approx\Qpi(x)$.

\paragraph{Audit convention.} The two residuals play different roles, which we fix once here: the
exact residual $R^{H}$ of \eqref{eqn:RH} is the reported equilibrium certificate, and the surrogate
residual $R^{L}$ of \eqref{eqn:RL} is an internal training object. ``Exact'' means evaluated with the
high-fidelity quadrature that defines the numerical model; any error relative to a continuous-shock model is a separate, refinable
approximation error. Throughout, only $R^{H}$ on held-out states, drawn independently of the
training batch and never used in any gradient step, is reported as an equilibrium error. The surrogate is therefore an accelerator of the exact
equilibrium conditions, never a substitute for them, and in our experiments also a stabilizer: it
reaches a verified solution from a larger share of random seeds than the surrogate-free coverage arm
(Section~\ref{sec:exp}).

EWM is therefore neither a DEQN, nor a world model in the sense of \citet{Ha2018WorldModels}, nor a parameterized-expectations method of the \citet{marcet:88}--\citet{denhaan2010} lineage: a DEQN is unsupervised with no continuation target, a Ha--Schmidhuber world model is supervised against simulation but carries no equilibrium fixed point, and a parameterized-expectations method learns a state-only continuation along the equilibrium path. EWM is the combination, an unsupervised equilibrium objective whose off-path continuation is supplied, and audited, by a supervised, structurally-anchored world model.

EWM trains the policy and the world model together
(Figure~\ref{fig:solver_vs_ewm}), alternating two
updates rather than minimizing a single joint objective: the policy arm
$\mathcal{L}_{\text{policy}}(\theta;\psi)$ is descended in $\theta$ with the world model held
fixed, and the world arm $\mathcal{L}_{\text{world}}(\psi;\theta)$ in $\psi$ with the policy
held fixed (the semicolon marks the held-fixed argument), both expectations taken over the
imagined coverage measure $\muK$:\footnote{The squared norm is a convention, not a theorem:
equilibrium is the zero set $\rwm=0$, so any loss with that zero set (mean absolute error, Huber)
identifies the same policy. We use it because it is smooth and renders the world arm
\eqref{eqn:Lworld} an $L^2(\muK)$ projection, the orthogonal structure the residual decomposition
(Proposition~\ref{prop:decomp}) and off-path bound (Theorem~\ref{thm:offpath}) exploit.}
\begin{align}
\mathcal{L}_{\text{policy}}(\theta;\psi)
&=\E_{x\sim\muK}\Big[\;
\underbrace{\big\|\,\rwm\big(x,\nnp(x),\,\wm(x)\big)\,\big\|^2}_{\substack{\text{squared surrogate residual }R^{L}\\ \text{(Euler}+\text{KKT/Fischer--Burmeister}+\text{market clearing)}}}
\;+\;\omega\;\underbrace{\calK\big(x,\nnp(x)\big)}_{\substack{\text{feasibility (admissibility);}\\ \omega\calK\equiv0\text{ in our runs}}}
\;\Big],
\label{eqn:Lpolicy}\\[2pt]
\mathcal{L}_{\text{world}}(\psi;\theta)
&=\E_{x\sim\muK}\Big[\;
\underbrace{\big\|\,\wm(x)-\Qpi(x)\,\big\|^2}_{\substack{\text{world model matches the exact}\\ \text{continuation }\Qpi\text{ computed on }\Gamma}}
\;\Big].
\label{eqn:Lworld}
\end{align}
The term $\omega\,\calK$ in \eqref{eqn:Lpolicy} is an optional admissibility (feasibility) penalty,
off in all our runs ($\omega=0$) because the softplus output heads enforce feasibility directly; it
is retained only so the objective also covers models without such heads.

\begin{figure}[ht!]
\centering
\resizebox{\textwidth}{!}{
\begin{tikzpicture}[
    deqnstep/.style={rectangle, draw=blue!55!black, thick, fill=blue!5,
        minimum width=4.7cm, minimum height=0.9cm, font=\small,
        rounded corners=3pt, align=center},
    ewmstep/.style={rectangle, draw=red!60!black, thick, fill=red!5,
        minimum width=4.7cm, minimum height=0.9cm, font=\small,
        rounded corners=3pt, align=center},
    arrB/.style={-{Stealth[length=2.5mm]}, thick, blue!55!black},
    arrR/.style={-{Stealth[length=2.5mm]}, thick, red!60!black},
    fbR/.style={-{Stealth[length=2.5mm]}, thick, red!60!black, dashed}
]
    \node[font=\small\bfseries, blue!55!black, align=center] at (-4.4,4.7)
      {Unsupervised residual solver};
    \node[deqnstep] (d1) at (-4.4,3.6)
      {States $x_i\sim\muerg$ \\ on-policy simulation};
    \node[deqnstep] (d2) at (-4.4,2.0)
      {Continuation $\Qpi(x_i)$ \\ exact quadrature on $\Gamma$};
    \node[deqnstep] (d3) at (-4.4,0.4)
      {Loss $\|\rwm(x_i,\nnp(x_i),\Qpi)\|^2$ \\ residual on $\theta$};
    \node[deqnstep] (d4) at (-4.4,-1.2)
      {SGD on $\theta\to\theta^\star$};
    \draw[arrB] (d1) -- (d2);
    \draw[arrB] (d2) -- (d3);
    \draw[arrB] (d3) -- (d4);

    \node[font=\small\bfseries, red!60!black] at (4.4,4.7) {EWM};
    \node[ewmstep] (e1) at (4.4,3.6)
      {States $x_i\sim\muK$ \\ coverage mixture};
    \node[ewmstep] (e2) at (4.4,2.0)
      {Exact continuation $\Qpi(x_i)$ \\ quadrature on $\Gamma$, once per episode};
    \node[ewmstep] (e3) at (4.4,0.4)
      {World arm: fit surrogate \\ $\psi\leftarrow\min_\psi\|\wm(x_i)-\Qpi(x_i)\|^2$};
    \node[ewmstep] (e4) at (4.4,-1.2)
      {Policy arm: drive residual \\ $\theta\leftarrow\min_\theta\|\rwm(x_i,\nnp(x_i),\wm(x_i))\|^2$};
    \draw[arrR] (e1) -- (e2);
    \draw[arrR] (e2) -- (e3);
    \draw[arrR] (e3) -- (e4);
    \draw[fbR] (e4.east) .. controls +(1.7,0) and +(1.7,0) .. (e1.east)
      node[midway, right, font=\scriptsize, red!60!black, align=left]
      {updated $\theta$:\\ re-simulate $\muK$};
\end{tikzpicture}}
\caption{Per-iteration structure of an unsupervised residual solver (left) versus EWM
(right). The pathwise solver \citep{maliarwinant2021,azinovicDEEPEQUILIBRIUMNETS2022}
samples its own simulated path $\muerg$ and recomputes the exact continuation $\Qpi$ by
quadrature on $\Gamma$ at every visited state, with the policy $\theta$ the only trainable
object. EWM samples from a coverage measure $\muK$ reaching the rare and stressed regions,
computes the exact continuation $\Qpi$ once per episode; the two arms then run in sequence, the
world arm fitting the surrogate $\wm$ to $\Qpi$ and the policy arm driving the residual with
$\wm$, rather than as a single joint update. The dashed arrow marks the updated policy re-entering
the next episode's coverage simulation.}
\label{fig:solver_vs_ewm}
\end{figure}

Read economically, the world model is a perceived law of motion for continuation objects:
it does not condition on every household, firm, and price history but summarizes the environment by
a compact predictive state, and in a macro economy the decision-relevant predictions are exactly
the continuation objects, future marginal utilities, returns, prices, and multipliers. Like
any learned object, the surrogate $\wm$ is reliable only on the states it was fit to and untested
elsewhere; so the coverage states it is trained on, and on which the residual is imposed, are what
decide whether the solution is certified only on the policy's own path or also off it. When the
aggregate state is itself a distribution, the learned object is a distributional encoder $h_\psi$
rather than a state-only surrogate (Section~\ref{sub:variants}).

\begin{definition}[Equilibrium world model]\label{def:ewm}
An equilibrium world model consists of (1) an exact, known structural transition
$\Gamma$; (2) an exact equilibrium residual $R$; (3) a coverage measure $\muK$ generated by
$\Gamma$; (4) a learned world object; and (5) an audit discipline that reports only the
exact held-out residual. The learned world object is one of: a continuation surrogate $\wm(x)$,
an action-conditioned continuation $\widehat W(x,a)$, or a distributional encoder $h_\psi$.
\end{definition}

\noindent This invariant core is common to all three world objects; the object itself is the one
component each economic model switches on. Section~\ref{sub:components} catalogues the components and
what each keeps fixed, Section~\ref{sub:variants} the three world objects, and the development that
follows specializes to the state-only continuation $\wm(x)$.

\subsection{Components and invariants}\label{sub:components}

The tuple $\mathrm{EWM}_\kappa$ above collects the six components of an EWM; we
now state which are held fixed and which vary across economic models. Three are fixed by the
economic model itself, the problem specification that any residual-based solver of the kind in
Algorithm~\ref{alg:deqn} already inherits: the policy $\nnp$, the exact transition $\Gamma$, and
the exact residual $\rwm$. The other three, the coverage measure $\muK$, the world component, and
the exact-residual audit, are what an EWM adds, and they form a modular reliability architecture
rather than a single device.

\paragraph{What never changes in an EWM.} Whatever the economic model, four things are invariant in
an EWM as we formulate it, and together
they are what keeps imagining off-path from drifting into a self-confirming fiction: (i) the
transition $\Gamma$ is never learned, every imagined next state being the exact
$\Gamma$ of \eqref{eqn:transition} at the learned action $a=\nnp(x)$, so the world model imagines
the value of the continuation, never the dynamics that generate it; (ii) equilibrium is the
exact residual $\rwm=0$, not an arbitrary reward, and coverage and monitoring objects never enter
it or supply policy labels; (iii) the world component never defines equilibrium, entering
only the continuation slot of that residual; and (iv) accuracy is the exact held-out
residual $R^{H}$, with the surrogate audited against exact quadrature on sensitive states and
fallback to exact wherever they disagree.

\paragraph{What varies across models.} The coverage measure $\muK$ and the exact-residual audit
are an invariant core present in every application; the world component takes whichever form the
model calls for, a continuation surrogate $\wm(x)$, its action-conditioned extension
$\widehat W(x,a)$, or a distributional encoder $h_\psi$ (Section~\ref{sub:variants}).
Table~\ref{tab:components} gives, for each of these five, the problem it solves and when it is
needed; that coverage, not the surrogate, is decisive is a statement about off-path accuracy on the
finite-dimensional models, not that the other components are dispensable.

\begin{table}[t!]\centering\small
\begin{tabular}{p{0.21\textwidth}p{0.40\textwidth}p{0.30\textwidth}}
\toprule
\textbf{Component} & \textbf{Problem it solves} & \textbf{When it is needed} \\
\midrule
Coverage measure $\muK$
 & closes the self-confirming gap; the decisive lever for off-path accuracy
 & whenever off-path or counterfactual regions matter \\[4pt]
Exact-residual audit ($R^{H}$ on held-out states)
 & certification; the operational distance-to-rational-expectations metric
 & always on; it is what makes any accuracy claim verifiable \\[4pt]
Continuation surrogate $\wm(x)$
 & conditioning, stabilization, and amortization of the off-path expectation
 & used throughout; becomes necessary where the exact residual is too stiff to descend or the integral too costly to repeat \\[4pt]
Action-conditioned continuation $\widehat W(x,a)$
 & identification when an action moves the transition measure, not the next state
 & endogenous-protection and other measure-moving-action settings \\[4pt]
Distributional encoder $h_\psi$ + anti-collapse regularizer (SIGReg)
 & perception for distribution-valued states; a stable low-dimensional population view
 & when the aggregate state is a cross-sectional distribution \\
\bottomrule
\end{tabular}
\caption{The components of an EWM. The columns are: Component, the EWM ingredient; Problem it
solves, the role that ingredient plays; and When it is needed, the settings in which it is switched
on. The first two rows are the invariant core, present in every application; the last three are
switched on by the economic model. SIGReg, in the final row, is an anti-collapse regularizer
that keeps the encoder from mapping every population to the same uninformative embedding, by
spreading the embeddings to fill an isotropic Gaussian (Appendix~\ref{app:encoder_what}).}
\label{tab:components}
\end{table}

\subsection{Coverage: imagination off the manifold, the \texorpdfstring{$\SCE\!\to\!\RE$}{SCE-to-RE} sieve}\label{sub:coverage}

Section~\ref{sub:econ} introduced the coverage measure $\muK$ as the first of the two changes;
this section develops its reach $\kappa$, the schedule that enlarges it, and the bridge it opens
from a self-confirming equilibrium toward rational expectations. Recall that $\muK$ of
\eqref{eqn:muK} replaces the ergodic training measure $\muerg$ in \eqref{eqn:pathloss} by the
economic counterpart of the set of states a world model rehearses on,
mixing the on-policy ergodic component with a stress component that imagines rare, off-ergodic
regimes (a disaster shock, or a distribution pushed forward under an adverse regime) and a local
component that perturbs visited states within the admissible domain (Figure~\ref{fig:coverage}).
Every component is generated by the exact transition $\Gamma$ (the distributional case uses
$\Phi$), never an arbitrary hypercube; the seed distributions, mixing weights $\rho_i$,
perturbation radii, disaster depth, and the roll horizon $H$ (the number of exact-transition steps a
stressed seed is pushed forward) are specified concretely for each model
below.\footnote{Across the models below, $H$ is typically $3$ or $5$; it is a coverage design
parameter set by experimentation, not tuned against the held-out residual $R^H$.}
The index $\kappa$ is the reach of
imagination, the radius and mass of off-ergodic support on which the exact
residual is enforced. Imposing the exact residual on this measure, while keeping exact quadrature for the
continuation, isolates the first change on its own,
\begin{equation}
\E_{x\sim\muK}\bigl[\|\rwm(x,\nnp(x),\Qpi(x))\|^2 + \omega\,\calK(x,\nnp(x))\bigr],
\label{eqn:cov_policy_loss}
\end{equation}
which is the baseline objective \eqref{eqn:pathloss} with the ergodic measure $\muerg$ replaced
by the coverage measure $\muK$ and nothing else touched. Because it still evaluates the exact
continuation $\Qpi$, it is not yet an EWM: it makes only the first change, the enlarged
measure, and so serves as the control that isolates what coverage alone buys; the equilibrium
world model's training loss \eqref{eqn:Lpolicy} is this same objective with $\Qpi$ replaced by the
learned surrogate $\wm$, the second change. As in \eqref{eqn:Lpolicy}, the admissibility penalty is
off in our runs ($\omega\calK\equiv0$), so \eqref{eqn:cov_policy_loss} reduces to the mean squared
exact residual. The solver runs ab initio: the first coverage stage starts from random
initialization and later stages warm-start from the previous EWM stage, requiring no known,
steady-state, or reference solution; prior information can be folded in when available, but is never
needed.

At $\kappa=0$, where $\muK=\muerg$,
\eqref{eqn:cov_policy_loss} is exactly the self-confirming objective \eqref{eqn:pathloss}, the
world model imagining only on-path; as $\kappa$ grows, $\muK$ fills the region the economy can
reach and the residual comes to hold everywhere the economy can go, not merely where the policy
happens to go, which is rational expectations. Imagination is thus the sieve that carries the
solver from a self-confirming equilibrium toward rational expectations, every imagined state
anchored to the exact transition $\Gamma$ and judged by the exact residual $R^{H}$;
Section~\ref{sec:theory} makes the bridge precise.

\paragraph{What $\kappa$ is, concretely.} The index $\kappa$ is not an abstract
parameter; it is the stage of a schedule that literally enlarges the off-ergodic
part of $\muK$. Increasing $\kappa$ means turning specific, namable dials: raising
the mass $\rho_2,\rho_3$ on the off-path components of \eqref{eqn:muK}; widening
the perturbation radius of the local component; raising the frequency and depth of
the forced rare regime in the stress component; and rolling stressed
cross-sections further off the manifold, where rolling forward means iterating the exact
transition $x'=\Gamma(x,\nnp(x),\varepsilon')$ for $H$ steps under the current policy, never a
learned simulator. The stages this produces, from the ergodic path through a local neighborhood and the rare regime
to deep off-manifold cross-sections, are the bridge of Figure~\ref{fig:bridge}.

\paragraph{Is the coverage measure not itself self-confirming?} Since coverage states are
generated by rolling stressed seeds forward, $\muK$ depends on $\nnp$, raising the worry that the
self-confirming loop has merely been displaced. It has not: the seeds are placed off the ergodic
set by construction (capital displaced, the disaster forced), not by the policy's own dynamics,
and the forward map is the exact $\Gamma$, never a learned simulator
(Section~\ref{sub:discipline}), so the only policy-dependence is through the equilibrium action,
which the exact residual drives toward the correct value at every imagined state. We prove below
that this is a well-defined population fixed point, not a displaced self-confirming loop; whether
the nonconvex algorithm reaches it is left to the experiments. The one remaining risk, a coverage
hole the optimal policy reaches but the seeds never do, cannot hide either: it surfaces as a high
residual on the frozen, policy-independent held-out set.

\begin{figure}[ht!]
\centering
\resizebox{\textwidth}{!}{
\begin{tikzpicture}[
    box/.style={draw, rounded corners=2pt, minimum width=2.5cm,
      minimum height=1.45cm, align=center, font=\scriptsize, inner sep=2pt,
      line width=0.6pt},
    scebox/.style={box, fill=BrickRed!10, draw=BrickRed!70},
    midbox/.style={box, fill=Goldenrod!18, draw=Goldenrod!70!black},
    rebox/.style={box, fill=MidnightBlue!12, draw=MidnightBlue!70},
    arrow/.style={-{Stealth}, line width=0.7pt, shorten >=2pt, shorten <=2pt}
]
  \node[font=\scriptsize\itshape, align=center] at (8.0, 4.05)
    {\textsc{held fixed:} structural model $\Gamma$, residual $\rwm$, equilibrium
     definition $\rwm=0$, exact transition.};
  \draw[-{Stealth}, line width=0.5pt, gray!55!black] (0.4, 3.45) -- (15.6, 3.45)
    node[midway, above=2pt, font=\scriptsize]
    {coverage reach $\kappa$: off-ergodic coverage support grows
     (on-path at left $\to$ full reachable support at right)};
  \node[font=\scriptsize, text=gray!45!black] at (3.2, 2.7) {$\muK$ widens};
  \node[font=\scriptsize, text=gray!45!black] at (6.4, 2.7) {$\muK$ widens};
  \node[font=\scriptsize, text=gray!45!black] at (9.6, 2.7) {$\muK$ widens};
  \node[font=\scriptsize, text=gray!45!black] at (12.8, 2.7) {$\muK\!\to$ full supp.};
  \node[scebox] (B1) at (1.6, 1.0)
    {\textbf{SCE}\\$\kappa=0$\\on-path only:\\cover where you go};
  \node[midbox] (B2) at (4.8, 1.0)
    {$\kappa=1$\\$+$ local\\neighborhood};
  \node[midbox] (B3) at (8.0, 1.0)
    {$\kappa=2$\\$+$ rare /\\disaster regime};
  \node[midbox] (B4) at (11.2, 1.0)
    {$\kappa=K$\\$+$ deep off-manifold\\cross-sections};
  \node[rebox] (B5) at (14.4, 1.0)
    {\textbf{RE}\\$\kappa\to\infty$\\imagine on full\\reachable support};
  \draw[arrow] (B1) -- (B2);
  \draw[arrow] (B2) -- (B3);
  \draw[arrow] (B3) -- (B4);
  \draw[arrow] (B4) -- (B5);
  \draw[decorate, decoration={brace, amplitude=4pt, mirror}, gray!60]
    (0.4, -0.05) -- (12.4, -0.05)
    node[midway, below=3pt, font=\scriptsize, text=gray!45!black]
      {Theorem~\ref{thm:br}: coverage-confirmed fixed point on
       $\muK$ at finite $\kappa$};
  \draw[decorate, decoration={brace, amplitude=4pt, mirror}, gray!60]
    (12.6, -0.05) -- (15.6, -0.05)
    node[midway, below=3pt, font=\scriptsize, text=gray!45!black]
      {Theorem~\ref{thm:consist}: RE};
  \node[font=\scriptsize\itshape, align=center, text=gray!50!black]
    at (8.0, -1.35)
    {\textsc{progress diagnostic:} off-equilibrium residual
     shrinks as the imagined support enlarges ($\kappa\to\infty$).};
\end{tikzpicture}
}
\caption{The $\kappa$-homotopy as a bridge from a self-confirming equilibrium to
rational expectations on the coverage axis. Each stage enlarges the imagined support
$\muK$, from the ergodic path to a neighborhood, the rare regime, and post-shock
cross-sections, imposing the same exact residual on strictly more of the reachable
set. At every finite $\kappa$ the fixed point is a coverage-confirmed fixed point on
$\muK$, self-confirming in the strict sense only at $\kappa{=}0$; as the
support fills the region the coverage schedule certifies, accumulation points satisfy the exact
residual everywhere that coverage reaches, its own path together with the audited off-ergodic
regions; filling the entire structurally reachable set is the idealized limit of enriching the
coverage design. The theory below establishes both ends of this bridge.}
\label{fig:bridge}
\end{figure}

\paragraph{The \texorpdfstring{$\kappa$}{kappa}-equilibrium world model and the homotopy.}\label{sub:kewm}
The bridge of Figure~\ref{fig:bridge} is realized by a specific object, which we now define by
specializing the EWM of Definition~\ref{def:ewm} to the finite-dimensional continuation instance the
theory analyzes.

\begin{definition}[Finite-dimensional continuation EWM ($\kappa$-equilibrium world model)]\label{def:kewm}
Fix a perception class $\calF_m$ for the world model and a policy class
$\Pi_m$, both of capacity index $m$, and a coverage measure $\muK$. A pair
$(\nnp,\wm)\in\Pi_m\times\calF_m$ is a $\kappa$-equilibrium world
model if it jointly solves the surrogate-continuation policy arm
\eqref{eqn:Lpolicy} in $\theta$ and the world arm \eqref{eqn:Lworld} in $\psi$.
\end{definition}

\noindent The definition is stated with the surrogate continuation $\wm$ in the policy
residual, the object the fixed-point and consistency arguments below actually analyze, so
that the world model genuinely enters the equilibrium. The exact-continuation coverage loss
\eqref{eqn:cov_policy_loss}, in which $\wm$ is replaced by the exact $\Qpi$, is the
surrogate-free control we use to isolate what the coverage measure alone contributes (Section~\ref{sub:protocol}): a diagnostic
special case, not the definition.

\noindent The surrogate $\wm$ can be any approximator of the continuation; in practice a small
feed-forward network with two hidden layers and $\tanh$ activations, chosen for simplicity rather
than required by the method. We analyze the state-only continuation $\wm(x)$ throughout; the
action-conditioned variant $\widehat W(x,a)$ of Section~\ref{sub:variants} carries none of the
present accuracy claims. The theory below is
stated not for this particular network but for an idealized convex class $\calF_m$ of growing
capacity \citep{cybenko1989approximation,hornik1989multilayer}, which is what lets the projection
and fixed-point arguments invoke a convexity the network itself lacks (Section~\ref{sec:theory}).
Two distinct levers are at work, and we keep them separate throughout: the coverage reach, how far
off its own path the world model is made to imagine (the radius and mass of off-ergodic support in
$\muK$), and the perception capacity, the width of $\calF_m$. On a simple model a small capacity
already spans the continuation, so coverage reach is the binding lever; the experiments isolate
each.

\noindent A $\kappa$-equilibrium world model need not satisfy the equilibrium conditions
exactly: its policy attains the smallest residual achievable in $\Pi_m$ on $\muK$, generically
positive. For $\kappa>0$ we call such a pair a coverage-confirmed fixed point, the
projection-best-response structure imposed on the modeler's chosen $\muK$ rather than on data the
agent's own behavior generates, so it is a self-confirming equilibrium in the strict sense only at
$\kappa{=}0$. Such a pair exists at each $\kappa$ under a single-valued policy best response, as
we establish below.

\noindent The solver realizes Definition~\ref{def:kewm} through the $\kappa$-homotopy of
Figure~\ref{fig:bridge}, a sequence of stages of increasing coverage reach $\kappa$ (and, where
it binds, perception capacity $m$), each warm-started from the last; Section~\ref{sec:theory}
proves it a well-defined bridge from a self-confirming equilibrium to rational expectations.

\paragraph{Coverage is not a sampling method.} The general-interest objection
writes itself: is coverage not importance sampling under another name,
oversampling rare states to estimate the same object more efficiently? It is
not, and the distinction is the difference between an estimator and an
estimand. A sampling design improves the rate at which a finite-sample
objective converges to a fixed population objective; coverage changes the
population objective itself. Formally, the coverage term of
Proposition~\ref{prop:decomp} is a total-variation distance between measures,
not a variance: it carries no sample-size rate, so a pathwise solver with
unlimited draws from its own path retains the full off-path error of one with
a thousand. Reweighting the path's own draws is at best a rare-event-simulation
implementation of a different training measure where the path's support
overlaps it, exponentially expensive exactly in the rare regions at issue, and
unavailable on the counterfactual cross-sections and post-shock populations
the path never generates at all. Nor is there an estimation problem left for a
cleverer sampler to improve: at a pathwise fixed point the on-path objective
is already driven to numerical tolerance, and what is missing is not precision
but constraints, the equilibrium conditions at states the path never visits,
which is what enlarging $\muK$ imposes. The measure accordingly enters the
definition of the computed object, each reach $\kappa$ indexing a distinct
coverage-confirmed fixed point (made precise in Section~\ref{sec:theory}), so
moving $\kappa$ moves the equilibrium concept rather
than sharpening an estimate of one, exactly as requiring optimality off the
equilibrium path is constitutive of a game-theoretic solution concept; asking
whether that is ``just sampling'' is asking whether subgame perfection is just
Nash equilibrium with more histories. The high-dimensional experiments of Section~\ref{sec:exp}
confirm that the gain is neither the point count nor the in-loop compute.

\paragraph{Relation to active and adaptive sampling.} Choosing where to evaluate invites comparison
with Bayesian active learning and adaptive design
\citep{mackay1992information,snoek2012practical,rennerscheidegger_2018}, but coverage is not active
learning: it governs the support of an unsupervised fixed-point constraint whose solution is
endogenous to the policy, not the sampling efficiency of a fixed, exogenous target, so each $\kappa$
indexes a distinct equilibrium concept rather than a better approximation of one. The model's own
uncertainty is in fact an unreliable acquisition signal here, smallest precisely in the rare regimes
where the residual is largest; active learning reappears only one level down, in choosing where to
spend the exact quadrature that trains $\wm$.

\subsection{Training an equilibrium world model}\label{sub:discipline}

Training an EWM runs the same four components as the pathwise solver of
Algorithm~\ref{alg:deqn}, changing the sampling measure in component~(iv) and adding a second
approximator. The function approximator~(i) now carries two networks, the policy $\nnp$ and the
continuation surrogate $\wm$; the loss~(ii) splits into a policy arm and a world arm; the
update~(iii) is a mini-batch gradient step on each; and the sampling~(iv) draws from the coverage
measure $\muK$ in place of the ergodic measure $\muerg$. Concretely, Algorithm~\ref{alg:ewm} runs an
outer loop over the coverage reach $\kappa$, each stage warm-started from the last, and within each
episode takes three steps:
\begin{enumerate}\setlength{\itemsep}{1pt}
\item \textbf{Generate coverage states (sampling).} Simulate $\Gamma$ under $\nnp$, carrying the
on-path batch forward from the previous episode's end state rather than resetting it to a fixed
point, then turn the visited states into off-path imagined states, seeding stressed and rare regions
and locally perturbing the path and rolling forward through the exact $\Gamma$, to draw a batch
$x\sim\muK$.
\item \textbf{Fit the world component (world arm).} Take $E$ mini-batch gradient steps on the world arm
\eqref{eqn:Lworld}, fitting $\wm$ to the exact continuation $\Qpi$ computed once per episode
(line~\ref{alg:Q}).
\item \textbf{Update the policy (policy arm).} Take $E$ mini-batch gradient steps on the policy arm
\eqref{eqn:Lpolicy}, driving the structural residual with the updated $\wm$.
\end{enumerate}
Every state fed to the residual is structurally generated by $\Gamma$ from an admissible seed or
perturbation, never an arbitrary hypercube point;\footnote{For a
practitioner building $\muK$ in a new model, the components of \eqref{eqn:muK} follow a short list
of design rules. (i) Reachability first: every coverage state is generated by the exact transition
$\Gamma$ (or $\Phi$) from an admissible seed, never drawn from a hypercube, for the reasons given
in Section~\ref{sub:coverage}; a box placebo isolates whether box sampling can substitute.
(ii) The cheapest local enlargement is to re-draw the one-step shock, evaluating
$x'=\Gamma(x_t,\nnp(x_t),\varepsilon)$ at alternative innovations $\varepsilon$, which is reachable
by construction; perturbing visited states in economically meaningful coordinates (capital,
productivity, distribution summaries) by a small fraction of the ergodic spread serves the same
purpose. (iii) Stress seeds should trace exactly the post-shock paths that counterfactuals, impulse
responses, and stress tests will later query, and should be placed so that occasionally binding
constraints, a borrowing limit, irreversibility, a zero lower bound, actually bind on a non-trivial
fraction of the pool, since the kink is where a smooth on-path policy extrapolates worst. (iv) When
the state is a distribution, perturb the compact summary (moments, or the learned embedding
$h_\psi(\mu)$) and push the implied population forward by the exact $\Phi$, rather than perturbing
individual agents. (v) Surprise scores are diagnostics, not generators: an ergodic-fit surprise
score locates under-covered regions after the fact, but model-internal
uncertainty is smallest precisely where a self-confirming solution is silently uncertified (see the
discussion of active and adaptive sampling in Section~\ref{sub:coverage}), so structural seeding through $\Gamma$ must
remain the generator. The mixing weights $\rho_i$, radii, depths, and horizon $H$ are then fixed by
the coverage criterion of Appendices~\ref{app:bmdc_spec} and~\ref{app:irbc_cov_spec}, enough mass
that the rare regions are sampled at all, and never tuned against the reported held-out $R^{H}$.} and
every reported error is the exact residual $R^{H}$, evaluated by full high-fidelity quadrature on
states held out from training (drawn independently of the gradient batches).
We advance through the coverage schedule by this exact audit, not by a fixed episode budget: each
stage is trained to verified stationarity at its current coverage reach $\kappa$ (and, where it
binds, perception capacity), its held-out $R^{H}$ is recorded, and only then is the reach enlarged
for the next stage, warm-started from the converged one. The off-path accuracy at every stage is
thus a measured quantity, not a hidden training artifact.

\begin{algorithm}[t]
\caption{EWM training, the same four components as the pathwise solver of
Algorithm~\ref{alg:deqn} mirrored step for step. The only changes are in component~(iv), where states
are drawn from the coverage measure $\muK$ (line~\ref{alg:cov}), and the added world arm that fits
the surrogate (line~\ref{alg:world}). Computing $\Qpi$ once per episode (line~\ref{alg:Q}) and
reusing it across the $E$ inner policy steps is the surrogate's compute saving.}
\label{alg:ewm}
\begin{algorithmic}[1]
\Require exact transition $\Gamma$, residual $\rwm$, exact continuation $\Qpi(\cdot)$ via
quadrature on $\Gamma$; classes $\Pi_m,\calF_m$; coverage schedule
$\kappa_1<\dots<\kappa_S$ (a single stage $S{=}1$ suffices for simple models); mixing
weights $\rho$, rollout horizon $H$, inner steps $E$.
\State \textbf{(i)~approximators:} initialize policy $\nnp$ and surrogate $\wm$ at random \Comment{ab initio: no external/reference warm start, no known solution}
\For{stage $s=1,\dots,S$ \textbf{(reach $\kappa_s$)}} \Comment{warm-start from stage $s{-}1$}
  \For{episode $=1,2,\dots$}
    \State \textbf{(iv)~sample (coverage):} simulate $\Gamma$ under $\nnp$ for on-path states $X_{\mathrm{erg}}\sim\muerg$, carried forward across episodes (never reset);
    \State \hspace{1.2em} seed stressed states (capital off the ergodic level; where a disaster regime exists, set $d{:=}1$), \label{alg:cov}
    \Statex \hspace{1.2em} roll $H$ steps through $\Gamma$; add local perturbations of $X_{\mathrm{erg}}$, clip to feasible;
    \Statex \hspace{1.2em} mix $\muK=\rho_1\muerg+\rho_2\mu^{\mathrm{stress}}_\kappa+\rho_3\mu^{\mathrm{local}}_\kappa$
    \State \textbf{(ii)~target:} draw a batch $x\sim\muK$; compute the exact continuation $\Qpi(x)$ by quadrature on $\Gamma$ \Comment{once per episode} \label{alg:Q}
    \State \textbf{(iii)~world arm} ($E$ SGD steps): $\psi \leftarrow \psi-\nabla_\psi\,\E_{\muK}\|\wm(x)-\sg\,\Qpi(x)\|^2$ \label{alg:world}
    \State \textbf{(iii)~policy arm} ($E$ SGD steps): $\theta \leftarrow \theta-\nabla_\theta\,\E_{\muK}\big[\|\rwm(x,\nnp(x),\wm(x))\|^2+\omega\,\calK(x,\nnp(x))\big]$ \label{alg:pol}
    \State \textbf{(optional) route:} where $\|\wm(x)-\Qpi(x)\|$ exceeds a threshold, substitute $\Qpi(x)$ for $\wm(x)$ in line~\ref{alg:pol} \label{alg:route}
  \EndFor
  \State evaluate held-out exact residual $R^{H}$ and verified stationarity
\EndFor
\State \Return $\nnp$; report $R^{H}$ on held-out states \Comment{the surrogate $R^L$ is never a reported number}
\end{algorithmic}
\end{algorithm}

In the joint limit of full coverage reach and zero world-arm loss, the policy arm is the exact
residual on all reachable states, which is rational expectations
(Theorems~\ref{thm:consist},~\ref{thm:offpath}).

\paragraph{The audited continuation surrogate.}\label{sub:surrogate}

Where the exact residual is too stiff to descend on the coverage region, as in the
disaster-conditioned Euler residual of Section~\ref{sec:exp}, the learned continuation
$\wm(x)\approx\Qpi(x)$ supplies a smooth, bounded target: $\Qpi$ is a live target, a sum over
sharply varying branches (disaster persistence versus recovery), each a fresh evaluation of the very
policy being trained, so differentiating through it from a random start is badly conditioned. The
world arm \eqref{eqn:Lworld} fits $\wm$ to a detached (stop-gradient) high-fidelity continuation
$Q^{H}$ computed on a sparse anchor subsample of the coverage batch; evaluating $Q^{H}$ only on those
anchors rather than at every state in every gradient step is the surrogate's compute saving
(Algorithm~\ref{alg:ewm}, Figure~\ref{fig:amortization}; the stop-gradient and Polyak-target
conventions are detailed in Appendix~\ref{app:arch}), and $\wm$ then carries the continuation to the
bulk of $\muK$. The surrogate never defines the equilibrium: unlike a parameterized expectation
(Section~\ref{sec:related}) it is audited against exact quadrature on sensitive states, with fallback
to exact wherever $R^{L}$ and $R^{H}$ disagree and the gap $R^{H}-R^{L}$ reported, and every reported
equilibrium error is the exact $R^{H}$, never $R^{L}$; it therefore stabilizes the solution of a
fixed equilibrium rather than redefining it. Its value is conditioning, not raw speed: when the exact
expectation is cheap the surrogate and pathwise arms run at essentially equal exact-evaluation
budget, the saving accruing only in higher dimensions. Whether it is decisive is empirical: the
surrogate-free control that keeps $R^{H}$ but imposes the exact residual directly on $\muK$ is already
the single largest accuracy gain on the rare-disaster model (Section~\ref{sub:bmdc}), while where
that residual is too stiff to descend the smooth surrogate conditions it into reliable convergence.

\paragraph{A disciplined dreamer.}\label{sub:whatis}
EWM is, in the end, a disciplined dreamer. Like a world-model agent in
artificial intelligence \citep{Ha2018WorldModels,hafner2023dreamerv3}, it improves its policy by
rehearsing states its own realized path never visits; unlike a generic world model, the transition
is the economy's exact, known law and is never learned, the imagined states are generated by the
model's own equations, and each is judged by the exact equilibrium residual rather than a learned
reward. It is a world model in the Dreamer sense only under a Lucas-critical adaptation: it imagines
off-path and learns a continuation, but never learns $\Gamma$, scores $\rwm=0$ rather than a reward,
and imagines a single period since $\Qpi$ already encodes the forward sum (Appendix~\ref{app:hyper},
Table~\ref{tab:wm_vs_ewm}).

\subsection{A worked example: Brock--Mirman under DEQN and EWM}\label{sub:worked}

In this section, we use the Brock--Mirman growth model as a microscope. It is simple enough that every building
block of an equilibrium world model, the state and its residual, the coverage measure, and the
continuation with its audit, can be read off by hand, so the role of each is visible in isolation.
We proceed in two stages. Section~\ref{sub:bm_basic} introduces the basic model and uses it to
illustrate the EWM construction qualitatively: its equilibrium is smooth and a
standard residual solver already drives it to machine precision, so the structural objects can be
fixed and the coverage construction built explicitly, in a case where it is deliberately inert. Section~\ref{sub:bmdc} then adds a rare disaster, the off-ergodic regime an
equilibrium world model is built for, where coverage stops being inert and the two methods separate.

\subsubsection{The basic model}\label{sub:bm_basic}

A representative agent with log utility chooses next-period capital
$k'$ to maximize $\E\sum_t\beta^t\log c_t$ subject to $c+k'=zk^{\alpha}$, with productivity
following $\log z'=\rho\log z+\varepsilon'$. Under the two assumptions of this basic version, log
utility and full capital depreciation, the equilibrium policy is available in closed form,
$\nnp^{\star}(k,z)=\alpha\beta\,zk^{\alpha}$, which is what makes it a clean laboratory; the disaster
variant of Section~\ref{sub:bmdc}, with partial depreciation, has no closed form. The state
is $x=(k,z)$ and the choice is next-period capital $a=k'$, approximated by the policy network
$\nnet_\theta$ of Section~\ref{sec:gap},
\[
  a \;=\; \nnp(x) \;=\; \nnet_\theta(x),\qquad \nnp:\;x=(k,z)\;\longmapsto\;y=k';
\]
training drives the residual to zero in $\theta$. Equilibrium collapses to
a single Euler equation,
\begin{equation}
\rwm(x,a,Q)=\frac{1}{c}-\beta\,Q=0,\qquad
Q=\Qpi(x)=\E\!\left[\frac{\alpha z'(k')^{\alpha-1}}{c'}\,\Big|\,x,a\right],
\qquad c=zk^{\alpha}-k',
\label{eqn:bm_resid}
\end{equation}
an instance of the general residual \eqref{eqn:residual} in which the continuation
object $\Qpi$ is the expected discounted marginal return to capital.
Table~\ref{tab:bm_map} maps each abstract object of Section~\ref{sec:framework} to its
Brock--Mirman counterpart.

\begin{table}[t]\centering\small

\begin{tabular}{lll}
\toprule
Framework object & Symbol & Brock--Mirman instance \\
\midrule
State        & $x$                       & capital and productivity $(k,z)$ \\
Choice       & $a=\nnp(x)$               & next-period capital $k'$ \\
Transition   & $x'=\Gamma(x,a,\varepsilon')$ & $k$-part $=k'$;\; $\log z'=\rho\log z+\varepsilon'$ \\
Continuation & $\Qpi(x)$                 & $\E[\,\alpha z'(k')^{\alpha-1}/c'\mid x,a\,]$ \\
Residual     & $\rwm(x,a,Q)$             & Euler equation \eqref{eqn:bm_resid} \\
\addlinespace
Training measure (DEQN) & $\muerg$       & the policy's own simulated path \\
Training measure (EWM)  & $\muK$         & exact forward map from stressed seeds \\
\bottomrule
\end{tabular}
\caption{Framework objects of Section~\ref{sec:framework} and their Brock--Mirman
instances. The columns are: Framework object, the abstract object of Section~\ref{sec:framework};
Symbol, its notation; and Brock--Mirman instance, its concrete counterpart in the Brock--Mirman
model. DEQN and EWM share every row but the last: they enforce the same
residual and differ only in the measure on which it is imposed.}\label{tab:bm_map}
\end{table}

The two solvers differ only in how the training batch is generated and in their trainable
objects (Appendix~\ref{app:arch}, Figure~\ref{fig:architecture}): DEQN carries a single policy,
EWM adds the audited continuation surrogate.

\paragraph{DEQN: on-path sampling.}
In DEQN the training states are the policy's own simulated path. Starting from a batch of $M$
feasible states $\{(k^{(m)}_0,z^{(m)}_0)\}_{m=1}^{M}$, we roll the economy forward under the
current policy: for each $m$ and $t=0,1,\dots,T-1$,
\begin{equation}
\varepsilon^{(m)}_{t+1}\sim\Normal(0,\sigma^2),\quad
z^{(m)}_{t+1}=\exp\!\bigl(\rho\log z^{(m)}_t+\varepsilon^{(m)}_{t+1}\bigr),\quad
k^{(m)}_{t+1}=\nnp\bigl(k^{(m)}_t,z^{(m)}_t\bigr).
\label{eqn:bm_deqn_sim}
\end{equation}
The training set is the visited collection
$\mathcal{S}_{\mathrm{path}}=\{(k^{(m)}_t,z^{(m)}_t)\}_{m,t}$, and the policy is trained by
minimizing the mean squared Euler residual accumulated along the simulated
trajectories,
\begin{equation}
\mathcal{L}_{\mathrm{DEQN}}(\theta)=\frac{1}{MT}\sum_{m=1}^{M}\sum_{t=0}^{T-1}
R\bigl(k^{(m)}_t,z^{(m)}_t;\theta\bigr)^2,
\label{eqn:bm_deqn_loss}
\end{equation}
where $R(k,z;\theta)$ is the left-hand side of the Euler equation \eqref{eqn:bm_resid}
evaluated at the network policy, a single number per state. The capital coordinate is whatever the
network currently outputs, so $\mathcal{S}_{\mathrm{path}}$ is distributed as the
policy's own ergodic law $\muerg$; the batch is carried forward from one episode to the next, so as
$\theta$ updates the recursion \eqref{eqn:bm_deqn_sim} drags it along. This is the closed
sampler-approximator loop of \eqref{eqn:pathloss}: the policy is tested only where
it already goes.

\paragraph{EWM: coverage sampling.}
EWM reuses the very same recursion \eqref{eqn:bm_deqn_sim} but imposes the residual on the coverage
measure $\muK=\rho_1\muerg+\rho_2\mu^{\mathrm{stress}}_\kappa+\rho_3\mu^{\mathrm{local}}_\kappa$ of
\eqref{eqn:muK} rather than on $\muerg$ alone. Each batch unions three pools, all generated by the
exact $\Gamma$ of \eqref{eqn:bm_deqn_sim}: the ergodic path itself; a stress pool of fresh
off-ergodic seeds (capital displaced, productivity in the lower tail) rolled $H$ steps forward; and a
local pool of small admissible perturbations of visited states (Figure~\ref{fig:bm_control}
draws the construction). The reach $\kappa$ dials the off-path mass $\rho_2,\rho_3$ and the seed
spread; every off-path state is reachable by $\Gamma$, never a hypercube draw. Writing
$\overline{R^2}_{\mathcal{S}}$ for the mean
squared residual \eqref{eqn:bm_resid} on a set $\mathcal{S}$, minimizing it over this batch gives the
coverage-weighted objective
\begin{equation}
\mathcal{L}_{\mathrm{EWM}}(\theta)=
\rho_1\,\overline{R^2}_{\mathcal{S}_{\mathrm{path}}}
+\rho_2\,\overline{R^2}_{\mathcal{S}_{\mathrm{stress}}}
+\rho_3\,\overline{R^2}_{\mathcal{S}_{\mathrm{local}}},
\label{eqn:bm_ewm_loss}
\end{equation}
the Brock--Mirman instance of the general policy arm \eqref{eqn:Lpolicy}; setting
$\rho_2{=}\rho_3{=}0$ collapses it to the DEQN loss \eqref{eqn:bm_deqn_loss}, while EWM keeps the
off-path weights $\rho_2,\rho_3$ positive, as Figure~\ref{fig:bm_control} draws. Here the continuation is a
one-dimensional expectation over $z'$, cheap by exact quadrature at every state, so the residual is
exact ($R{=}R^{H}$), the world arm \eqref{eqn:Lworld} is trivial ($\wm{=}\Qpi$), and the surrogate
plays no role: Brock--Mirman isolates the contribution of coverage alone, amortized by the surrogate
only where quadrature is costly (Figure~\ref{fig:amortization}).

\begin{figure}[ht!]\centering
\includegraphics[width=0.72\textwidth]{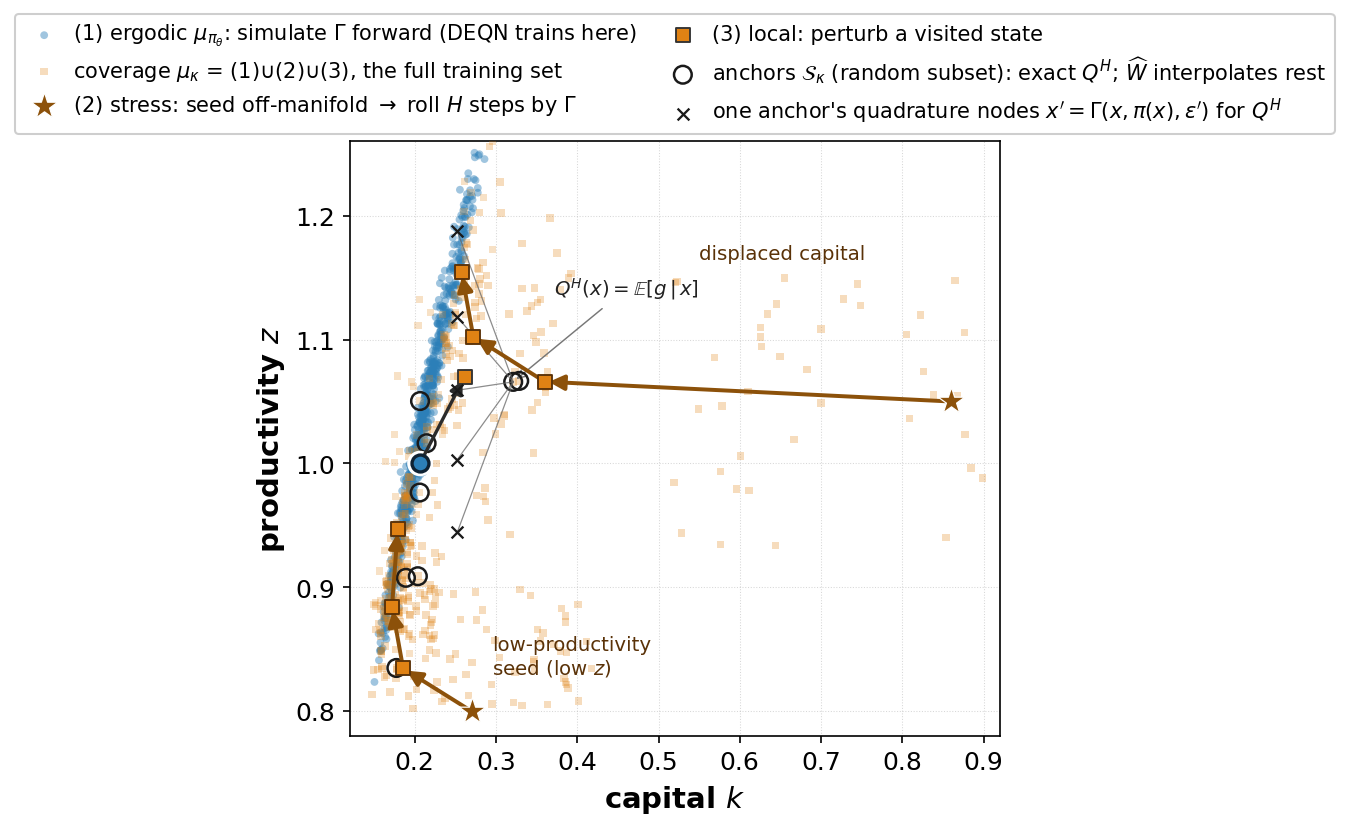}
\caption{How the coverage measure $\muK$ of \eqref{eqn:muK} is built in Brock--Mirman, on the state
$(k,z)$ of capital and productivity, and how the continuation is amortized on it. Building
$\muK$ (the three-step coverage sampling of this section): (1)~ergodic $\muerg$, the
policy's own path, obtained by simulating the exact transition $\Gamma$ forward (blue), the set DEQN
trains on; (2)~stress, seeds drawn off the ergodic set, a low-productivity seed (low $z$) or
displaced capital (brown stars), each rolled $H$ steps forward through the exact $\Gamma$
(arrows), every landing a $\muK$ state (orange squares), the rollout mean-reverting back toward the
ergodic band; and (3)~local, a small perturbation of a visited state (grey arrow). The
training set is the union, $\muK=\rho_1\muerg+\rho_2\mu^{\mathrm{stress}}_\kappa+\rho_3
\mu^{\mathrm{local}}_\kappa$ \eqref{eqn:bm_ewm_loss}; every off-path state is produced by $\Gamma$
from an admissible seed, never an arbitrary box, and each arrow is one exact $\Gamma$-step under the
closed-form policy, not a sketch. On this smooth model the off-path residual is already small, so
coverage is correctly inert here; the contrast appears in the rare-disaster model of
Section~\ref{sub:bmdc}. Amortization: the exact continuation $Q^{H}$ is computed by
quadrature only on a sparse anchor set $\mathcal{S}_\kappa$ (rings), a uniform random subsample of
the coverage states refreshed each episode; at one anchor its exact continuation
$Q^{H}(x)=\E[g\mid x]$ is expanded into the next-period Gauss--Hermite quadrature nodes
$x'=\Gamma(x,\nnp(x),\varepsilon')$ (crosses, joined to the anchor). The surrogate $\wm$ carries the
continuation to the rest, so exact solves scale as
$\mathcal{O}(K)$ with $K\ll N$ while the reported error stays the exact $R^{H}$. On Brock--Mirman the
continuation is closed-form, so $\wm$ is unnecessary ($\wm=\Qpi$) and the anchors are illustrative;
the saving reaches $40\text{--}70\times$ in high dimensions (Section~\ref{sub:irbc}).}
\label{fig:bm_control}\label{fig:amortization}
\end{figure}

Because the residual \eqref{eqn:bm_resid} vanishes identically on and off the ergodic path, DEQN and
EWM are indistinguishable here by construction, with nothing for coverage to repair. Brock--Mirman in
this benign form earns its place not as a performance comparison but as the model where the two arms
and the coverage measure can be drawn and checked by hand (Figure~\ref{fig:bm_control}); the first
measurable gap between the methods appears only once the economy has an off-ergodic region the policy
cannot generalize across, the rare disaster we add next (Section~\ref{sub:bmdc}).

\subsubsection{Brock--Mirman with a rare disaster}\label{sub:bmdc}

We next add the smallest departure from the smooth Brock--Mirman benchmark that creates a genuine
coverage problem. The state is $x=(k,z,d)$, where $d\in\{0,1\}$ is a disaster regime. The disaster is
rare but persistent: it is entered with probability $p_d=0.002$, remains active with probability
$p_{dd}=0.6$, and scales output to
$Y=e^{z}(1-bd)k^{\alpha}$ with $b=0.45$. Its stationary probability is therefore about $0.5\%$.
We also impose irreversible investment. With capital accumulation
$k'=(1-\delta)k+i$ and investment $i=Y-c\ge0$, the non-negativity constraint has multiplier
$\mu\ge0$ and is written as the Fischer--Burmeister equation
$\phi(i,\mu)=i+\mu-\sqrt{i^2+\mu^2}=0$. The constraint binds after a disaster shock: the agent would
like to reduce capital quickly, but disinvestment is infeasible. This kink is precisely the feature
that a policy trained only on the normal ergodic region does not learn to handle.

The exact residual contains the Euler equation and the complementarity condition. The Euler block is
\begin{equation*}
  c\bigl(\beta\Qpi+\mu\bigr)-1=0,
  \qquad
  \Qpi(x)=\E\!\left[\frac{1}{c'}\bigl(\mathrm{MPK}'+1-\delta\bigr)\,\bigm|\,x\right],
\end{equation*}
where
$\mathrm{MPK}'=\alpha e^{z'}(1-bd')(k')^{\alpha-1}$ is the next-period marginal product of
capital, evaluated at the next-period state and policy-implied consumption $c'$. The expectation is
over the next-period productivity innovation and the two-state disaster transition. We use the
calibration $\alpha=0.36$, $\beta=0.95$, $\delta=0.10$, $\rho=0.9$, and $\sigma=0.04$. The purpose of
this example is not to study a family of calibrations. It is to put the coverage gap in a
low-dimensional model where the relevant states can be plotted and where solver variants that differ
only in the measure on which the residual is imposed can be compared directly. Table~\ref{tab:bmdc}
reports ten independent seeds for each variant.

The comparison has three rungs. The pathwise baseline, \texttt{DEQN-path-exact}, imposes the exact
residual only on the policy's own simulated path; it uses no coverage states and no surrogate. The
coverage control, \texttt{DEQN-coverage-exact}, imposes the same exact residual on the enlarged
coverage measure $\muK$; it uses coverage but no surrogate. The full EWM arm,
\texttt{EWM-coverage-surrogate}, uses the same coverage measure and replaces repeated exact
continuation evaluations inside the policy loop by the audited continuation surrogate. Thus the policy
class, the transition law, the calibration, and the held-out exact residual are common across arms.
All runs use $3000$ episodes; the reported policy is a Polyak average over the final $500$ episodes of
the carried-forward simulation.

The results identify coverage as the main accuracy lever. The pathwise arm is accurate in the normal
region but has a disaster-region residual about sixteen times larger. Imposing the exact residual on
the coverage states reduces the disaster residual by $7.3\times$ at the mean and by about $10\times$ in
the tail. The surrogate arm remains more than threefold more accurate in the disaster region than the
pathwise baseline, while using about one sixth of the policy-loop exact-quadrature evaluations of the
surrogate-free coverage control. Normal-region errors remain small in all three arms.

Two held-out diagnostics make the same point. First, the disaster test set is far from the pathwise
training law. Section~\ref{sec:theory} later formalizes this mismatch through the distance between a
test measure and the coverage measure; here we report the empirical ingredients directly. The disaster
regime has only about $0.5\%$ of the ergodic mass but $34\%$ of the coverage mass, so the total-variation
distance between the ergodic law and the coverage law is at least $0.33$. In the capital dimension,
$41\%$ of the disaster-region held-out states lie below the first percentile of the ergodic-path
capital distribution. These are not merely low-probability realizations on the normal path; they are
states that the pathwise objective essentially never samples. Second, the improvement is not only an
average effect. The supremum norm $\lVert R^{H}\rVert_{\infty}$, here the maximum absolute exact
residual over the held-out disaster states, falls from $3.1\times10^{-2}$ for the pathwise arm to
$4.1\times10^{-3}$ for the coverage-exact arm. This is a $7.6\times$ reduction in the worst case. The
surrogate arm is intermediate at $7.4\times10^{-3}$, the same ranking shown by the mean and tail
columns of Table~\ref{tab:bmdc}.

The figures illustrate the mechanism behind the table. Figure~\ref{fig:bm_control} shows where the
coverage states lie in the $(k,z)$ plane. Figure~\ref{fig:bmdc_conv} then reports the dynamic evidence:
the disaster residual over training, the exact-compute frontier, and the residual along a
Koop--Pesaran--Potter generalized impulse response to a disaster onset. The impulse-response path is
rolled forward with the exact transition $\Gamma$. The pathwise arm is least accurate at impact, when
the state is farthest from its ergodic support, and the gap closes as the disaster state mean-reverts.
This is the narrow claim of the example: adding structurally generated coverage states repairs the
self-confirming off-path gap left by a pathwise residual solver.

\begin{table}[t!]\centering\small

\input{tab_bmdc.tex}
\caption{Brock--Mirman with a rare disaster and irreversible investment. The columns are: Arm, the
solver configuration; normal, the median held-out exact residual $R^{H}$ on the on-path region;
disaster mean, the median held-out $R^{H}$ over the disaster-region set; disaster $p_{95}$, its 95th
percentile over that set; $B_{\mathrm{policy}}$, the cumulative policy-loop exact-quadrature budget in
millions; and low-error, the number of seeds whose held-out disaster-region mean residual is below
$10^{-2}$. All arms use the same
calibration, architecture, episode budget, and held-out evaluation sets; they differ only in the
training measure and in whether the continuation is evaluated exactly at every policy step or through
the audited surrogate. Each arm is run for $3000$ episodes on ten independent seeds, and each reported
policy is a Polyak average over the final $500$ episodes. Cells report the median held-out exact
residual $R^{H}$ across seeds, with the interquartile range in brackets. $B_{\mathrm{policy}}$ is the
cumulative policy-loop exact-quadrature budget, in millions of evaluations.
The final column reports the
number of seeds whose held-out disaster-region mean residual is below $10^{-2}$; stationarity is
checked separately by the sup-norm policy drift criterion on a held-out grid.
}\label{tab:bmdc}
\end{table}

\begin{figure}[t!]\centering
\includegraphics[width=\textwidth]{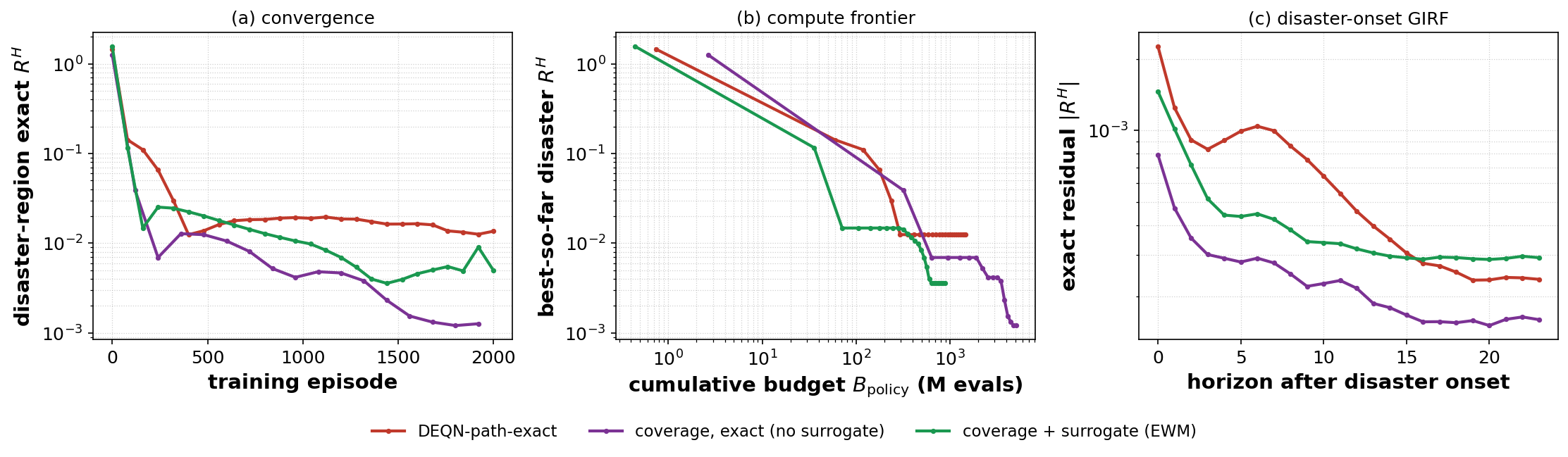}
\caption{Brock--Mirman with a rare disaster: held-out exact disaster residuals by arm. Panel (a) plots
$R^{H}$ in the disaster region against the training episode. All arms improve, but the pathwise
baseline plateaus above the two coverage arms because it does not train on the off-ergodic disaster
region. Panel (b) plots the best-so-far disaster residual against the cumulative policy-loop exact
quadrature budget $B_{\mathrm{policy}}$. The surrogate arm approaches the coverage-exact frontier while
using substantially fewer exact evaluations. Panel (c) plots $|R^{H}|$ along a
Koop--Pesaran--Potter generalized impulse response to a disaster onset \citep{koop1996impulse},
computed as the average difference between paths that enter the disaster at horizon zero and
counterfactual paths that do not, both rolled forward by the exact transition $\Gamma$. The pathwise
arm is least accurate at impact, where the economy is farthest from its ergodic set, and the residual
falls back as the disaster dissipates. The gross $(k,c)$ responses are close across arms; the solver
difference is in how accurately each policy satisfies its equilibrium conditions off the ergodic path.
All panels report exact $R^{H}$ on held-out states.}
\label{fig:bmdc_conv}\label{fig:bmdc_girf}
\end{figure}

\paragraph{The surrogate's disaster residual: amortization, not degradation.} At convergence the
surrogate arm's disaster residual lies between the pathwise baseline and the surrogate-free
coverage-exact arm (Table~\ref{tab:bmdc}). This ordering has a simple interpretation. The
coverage-exact arm recomputes the true quadrature continuation at every policy update. The surrogate
arm computes exact continuation targets on the coverage batch, fits $\wm$ to those targets, and then
uses $\wm$ inside the inner policy updates. The residual difference between the two coverage arms is
therefore the measured continuation-fit error of the surrogate on the covered disaster states. The
formal decomposition in Section~\ref{sec:theory} gives this component a name; no additional object is
being introduced here. In this example the component is small relative to the pathwise coverage gap,
and it buys a sixfold reduction in policy-loop exact evaluations. The result should be read as the
expected accuracy-amortization trade-off of using a learned continuation, not as a failure of the
coverage mechanism.

\paragraph{Coverage, not optimization, is what is missing off-path.} The experiment separates two
possible explanations for the pathwise failure. If optimization were the bottleneck, the residual would
remain large on the states used for training. Instead, the pathwise arm drives its on-path residual to
a small value while its held-out disaster residual remains large. The problem is therefore not that the
policy fails to minimize its own objective; the problem is that the objective omits the disaster states
on which the policy is later evaluated. Section~\ref{sec:theory} formalizes this distinction by
bounding the test residual by an in-coverage residual plus a change-of-measure term. The empirical
reading here is more direct: enlarging $\muK$ puts the missing disaster states into the objective, and
the held-out disaster residual falls.

\paragraph{What is varied in the sweeps.} It is useful to keep two dials separate. In this
Brock--Mirman experiment, increasing the coverage reach $\kappa$ means enlarging the stress and local
parts of $\muK$: more mass on disaster and post-disaster states, larger admissible capital
displacements, and deeper forced-disaster rollouts. Appendices~\ref{app:bmdc_sampling} and~\ref{app:bmdc_cov} give these coverage
settings. The endpoint $\kappa=0$ corresponds to the pathwise row of Table~\ref{tab:bmdc}; the
full-coverage endpoint corresponds to the coverage-exact row. This sweep changes where the exact
residual is imposed. It does not change the network architecture.

At scale, the corresponding sensitivity exercise is different. Figure~\ref{fig:irbc_kappa} varies the
surrogate's perception capacity and fidelity at fixed coverage reach: the surrogate width is increased
from $m=16$ to $32$ to $64$, and the exact-target anchor share from $0.1$ to $0.2$ to $0.4$.
Table~\ref{tab:hyper} gives the network settings, and Appendix~\ref{app:irbc_cov_spec} gives the fixed
coverage measure used in that exercise. Thus Figure~\ref{fig:irbc_kappa} is not evidence from a larger
coverage reach; it asks how large and how well-audited the continuation network must be once the
coverage region has been fixed. The practical rule is consequently two-step: first choose coverage
that reaches the economically relevant off-path region; then increase surrogate capacity and exact
anchor coverage until the held-out exact residual no longer improves materially.

The surrogate $\wm$ in the Brock--Mirman rare-disaster experiment is intentionally small: two hidden
layers of width $32$ with $\tanh$ activations, mapping the scaled state $(\log k,z,d)$ to the scalar
continuation $\Qpi$, with about $1.2\times10^{3}$ parameters. It has the same size and trunk structure
as the policy it serves. Its error is small on the covered disaster region, including the
Fischer--Burmeister kink where the investment constraint switches on, and grows only in the
extrapolation corner outside the covered support. The audit is designed precisely for that case: when
$\wm$ and exact quadrature disagree on sensitive states, those states are routed back to the exact
residual rather than being left to the policy update. The kink itself is not where the error lives.

\paragraph{What this isolates: the cross-model reading.} In this rare-disaster Brock--Mirman model the
surrogate-free coverage arm trains without stalling and is the most accurate arm in the disaster
region (Table~\ref{tab:bmdc}). The only non-smoothness is the Fischer--Burmeister kink, a continuous
and Lipschitz residual at the point where the irreversibility multiplier switches on. Thus the exact
coverage objective is well-conditioned enough here, and the surrogate's role is amortization: it
replaces repeated exact continuation evaluations inside the policy loop by a small, audited
approximation error. The endogenous-protection model in Section~\ref{sub:actcond} is different. There
the action changes the law of the disaster itself, and a state-only continuation cannot price the
marginal effect of the action. The two Brock--Mirman variants therefore separate the two roles of the
world component: amortization when exact coverage is feasible but expensive, and conditioning when the
state-only exact residual does not provide a usable descent direction.\footnote{The sampling,
calibration, and coverage-parameter specifications are in Appendices~\ref{app:bmdc_sampling} and
\ref{app:bmdc_cov}.}

\subsection{World-component variants}\label{sub:variants}

The world component is the one part of an equilibrium world model that adapts to the economy being
solved; the invariants of Table~\ref{tab:components}, the exact transition $\Gamma$, the structural
residual, and the held-out audit $R^{H}$, never change. Which world component to use, and how to
modify it, follows from the structure of the problem: whether the costly object is a continuation that
only needs amortizing, a continuation whose value an action itself moves, or an aggregate state too
large for the policy to take as a direct input. The paper uses three instances, in increasing order of what the
world component must carry; each is a modification chosen for the model at hand, never a change to the
conditions that certify the solution.

\paragraph{Continuation surrogate $\wm(x)$.} This is the state-only continuation used in the worked
example above and in the international real business cycle experiment of Section~\ref{sub:irbc}. It
approximates $\Qpi(x)$ on the coverage set and is audited by $R^{H}$; its role is amortization and
conditioning, and it underlies the accuracy results of Section~\ref{sec:exp}.

\paragraph{Action-conditioned continuation $\widehat W(x,a)$.}\label{sub:actcond} When an action
changes the probability law of future regimes rather than only the next physical state, a state-only
continuation averages away the very margin the action controls and leaves it unidentified. Endogenous
protection is the canonical case: the protection choice changes disaster exposure, so the world
component must resolve the continuation by the action that moves it. The move is the one from a value
function to a Q-function in dynamic programming: a state-only continuation evaluates the action the
policy already prescribes, while an action-conditioned continuation evaluates a deviation from it, the
margin a measure-moving action turns on. This belongs to the definition of EWM as an
architecture, not to the coverage-accuracy results of Section~\ref{sec:exp}. Appendix~\ref{sub:bmdcp}
works this case out in full: a calibrated example, the four experimental variants that together show
the action-conditioned continuation is the one that recovers the correct policy, and a proof that a
state-only continuation cannot value how the action shifts disaster exposure. Coverage decides where
the equilibrium conditions are imposed; the world component decides how the continuation is carried,
and where an action moves the measure no choice of sampling can substitute for it, so EWM's
contribution is not coverage sampling alone.

\paragraph{Distributional encoder $h_\psi$.} When the aggregate state is itself a cross-sectional
population $\mu$, the world component is a learned, permutation-invariant encoder that compresses it
into a finite-dimensional summary the policy conditions on in place of $\mu$,
\begin{equation}
  z=h_\psi(\mu)\in\R^d,\qquad \mu'=\Phi(\mu,a,\varepsilon'),\qquad a=\nnp(\mu),
\label{eqn:encoder}
\end{equation}
where $\Phi$ is the exact distributional transition, the counterpart of $\Gamma$ for distributions
and, like $\Gamma$, never learned. The push-forward depends on the equilibrium
decisions $a=\nnp(\mu)$ applied across the cross-section and on the aggregate shock, so the
next-period population is determined by the current population, the policy, and the realized
aggregate state. In the heterogeneous-agent application of Section~\ref{sub:bewley}, $\Phi$ is the
non-stochastic simulation update of \citet{young2010}, which advances the distribution exactly by
pushing the population mass forward through the policy and the shock onto a fixed grid,
$\mu'=\Phi(\mu,a',z)$. This summary $z$ is low-dimensional yet preserves the information the
equilibrium needs, in the lineage of the learned generalized moments of \citet{han2023deepham} and the
LeWorldModel of \citet{MaesLeLidec2026LeWM}, and a data-driven state reduction in the spirit of
active-subspace methods \citep[applied to economic models in][]{Scheidegger2019JCS}. Trained by a
joint-embedding predictive objective with an anti-collapse regularizer
\citep{LeCun2022AMI,BalestrieroLeCun2025SIGReg,MaesLeLidec2026LeWM}, it is the mechanism used in the
heterogeneous-agent experiment of Section~\ref{sub:bewley}; to our knowledge this is the
first use of a joint-embedding predictive architecture to solve, rather than estimate, a
dynamic stochastic economic model. Throughout, the EWM principle is unchanged: learn perception where
necessary, never the equilibrium conditions that certify the solution.

\section{Theory: From Path Confirmation to Rational Expectations}\label{sec:theory}

This section makes the central claim of the paper precise: that widening the coverage
reach $\kappa$ carries the solver's fixed point from a self-confirming equilibrium,
correct only on its own simulated path, to a rational-expectations equilibrium, correct
everywhere the coverage schedule certifies, its own path together with the audited off-ergodic
regions, with the full reachable set as the idealized limit of enriching coverage. One quantity tracks the journey, the distance between
where a candidate solution is tested and where the world model has actually imagined. We
measure that distance in total variation: for two probability measures $\mu$ and $\nu$,
\[
\TV(\mu,\nu):=\sup_{A}\,|\mu(A)-\nu(A)|\in[0,1],
\qquad\text{so that}\qquad
|\E_\mu f-\E_\nu f|\le 2\|f\|_\infty\,\TV(\mu,\nu)
\]
for every bounded measurable $f$. The coverage gap $\TV(\mu,\muK)$, between a test measure
$\mu$ and the coverage measure $\muK$ of Section~\ref{sec:framework}, is then a worst-case
diagnostic: it vanishes only when the test law and the coverage law coincide, not merely when
their supports overlap, so support expansion alone need not reduce it, and it bounds
the residual transport between any two measures regardless of their geometry. We are careful not
to overstate it. The homotopy's actual mode of convergence is weak,
$\muK\Rightarrow\mu_\infty$ (Assumption~\ref{ass:nested}), which does not imply
$\TV(\mu,\muK)\to0$ for a fixed test measure, and indeed cannot when $\mu$ is singular with
respect to $\muK$ (an empirical grid against a continuous coverage law). The total-variation term
is therefore the conservative fallback of Proposition~\ref{prop:decomp}; for the dominated audits
we actually report it is superseded by the sharper density-ratio bound of
Theorem~\ref{thm:offpath}, and the operational progress measure throughout is the held-out
off-path residual, not an estimate of $\TV$ itself. To keep the main text readable,
all proofs, together with the most standard supporting results (existence and uniqueness of
the projection, Proposition~\ref{prop:exist}; its monotonicity in the perception class,
Proposition~\ref{prop:mono}; the surrogate--exact gap, Proposition~\ref{prop:surrgap}; and
the off-ergodic discrepancy, Lemma~\ref{lem:offpath}), are collected in
Appendix~\ref{app:proofs}. The main text keeps only the results that carry the argument.

\paragraph{The theory in plain terms.} Three statements, in words. (a) At zero
coverage the solver's fixed point is a self-confirming equilibrium: the policy is correct
on the very measure it trains on and unconstrained off it. (b) Off-path, the residual decomposition has four
terms, optimization, perception-class approximation, surrogate-fit, and coverage
(Proposition~\ref{prop:decomp}), but coverage is the term a pathwise solver never controls off
its own path, and it is the distinctive lever EWM adds; the other three are already in play
on-path, and surrogate fidelity is the only additional term the surrogate introduces. So the two
dials EWM turns are coverage reach, how far the test states lie outside the coverage measure, and
surrogate fidelity, how well the learned continuation matches the exact one there. (c) As imagination
expands to fill the reachable set and the surrogate becomes exact, the fixed points
converge to rational expectations. The off-path residual is the diagnostic that tracks progress
along the way.

\paragraph{Scope of the theorems.} We are precise about what the theorems do and do not
establish. (i) Conditional, not a guarantee. Theorem~\ref{thm:consist}
shows that if the perception gap, the surrogate-fit error, and the optimization
error vanish and the coverage measure fills the reachable set, then accumulation
points are rational-expectations equilibria; it does not prove that a given coverage
schedule achieves these conditions, nor does it supply a rate. The conditions are objects
the modeler checks empirically, they are exactly the terms of the residual
decomposition (Proposition~\ref{prop:decomp}) and the coverage sweep; quantifying the
homotopy's rate as a function of the coverage reach $\kappa$, and characterizing which
off-ergodic regions a given coverage schedule can reach, remain open. (ii) Population,
not SGD. The statements characterize idealized fixed points (best-in-class projections on
$\muK$), not the SGD iterate that the network actually computes; whether
training reaches the fixed point is an empirical question the experiments answer. The same
boundary scopes the compute and robustness headlines: the seed-robustness results are
measured, not proved (a plausible mechanism, which we do not formalize, is that coverage
breaks the closed sampler-approximator loop of \eqref{eqn:pathloss}, in which the training
measure is dragged by the very policy it trains, the feedback the introduction identifies as
the source of distributional runaway; imposing the residual on a designer-controlled support
removes it, but we offer this only as an interpretation of the measured robustness, not a
theorem); the amortization index $A(\tau)$ (the per-step wall-clock speed-up the audited continuation
surrogate buys over exact quadrature) has its growth with the model
dimension accounted for by Proposition~\ref{prop:amort}, while its numerical value is a
measured constant.
(iii) ``Self-confirming'' is exact only at $\kappa=0$, and only when the perception gap is
closed on the path. With $\muK$ equal to the ergodic measure $\muerg$, the outer loop
self-consistent, and the perception gap closed there ($\eta_\kappa(\theta)=0$, so $\wm=\Qpi$ and
$R^H=0$ $\muerg$-a.s.), the fixed point is a Sargent self-confirming equilibrium in the literal
sense, the agent's belief being confirmed on the events its behavior visits. Without those
equalities it is a coverage-confirmed projection-best-response fixed point, closer to a
restricted-perceptions equilibrium when $\calF_m$ binds; and for $\kappa>0$, $\muK$ is the
modeler's chosen measure, so the object is coverage-confirmed regardless: a fixed point of
projection-best-response on a designer-controlled support, not an agent's self-confirming belief. (iv) The
off-path bound needs domination. Theorem~\ref{thm:offpath} controls $\E_\eta\|R^H\|^2$ only
for test measures $\eta$ dominated by $\muK$. A single deterministic impulse-response path
is $\muK$-null (density ratio $B=\infty$) and is not covered; the generalized impulse
response we actually report \citep[the Koop--Pesaran--Potter object,][]{koop1996impulse}, an average over a distribution
of disaster-entering and counterfactual paths, Section~\ref{sub:bmdc}) is such a distribution,
dominated by $\muK$ with $B<\infty$, so the bound does cover the reported object even though
it would not cover a single path. The supporting
propositions are standard projection and change-of-measure facts; the contribution is the
coverage-indexed bridge they assemble into, and its use as a verifiable diagnostic.

\noindent The formal assumptions are collected in Appendix~\ref{app:assumptions}: a regular
structural environment (Assumption~\ref{ass:regular}, compactness, continuity, and a bounded,
$L_R$-Lipschitz residual, with $L_R$ read off the equations exactly, the discount factor in the
affine case), nested perception classes of growing capacity $m$ with full-support coverage in
the reach limit (Assumption~\ref{ass:nested}, the two levers $\kappa$ and $m$), a unique
invariant measure (Assumption~\ref{ass:invariant}), and audit domination
(Assumption~\ref{ass:audit}, the density bound $B$ that controls off-path transport and is
vacuous on a single measure-zero path but finite on the reported tube and region distributions).
We state the results here and prove them, with the assumptions, in Appendix~\ref{app:proofs}.

\paragraph{Self-confirming, not self-justified.} Which departure from rational
expectations this is matters, because two distinct restrictions are
easy to conflate. A restricted-perceptions equilibrium \citep{BranchEvans2006RPE} or a
self-justified equilibrium \citep{KUBLER2023105707,10.1093/jeea/jvaf062} is limited by its
perception class: the agent forecasts the continuation within a restricted function
family and attains only the best approximation in that class, which can be inexact even on the
ergodic path. A self-confirming equilibrium \citep{fudenberg1993selfconfirming,sargent1999conquest}
is limited instead by its data: beliefs are correct on the events the agent's own behavior
generates and untested off them, while the forecast class itself need not be misspecified. These
are orthogonal restrictions, and they are the two error sources
of the residual decomposition developed below, a perception-class term and a coverage term. A modern deep-learning solver
has largely removed the first: a neural network is an expressive forecaster, and on the states it
trains on it drives the exact residual to numerical tolerance rather than to a nonzero
best-in-class projection. What still binds is that it samples only its own ergodic path, so its
forecasts are confirmed where it goes and silently uncertified off it. That is why our baseline endpoint
is a self-confirming equilibrium rather than a self-justified one, and why this paper attacks the
coverage axis rather than the perception-class axis. The two need not be mutually exclusive in
finite samples, since finite capacity could reintroduce a class restriction; what makes coverage
the operative constraint here is that the covered-region residual is driven to near zero
(Section~\ref{sec:exp}), not left at a best-in-class floor.

Stated as instruments, the two gaps close by different means. The perception-class gap
closes by enriching the forecast class: the function-space sieve on finite-dimensional
states, and, when the state is itself a distribution, the learned distributional sieve
$h_\psi$ of Section~\ref{sub:econ}. The coverage gap closes by enlarging the support on
which the exact residual is imposed. The two closures are orthogonal, the two error terms
of the residual decomposition (Proposition~\ref{prop:decomp}), and rational expectations is
the joint limit in which both vanish. On the finite-dimensional models a deep network leaves
the perception gap slack, so coverage is the binding axis; on heterogeneous-agent economies
the choice of distributional summary makes the perception axis operative again, which is why
the Bewley experiment of Section~\ref{sub:bewley} turns on the encoder.

\subsection{Existence, monotonicity, and the residual decomposition}

Under Assumptions~\ref{ass:regular}--\ref{ass:nested}, for every admissible policy
$\pi$ and every $\kappa$ the continuation $\Qpi$ is bounded and lies in
$L^2(\muK;\R^{d_Q})$, so the projection onto the closed convex class $\calF_m$ exists and
is unique ($\muK$-a.s.) by the Hilbert projection theorem; write
\[
\hat q^{\pi}_\kappa:=\Pi^{L^2(\muK)}_{\calF_m}\Qpi
=\argmin_{q\in\calF_m}\E_{x\sim\muK}\bigl[\|q(x)-\Qpi(x)\|^2\bigr]
\]
for it (Proposition~\ref{prop:exist}, Appendix~\ref{app:proofs}). Define the perception gap
$\eta_\kappa(\theta):=\inf_{q\in\calF_m}\E_{\muK}[\|q(x)-\Qpi(x)\|^2]$: the
in-class error of the best imaginable world model on the coverage measure.

\noindent At a fixed evaluation measure, enlarging the perception class weakly reduces the
in-class approximation error, $m\le m'\Rightarrow
\inf_{q\in\calF_{m'}}\E_\mu[\|q-\Qpi\|^2]\le\inf_{q\in\calF_m}\E_\mu[\|q-\Qpi\|^2]$
(Proposition~\ref{prop:mono}, Appendix~\ref{app:proofs}); when the coverage measure itself
varies with $\kappa$ this is the class-approximation component only, the change of measure
being governed separately by the coverage gap of Proposition~\ref{prop:decomp}.

\noindent In words: off the simulated path, a solution's error splits into what the solver
optimizes on its coverage set and what that set fails to reach.

\begin{proposition}[Residual decomposition: coverage as the off-path term]\label{prop:decomp}
Let $\mu$ be any test measure (ergodic, disaster-conditioned, or
impulse-response), and let $\hat q^{\nnp}_\kappa\in\calF_m$ be the population
projection of the current policy's continuation
(Proposition~\ref{prop:exist}). With
\begin{align*}
\mathrm{OptErr}&:=\E_{\muK}\!\bigl[\|\rwm(x,\nnp(x),\wm(x))\|^2\bigr], &
\mathrm{ApproxErr}_\kappa&:=L_R^2\,\eta_\kappa(\theta),\\
\mathrm{SurrFitErr}&:=L_R^2\,\E_{\muK}\!\bigl[\|\hat q^{\nnp}_\kappa(x)-\wm(x)\|^2\bigr], &
\mathrm{CoverageErr}&:=2\|\rwm\|_\infty^2\,\TV(\mu,\muK),
\end{align*}
the exact equilibrium residual on the test measure obeys
\[
\E_{x\sim\mu}\!\bigl[\|\rwm(x,\nnp(x),\Qpi(x))\|^2\bigr]
\;\le\;
3\bigl(\mathrm{OptErr}+\mathrm{ApproxErr}_\kappa+\mathrm{SurrFitErr}\bigr)
+\mathrm{CoverageErr}.
\]
The bound combines a change of measure (the $\mathrm{CoverageErr}$ term) with a
three-term split of the residual on $\muK$; the proof is in
Appendix~\ref{app:proofs}. Here $\mathrm{SurrFitErr}$ is the gap between the trained
surrogate and the best in-class projection of the continuation, that is, the
surrogate-fitting error, not a numerical-integration error; genuine finite-quadrature
error in evaluating $\Qpi$ on the known transition is a separate term, controlled by the
fixed quadrature rule and assumed negligible throughout.
\end{proposition}

\noindent The decomposition isolates the mechanism. A pathwise solver controls
the first three terms only on $\mu=\muerg$; on any off-ergodic test measure its
error is governed by $\mathrm{CoverageErr}=2\|\rwm\|_\infty^2\,\TV(\mu,\muerg)$,
which it never reduces, it never imagined there. EWM attacks
$\mathrm{CoverageErr}$ directly by expanding $\muK$ toward $\mu$. The term is a
distance between measures, not a variance: it carries no sample-size rate and
does not shrink as the number of path draws grows, the formal sense in which
coverage is a property of the objective rather than of the sampler
(Section~\ref{sub:coverage}). The coverage
gap $\TV(\mu,\muK)$ measures progress along the
self-confirming-to-rational-expectations transition: it is exactly what
separates a solver that is right where it looks from one that is right where the
economy can go. The experimental ladder of Section~\ref{sub:protocol} isolates these terms one at a time: a
pathwise baseline carries the full $\mathrm{CoverageErr}$ off the ergodic path; a
coverage-matched arm that imposes the exact residual on the enlarged measure removes the support
mismatch and is left with $\mathrm{OptErr}$ alone; and the surrogate arm adds back only
$\mathrm{SurrFitErr}$. The disaster-region residuals in Tables~\ref{tab:bmdc} and
\ref{tab:irbc} order themselves accordingly; the three arms are named and specified in
Section~\ref{sub:protocol}. We do not estimate $\TV$ directly, which is not
tractable between high-dimensional measures; the held-out disaster-region residual, together
with the density ratio $B$ of Theorem~\ref{thm:offpath}, is its operational proxy throughout
the experiments.

\noindent In words: a counterfactual is only as accurate as the coverage that reaches it, and
no more inaccurate than the in-coverage error the audit can measure.

\begin{theorem}[Exact off-ergodic residual bound]\label{thm:offpath}
Let $\eta$ be any test measure (an impulse-response, disaster, or large-shock
distribution) absolutely continuous with respect to the coverage measure $\muK$, with
$d\eta/d\muK\le B$. Then the exact equilibrium residual on $\eta$ is controlled by the
two quantities the solver drives down on the coverage measure, the surrogate residual
and the continuation-approximation error:
\[
\E_\eta\!\bigl[\|R^{H}\|^2\bigr]
\;\le\;
2B\,\E_{\muK}\!\bigl[\|R^{L}\|^2\bigr]
\;+\;
2B\,L_R^2\,\E_{\muK}\!\bigl[\|\wm-\Qpi\|^2\bigr],
\]
where $R^{H}(x)=\rwm(x,\nnp(x),\Qpi(x))$ is the exact residual and
$R^{L}(x)=\rwm(x,\nnp(x),\wm(x))$ the surrogate residual.
\end{theorem}

\noindent Theorem~\ref{thm:offpath} is the formal content of the headline: it bounds the
exact off-path residual, the quantity we report, by the surrogate residual that
training minimizes and the continuation error that the audit measures, amplified only by
the coverage density bound $B$. Lemma~\ref{lem:offpath} is the supporting step that
isolates the continuation-error contribution. A counterfactual experiment that stays
within the imagined support ($B<\infty$) inherits the trained accuracy; one that leaves
it ($B=\infty$) is exactly where a self-confirming solver is silently uncertified. For the
off-path experiments that matter here, the impulse-response and branch distributions
dominated by the coverage measure, this density-ratio bound is the operative one and is
sharper than the uniform total-variation term in Proposition~\ref{prop:decomp}; the
latter we keep only as a worst-case fallback when no density ratio is available, and it
becomes uninformative precisely when $\eta$ is singular with respect to $\muK$, the case
in which a Wasserstein control would be the natural substitute (an open direction).

\noindent On any measure $\mu$, the exact equilibrium error is bounded by the surrogate
(training) residual plus the continuation-approximation error,
$\E_\mu[\|R^{H}\|^2]\le 2\,\E_\mu[\|R^{L}\|^2]+2L_R^2\,\E_\mu[\|\wm-\Qpi\|^2]$
(Proposition~\ref{prop:surrgap}, Appendix~\ref{app:proofs}); driving both to zero drives
the exact residual to zero. This is what licenses training on $R^{L}$ while reporting
$R^{H}$, and it separates the two ways a surrogate can hurt: a poorly fit continuation
(second term), or a policy that has learned to satisfy a wrong surrogate residual
(first term small, second large), which the audit of Section~\ref{sub:surrogate} is
designed to detect.

\paragraph{Auxiliary results.} Four further results, stated and proved in
Appendix~\ref{app:proofs}, support the construction without bearing the headline claims.
Proposition~\ref{prop:amort} (amortization) shows that replacing exact quadrature by a surrogate
query saves a per-continuation cost of $\Theta(n_q(N))=\Theta(N)$ exact evaluations, linear in the
shock dimension, the compute counterpart to the accuracy bounds above; we rest the scaling claim
on this hardware-independent evaluation count, not on timings.
Proposition~\ref{prop:actcond} (identification) shows that when an action moves the transition
measure rather than the next state, a state-only continuation integrates the margin away
($\partial_a\wm\equiv0$) while an action-conditioned surrogate restores it, the surrogate required
to be accurate only at the value level; this is the identification content behind the
action-conditioning ladder of Appendix~\ref{sub:bmdcp}: an in-scope component of the architecture,
required wherever an action moves the measure, though not part of the coverage-accuracy headline. Proposition~\ref{prop:fb} (complementarity) writes occasionally-binding
constraints in Fischer--Burmeister form and places them inside the same $L^2$ theory without
smoothing, feasibility and complementary-slackness violations remaining distinct diagnostics not
recoverable from the stacked residual. Proposition~\ref{prop:girf} (impulse-response stability)
bridges the one-step residual the theory controls to the multi-step generalized impulse response
the experiments report, the horizon amplification set by the model's own stability modulus
$\lambda$.

\subsection{Finite imagination: a coverage-confirmed fixed point on the coverage measure}

\noindent In words: at finite coverage the solver computes a coverage-confirmed equilibrium,
exact on the states it imagines but not yet rational expectations.

\begin{theorem}[Fixed-point characterization of finite-coverage equilibria]\label{thm:br}
Fix $\kappa$ and a (Polyak) target $\bar\theta$ with $\muK=\muK(\bar\theta)$ held
at the target. Define the policy-selection correspondence and the projection map
\begin{equation*}
\Theta_\kappa(\widehat W):=\argmin_{\theta\in\Theta}\E_{\muK}\bigl[\|\rwm(x,\nnp(x),\widehat W(x))\|^2+\omega\,\calK(x,\nnp(x))\bigr],
\qquad
T_\kappa(\widehat W):=\Pi^{L^2(\muK)}_{\calF_m}\bigl(Q^{\pi_{\theta(\widehat W)}}\bigr),
\end{equation*}
where $\theta(\widehat W)$ is a measurable selector (Berge's maximum theorem with
Kuratowski--Ryll-Nardzewski selection, under Assumption~\ref{ass:regular}, on the admissible
parameter set $\Theta_{\mathrm{ad}}:=\{\theta\in\Theta:\nnp(x)\in\A(x)\ \text{for all }x\in\X\}$,
assumed non-empty and compact, with $\theta\mapsto\nnp$ continuous from $\Theta_{\mathrm{ad}}$
into $C(\X;\A)$ in the supremum norm, both maintained properties of the fixed network
parametrization). If, in
addition, the policy best response $\Theta_\kappa(\widehat W)$ is
single-valued (for instance under strict convexity of the penalized policy
objective, or an entropic regularization of the selection), then, on the finite-dimensional
sieve $\calF_m$ of Assumption~\ref{ass:nested}, the set
$\calB_\kappa:=\{\widehat W\in\calF_m:\|\widehat W\|_{L^2(\muK)}\le\|g\|_\infty\}$ is
compact and convex and $T_\kappa$ maps it into itself (projection onto a convex set containing
the origin is non-expansive and $\|Q^{\pi}\|_{L^2(\muK)}\le\|Q^{\pi}\|_\infty\le\|g\|_\infty$),
so by Brouwer's theorem $T_\kappa$ admits a fixed point in $\calB_\kappa$. Any fixed point $\widehat W_\star$ of $T_\kappa$ yields
$(\pi_{\theta(\widehat W_\star)},\widehat W_\star)$, an exact $\kappa$-equilibrium
world model (exact in the sense that both the world arm and the policy arm are solved at their
population optimum on $\muK$, not that the structural residual is zero; cf.\
Definition~\ref{def:kewm}), and, under the same single-valued best response, conversely every
exact $\kappa$-EWM arises this way. Here
$\muK=\muK(\bar\theta)$ is held at the Polyak target throughout, so
$\Pi^{L^2(\muK)}_{\calF_m}$ acts on a fixed Hilbert space; the self-consistency
requirement $\bar\theta=\theta(\widehat W_\star)$ is imposed by the outer Polyak loop
and is not part of the fixed-point statement. Such a fixed point is a coverage-confirmed fixed point: a
projection-best-response fixed point on the imagined measure $\muK$, the perception
best-in-class on $\muK$ and the policy consistent with it, limited by coverage rather than by a
restricted perception class. At zero coverage ($\muK=\muerg$), and provided the outer Polyak loop has also reached
self-consistency $\bar\theta=\theta(\widehat W_\star)$ so that the frozen measure is the
fixed-point policy's own ergodic measure, $\muerg=\mu^{\pi_{\theta(\widehat W_\star)}}$, and
provided in addition the perception gap is closed on that path
($\eta_\kappa(\theta)=0$, so $\widehat W_\star=Q^{\pi_{\theta(\widehat W_\star)}}$ and $R^H=0$
$\muerg$-a.s.), it is a
self-confirming equilibrium on the policy-induced measure in the strict sense of
\citet{fudenberg1993selfconfirming} and \citet[Ch.~6]{sargent1999conquest}, because there the
imagined measure is the one the agent's own behavior generates and the belief is confirmed on the
events that behavior visits; without that outer consistency
the theorem is a fixed-measure projection-best-response result, and for $\kappa>0$ the measure is
the modeler's, so the object is coverage-confirmed rather than self-confirming in that strict
sense. When the perception class
$\calF_m$ is itself restricted, it coincides instead with the restricted-perceptions
equilibrium of \citet{BranchEvans2006RPE}, the special case $\TV(\muK,\muerg)\to0$.
\end{theorem}

\noindent At zero coverage this coverage-confirmed fixed point reduces to the self-confirming
pathwise gap of Section~\ref{sec:gap}; the binding limitation throughout is coverage rather than
the perception class, the perception-limited self-justified equilibrium of
\citet{KUBLER2023105707,10.1093/jeea/jvaf062} (Section~\ref{sec:gap}) being the complementary
restriction that EWM does not address.

\begin{remark}[Population object, not an optimization guarantee]
Theorem~\ref{thm:br} characterizes the idealized population object: a fixed
point of the projection--best-response map over a closed convex perception class, under
a single-valued best response. The neural network implementation replaces these convex metric
projections with non-convex network fits computed by SGD. The
theorem does not assert that the descent converges to this fixed point, nor that the
trained network attains the population projection; it describes the equilibrium the
construction targets, not the optimization that approximates it. Read for a generic
deep-learning solver, it is best understood as a statement about convex sieve classes,
to which the network is the practical, non-convex surrogate. The fixed point is moreover
characterized with $\muK$ held at the Polyak target; a fixed-point theory for the full
coupled map that moves the policy, the continuation, and the coverage measure jointly is
left open.
\end{remark}

\subsection{Enriched imagination: rational expectations in the limit}

\noindent In words: as coverage fills the reachable set, the solver's fixed points become
rational-expectations equilibria.

\begin{theorem}[Conditional consistency of accumulation-point policies]\label{thm:consist}
Let $(\pi^\kappa,\wm^\kappa,\muK,\calF_{m_\kappa})_{\kappa\ge1}$ be exact
$\kappa$-equilibrium world models indexed by a joint schedule of coverage reach $\kappa$ and
perception capacity $m_\kappa$, both nondecreasing. Suppose
Assumptions~\ref{ass:regular} and \ref{ass:nested} hold and, along the homotopy,
(i) the realized perception gap vanishes, $\eta_\kappa(\theta^\kappa)\to0$, a condition imposed
along the joint coverage-capacity schedule and not implied by the capacity limit
$m_\kappa\to\infty$ alone, since the continuation target $Q^{\pi^\kappa}$ and the evaluation
measure $\muK$ both move with $\kappa$;
(ii) the surrogate-fit error vanishes,
$\E_{\muK}[\|\hat q^{\pi^\kappa}_\kappa(x)-\wm^\kappa(x)\|^2]\to0$; (iii) the achieved
policy residual vanishes,
$\mathrm{OptErr}_\kappa:=\E_{\muK}[\|\rwm(x,\pi^\kappa(x),\wm^\kappa(x))\|^2]\to0$;
(iv) the coverage measures converge weakly to a limit (the coverage-reach limit):
$\muK\Rightarrow\mu_\infty$ (weak convergence; \citealp{billingsley1999}), whose support is the
coverage-certified region $\calR:=\operatorname{supp}\mu_\infty$ of Assumption~\ref{ass:nested}
(full support there holding by that definition); and (v) the
admissible policies along the homotopy
share a common modulus of continuity and take values in the compact $\A$, so the
class is relatively compact in $C(\X;\A)$ and closed under uniform limits
\citep[the Arzel\`a--Ascoli theorem,][Theorem~7.25]{rudin1976}. Then for any accumulation subsequence
$\pi^{\kappa_j}\to\pi^\infty$ uniformly,
\[
\rwm\bigl(x,\pi^\infty(x),Q^{\pi^\infty}(x)\bigr)=0
\qquad\text{for every }x\in\calR
\]
(equivalently, $\mu_\infty$-almost surely; continuity of the limiting residual and full support of
$\mu_\infty$ on $\calR$ upgrade the almost-sure statement to every point of the reachable region).
Because $\mu_\infty$ has full support on the coverage-certified region $\calR$, $\pi^\infty$
satisfies the exact equilibrium condition everywhere that region reaches: on its own ergodic path
and on every audited off-ergodic region (the disaster and impulse-response sets), not merely on the
path it realizes. It is thus a rational-expectations equilibrium on the certified region $\calR$;
enriching the coverage design expands $\calR$ toward the full structurally reachable set, of which
this is the constructive limit. The result is stated for the structural model under the maintained
quadrature rule (Assumption~\ref{ass:regular}, with $\calE$ the finite node set); the integration
error of that rule relative to the continuous-shock model is a separate numerical-approximation
error, controllable by refining the rule, and is not part of this statement.
\end{theorem}

\noindent Theorems~\ref{thm:br} and \ref{thm:consist} are the
$\SCE\!\to\!\RE$ transition stated precisely. A finite-$\kappa$ solver certifies
a coverage-confirmed fixed point on the measure it imagines (self-confirming in the strict sense
only at $\kappa{=}0$); as imagination expands
so that condition~(iv) holds, the imagined support fills the coverage-certified region $\calR$,
accumulation points satisfy the exact residual everywhere that region reaches,
i.e.\ rational expectations. The
homotopy is the constructive bridge and $\kappa$ indexes the coverage reach; here weak
convergence, not total-variation convergence, is the operative limit mode, so as in
Proposition~\ref{prop:decomp} the held-out residual on dominated audit laws, not $\TV$ itself,
is the progress diagnostic. This is the operational,
computational sense in which disciplined imagination upgrades a self-confirming
equilibrium to rational expectations; it is not a behavioral theory of how agents
form expectations.

\begin{remark}[What the limit argument does and does not use]
Weak convergence of the coverage measures, condition~(iv), does not by itself
control the exact residual: the integrand changes with $\kappa$. The proof of
Theorem~\ref{thm:consist} therefore couples weak convergence with two further
ingredients, both supplied by the maintained assumptions: the residual converges
uniformly in the state along the policy subsequence, which uses the
compactness of the admissible policy class (condition~(v), Arzel\`a--Ascoli) and
the continuity of $g,\Gamma,\rwm$ on the compact state space; and the conclusion is
extended from $\mu_\infty$-almost-everywhere to the whole reachable region using
the continuity of the limiting residual and the full support of $\mu_\infty$.
Dropping any of these weakens the conclusion from ``everywhere reachable'' to
``$\mu_\infty$-almost everywhere'' or removes it altogether. The qualifier that the modulus is
common in condition~(v) is essential here: it is an equicontinuity hypothesis on the whole
admissible policy family, not merely continuity of each policy separately, and is what makes the
continuation-comparison modulus uniform along the moving homotopy. In the implementation the
practical counterparts of this hypothesis are architectural: bounded $\tanh$ activations and the
gradient-norm clipping of Table~\ref{tab:hyper} keep the trained policies' Lipschitz constants
controlled along the schedule, which is how the equicontinuity invoked here is maintained in
practice rather than assumed away; unbounded growth of weights with capacity or training length
would be the failure mode the hypothesis rules out.
\end{remark}

\paragraph{The distributional case.} The same support logic applies when the state is itself a
distribution. In heterogeneous-agent (mean-field) economies the path support is the realized
population path and coverage means off-equilibrium populations; the off-path object, the
continuation at counterfactual cross-sections, is then a genuine economic primitive, the
value the master equation already defines, rather than a numerical nicety. The distributional
version of this reading, including the master-equation interpretation and the
never-learned exact push-forward, is developed in Section~\ref{sub:bewley}.

\section{Numerical Experiments}\label{sec:exp}

We demonstrate EWM through a sequence of test cases that isolate the main pathologies of classical
deep-learning solvers and then scale them to richer economies. The Brock--Mirman laboratory of
Section~\ref{sub:worked} isolated the coverage gap against a transparent benchmark; this section
reports the two production experiments, each a deliberate
stress test of one of the two services EWM adds to a conventional solver: (A) integration, amortizing the costly continuation
expectation with an audited surrogate while imposing the equilibrium conditions on a broader,
model-generated set of states; and (B) perception, compressing a high-dimensional or
distribution-valued state into a learned summary the policy can act on. Across both, the through-line
is reliability. A classical pathwise solver computes a self-confirming solution: accurate only along
the path it simulates, fragile to the random seed, and silently uncertified in the rare, off-ergodic
regions. We show that EWM visibly lifts these pathologies, converging from arbitrary starts to a
solution that holds off the simulated path and not only along it. We solve both ab initio,
assuming no prior knowledge of the solution and using no pre-training; where a published or analytical
solution exists we compare against it, but the solver never relies on one.

The first economy is an international real business cycle (IRBC) model with irreversible investment
and a rare, persistent disaster, a workhorse of quantitative macroeconomics. The disaster and the
binding irreversibility constraint create a genuinely hard, off-ergodic region a pathwise solver
rarely visits. We choose the IRBC because its state grows meaningfully with the number of interacting
countries: we solve it for up to $N{=}32$ countries (a $65$-dimensional state), with nothing in the
method capping $N$, to see how EWM performs in a highly nonlinear, high-dimensional equilibrium. This
is where we study the coverage construction and the continuation surrogate (the world component)
head-on, together with robustness across random initializations and the compute it costs. The second
economy is a heterogeneous-agent Bewley economy, a workhorse for questions about inequality and about
who bears aggregate risk. Its aggregate state is the entire cross-sectional distribution of wealth, a
continuous, in-principle infinite-dimensional object with no closed form, which we represent
numerically as a fine histogram. Here the EWM component under test is the distributional encoder and
its JEPA training objective (Section~\ref{sub:variants}), layered on top of the coverage and surrogate
machinery the IRBC already exercises. The two economies are different in kind, not in difficulty: one
grows the state through a discrete, increasing number of agents, the other through a continuous
distribution.

Section~\ref{sub:protocol} fixes the common setup for these experiments, the solvers compared, the
accuracy measure, and the compute accounting; Section~\ref{sub:irbc} then studies systematically how
EWM performs on the rare-disaster IRBC model, and Section~\ref{sub:bewley} does the same on the Bewley
economy.

\subsection{Experimental setup}\label{sub:protocol}

Both economies are solved ab initio from many independent random seeds (ten unless stated). The only
reported error is the exact held-out residual $R^{H}$ of \eqref{eqn:residual}, the relative error in
the model's own equilibrium conditions on states the policy never trained on, including the rare and
counterfactual regions; cost is a count of the expensive exact-quadrature evaluations. Success is
defined in advance and separately for each experiment: a seed counts only if its resulting policy
passes that experiment's exact held-out certificate, which is verified stationarity with the
disaster-region residual below tolerance in the IRBC exercise, and the RE-internal stationary
cross-sectional audit in the heterogeneous-agent (Bewley) exercise, so what we report is reliable
success from arbitrary seeds rather than one tuned run. To separate the coverage lever from the
surrogate we use the three-rung ladder of Section~\ref{sub:worked}: \texttt{DEQN-path-exact} (the
pathwise baseline, which exposes the coverage gap), \texttt{DEQN-coverage-exact} (the surrogate-free
control), and \texttt{EWM-coverage-surrogate} (the full construction, coverage plus the audited
continuation surrogate). Full metric definitions, the seed and compute protocol, the reproducibility
and provenance statement, and the proposition-to-experiment map are in Appendix~\ref{app:protocol}.

\subsection{International real business cycle with rare disasters}\label{sub:irbc}

\paragraph{Empirical contract.}
The IRBC model is the paper's headline economy. It tests the three properties the abstract promises
at once: off-path certification in the hard region, robustness across random seeds, and affordability
as the state dimension grows. Its two frictions, irreversible investment and a rare, persistent
disaster, carve out an off-ergodic region the simulated path almost never visits. There the
irreversibility constraint binds with a kink that a smooth on-path policy cannot extrapolate to, and
that is exactly where a global rational-expectations solve is genuinely tested. Read through
Section~\ref{sub:coverage}, the experiment makes the sieve from self-confirming equilibrium to
rational expectations quantitative. The pathwise baseline settles into a self-confirming solution:
accurate on its own path, but an order of magnitude worse in the disaster it never imagines. EWM
instead imposes the same exact residual on the disaster coverage and certifies the result by the
exact held-out residual the policy never trained on. Two controls isolate the mechanism: a pathwise
baseline that exposes the self-confirming gap, and a surrogate-free coverage control that imposes the
exact residual on at least as much off-path mass as EWM, isolating the coverage lever. All reported numbers are exact held-out residuals $R^{H}$. The four results below
follow the abstract's order; the supporting amortization scaling, disaster-pricing certificate, and
surrogate-capacity homotopy are in Appendices~\ref{app:irbc_supp} and~\ref{app:highn}.

\paragraph{The model.}
We solve an asymmetric $N$-country real business cycle model with irreversible investment and a
rare, persistent disaster. The state $x=(k_1,\dots,k_N,\,z_1,\dots,z_N,\,d)$ collects capital and
log-productivity (total factor productivity, TFP) per country and a binary disaster indicator
$d\in\{0,1\}$; output is $Y_j=a_j k_j^{\zeta}$ with disaster-scaled productivity $a_j=e^{z_j}(1-b\,d)$
and a Barro-style depreciation jump ($\delta(0){=}0.01\!\to\!\delta(1){=}0.06$). The equilibrium
policy is a neural network $a=\nnp(x)=\nnet_\theta(x)$ returning the per-country controls
(next-period capital $k'_j$, investment $i_j$, the irreversibility multiplier $\mu_j\ge0$) and the
common marginal utility $\lambda>0$ on the world resource constraint, the same learned map as in the
Brock--Mirman example (Section~\ref{sub:worked}), now vector-valued. The equilibrium conditions are
the standard first-order conditions of the IRBC model, extended to a setting with rare disasters; the
derivation for the regular model is in \citet{brumm2017using} and the references therein. Dropping
time subscripts (a prime denotes next period), the equilibrium is the system
\begin{align}
\text{Euler}_j:&\quad
\lambda\,(1+\phi g_j)-\mu_j-\beta\,\Qpi_j(x)=0, &&j=1,\dots,N,
\label{eqn:irbc_euler}\\
\text{Clearing:}&\quad \textstyle\sum_{j}\bigl(Y_j-i_j-\mathrm{AC}_j-c_j\bigr)=0,
\label{eqn:irbc_arc}\\
\text{KKT}_j:&\quad i_j+\mu_j-\sqrt{i_j^2+\mu_j^2}=0, &&j=1,\dots,N,
\label{eqn:irbc_fb}
\end{align}
where $g_j=k'_j/k_j-1$ is the investment gap, $\mathrm{AC}_j=\tfrac{\phi}{2}k_j g_j^2$ the convex
adjustment cost, consumption follows the intratemporal rule $c_j=(\lambda/\tau_j)^{-\gamma_j}$, the
last line is the Fischer--Burmeister encoding of the irreversibility complementarity
$i_j\ge0,\ \mu_j\ge0,\ i_j\mu_j=0$, and
\begin{equation}
\Qpi_j(x)=\E\!\left[\lambda'\bigl(\mathrm{MPK}'_j+1-\delta(d')+\mathrm{adj}'_j\bigr)
-\bigl(1-\delta(d')\bigr)\mu'_j \;\middle|\; x,a\right]
\label{eqn:irbc_Q}
\end{equation}
is the per-country continuation, the expected discounted marginal value of an extra unit of installed
capital, taken over next-period TFP ($z'_j=\rho z_j+\sigma_j\varepsilon'_j+\sigma_a\varepsilon'_a$)
and the two-state disaster chain ($\Pr(d'{=}1\mid d{=}0)=p_d$). Equations
\eqref{eqn:irbc_euler}--\eqref{eqn:irbc_fb} are an instance of the general residual
\eqref{eqn:residual}; \eqref{eqn:irbc_Q} is the expensive continuation the audited surrogate $\wm$
stands in for. The network (identical across all solver arms), the calibration, and the
framework-to-IRBC object map are in Appendix~\ref{app:hyper} (Tables~\ref{tab:hyper}
and~\ref{tab:irbc_map}).

\paragraph{Constructing EWM here.}
DEQN and EWM share the residual \eqref{eqn:irbc_euler}--\eqref{eqn:irbc_fb}, an identical policy
head, and softplus feasibility, so the admissibility penalty is off ($\omega\calK\equiv0$) and the
policy is identified by the structural residual alone. EWM differs from DEQN in exactly two places,
each an instance of Section~\ref{sec:framework}: where the residual is imposed, and how the
continuation is evaluated.

The first is the coverage measure. EWM imposes the residual on
$\muK=\rho_1\muerg+\rho_2\mu^{\mathrm{stress}}_\kappa+\rho_3\mu^{\mathrm{local}}_\kappa$ of
\eqref{eqn:muK} rather than on the ergodic path $\muerg$. This measure is the policy's own simulated
path, augmented with a fixed fraction of those states with the disaster forced on ($d{:=}1$, the
off-ergodic region a pathwise solve almost never visits) and a further fraction given a small
repaired perturbation in $(k,z)$. Every component lives on states the exact transition $\Gamma$ can
reach, and the repair bounds the simulation only, never the residual; the replicable specification
is in Appendix~\ref{app:irbc_cov_spec}.

The second is the continuation. EWM evaluates the expensive expectation \eqref{eqn:irbc_Q} inside
the policy loss \eqref{eqn:Lpolicy}--\eqref{eqn:Lworld} with the audited surrogate $\wm$
(Section~\ref{sub:surrogate}) in place of exact quadrature. The surrogate is trained on a
stop-gradient target equal to the exact $\Qpi$ from the $4(N{+}1)$-node monomial rule on the true
$\Gamma$, and is audited on held-out states by the gap $R^{H}-R^{L}$, which is observed to stay nonnegative
($R^{H}\le R^{L}$) in every reported run: the learned continuation is conservative and is never
exploited to lower the exact residual. We report this ordering as an audited observation, not a
structural guarantee; were it to reverse, the off-path bound of Theorem~\ref{thm:offpath} would
still control the certified residual. EWM never learns the structural transition $\Gamma$. Across the warm-started stages the coverage support is held
fixed; what advances is the surrogate's fidelity and width (Appendix~\ref{app:irbc_cov_spec}).
Algorithm~\ref{alg:ewm} collects the procedure. In short, EWM is DEQN with the same exact residual
imposed on $\muK$ rather than $\muerg$ and the continuation amortized by an audited surrogate; this
is what closes the self-confirming gap and keeps the expectation affordable as $N$ grows.

\paragraph{The headline comparison.}
Table~\ref{tab:irbc} and Figure~\ref{fig:irbc_seed_basin} report the main IRBC comparison. We run it
on the country-modular shared-weight architecture at $N{=}2$ and
$N{=}4$, ten seeds each, $3500$ episodes, $p_d{=}0.001$. Every entry is the exact $R^{H}$ on held-out
states, the median over seeds (Table~\ref{tab:irbc}). Figure~\ref{fig:irbc_seed_basin} shows the same
runs one seed at a time: each point plots a seed's on-path exact residual (horizontal axis) against
its disaster-region exact residual (vertical axis), with the $45^\circ$ line of equal on- and
off-path accuracy and filled markers for the seeds that pass verified stationarity. A seed on the
diagonal is as accurate in the disaster as on its own path; a seed above it is accurate on-path but
not in the disaster. The two reliable arms collapse onto the diagonal and fill in; the pathwise
baseline sits an order of magnitude above it and stays open, so low on-path error is not evidence of
a reliable disaster solution.

\paragraph{Pathwise accuracy is not off-path certification.}
The pathwise baseline \texttt{DEQN-path-exact} exhibits the self-confirming gap and
does not converge. Its disaster-region residual runs an order of magnitude above its on-path
residual, and zero of ten seeds pass verified stationarity at either $N$
(Figure~\ref{fig:irbc_seed_basin}). The failure is
structural, not a budget artifact: re-running the $N{=}2$ baseline at $10{,}000$ episodes,
nearly three times the canonical budget, does not repair it. At most one seed in ten scrapes the
convergence tolerance, and it does so at a disaster residual ($1.4\times10^{-2}$) still an order of
magnitude above the certified arms. The disaster residual stays an order of magnitude above the
on-path residual, and the gap between them is unchanged, so the larger budget does not move the
self-confirming failure at all.

\paragraph{Coverage closes the disaster residual; the surrogate buys reliability and cost.}
The surrogate-free coverage control isolates the coverage lever, and shows that
coverage is necessary but on its own neither cheap nor fully reliable. The
\texttt{DEQN-coverage-exact} control imposes the exact residual on the disaster coverage with no
surrogate, on at least as much off-path mass as EWM (a conservative control, not an easier support;
Appendix~\ref{app:irbc_cov_spec}); it cuts the disaster residual by an order of
magnitude and reaches $8/10$ ($N{=}2$) and $9/10$ ($N{=}4$) verified, but at the largest
in-loop exact-evaluation budget of any arm ($B_{\mathrm{policy}}$ rising from $301$ to $502$
million across $N{=}2\!\to\!4$).

\texttt{EWM-coverage-surrogate} then matches that accuracy at the lowest in-loop compute and
the best robustness. Its disaster residual equals the coverage arm's to within the seed
scatter, while verified stationarity rises to $9/10$ ($N{=}2$) and $10/10$ ($N{=}4$), and it
does so at less than half the baseline's policy-loop budget ($B_{\mathrm{policy}}$ of $77$ and
$129$ million at $N{=}2,4$), because the surrogate replaces the in-residual quadrature. The
surrogate is trained at a separable cost ($B_{\mathrm{world}}=206$ and $344$ million), and the
per-query speed-up this buys grows with the country dimension, from $43\times$ at $N{=}2$
to $53\times$ at $N{=}4$ (Table~\ref{tab:irbc_scaling}). The two budgets are separable by
design: $B_{\mathrm{world}}$ is paid in the world arm, outside the policy gradient step.
Because the continuation target moves with the policy, the surrogate is refit by alternation: its
exact target is evaluated once per episode and reused across the inner policy steps
(Algorithm~\ref{alg:ewm}), and each stage's surrogate warm-starts the next. Here $B_{\mathrm{policy}}$
is the exact cost paid inside every policy step. Even on the combined exact budget
$B_{\mathrm{policy}}+B_{\mathrm{world}}$, EWM is no more expensive than the surrogate-free coverage
arm ($77{+}206$ versus $301$ million at $N{=}2$, $129{+}344$ versus $502$ at $N{=}4$), and it reaches
the best accuracy and robustness at that total cost.

In economic terms the reliability result is the consequential one. The disaster-IRBC is the
class of model used for rare-disaster asset pricing and for constrained policy analysis, and
obtaining a single trustworthy global solution of it today is largely a matter of reseeding and
warm-starting by hand. EWM removes that manual step: a solve that succeeds from none of ten random
starts becomes one that succeeds from nearly all, so the resulting
counterfactuals are usable rather than seed-dependent.

Throughout, the audited self-confirming-gap ratio satisfies $R^{H}/R^{L}\le1$ in every reported
run: the learned continuation is conservative, never exploited to manufacture a residual the exact
model would not confirm. Read
against Proposition~\ref{prop:decomp}, the three disaster-residual rows are its three
controllable terms: the pathwise baseline carries the full $\mathrm{CoverageErr}$; the surrogate-free
coverage control removes it, leaving $\mathrm{OptErr}$; and EWM adds back only the small, audited
surrogate-fit term $\mathrm{SurrFitErr}$, so the disaster residual orders
\texttt{coverage-exact}$\,\le\,$\texttt{EWM}$\,\ll\,$\texttt{path}. Both coverage arms thus
reach the certified disaster floor (order $10^{-3}$, against the pathwise $10^{-2}$); what
distinguishes EWM is that it reaches that floor from nearly every seed and at about one-third of the
coverage control's policy-step exact-evaluation budget. At business-cycle scale the decisive EWM
properties are therefore the rise from zero to near-universal verified convergence and the lower
in-loop cost, neither of which the surrogate-free control attains.

\begin{table}[t!]\centering\small

\setlength{\tabcolsep}{4pt}
\begin{tabular}{llcccccc}
\toprule
$N$ & Arm & on-path $R^{H}$ & disaster $R^{H}$ & disaster $R^{H}_{p99}$ & verified & $B_{\mathrm{policy}}$ & $B_{\mathrm{world}}$ \\
\midrule
$2$ & \texttt{DEQN-path-exact}        & $2.5\times10^{-3}$ & $1.9\times10^{-2}$ & $3.7\times10^{-2}$ & $0/10$  & $172$ & --- \\
    & \texttt{DEQN-coverage-exact}    & $1.3\times10^{-3}$ & $1.1\times10^{-3}$ & $5.6\times10^{-3}$ & $8/10$  & $301$ & --- \\
    & \texttt{EWM-coverage-surrogate} & $1.2\times10^{-3}$ & $1.4\times10^{-3}$ & $5.8\times10^{-3}$ & $9/10$  & $77$  & $206$ \\
\addlinespace
$4$ & \texttt{DEQN-path-exact}        & $2.6\times10^{-3}$ & $1.8\times10^{-2}$ & $3.2\times10^{-2}$ & $0/10$  & $287$ & --- \\
    & \texttt{DEQN-coverage-exact}    & $1.5\times10^{-3}$ & $1.0\times10^{-3}$ & $4.7\times10^{-3}$ & $9/10$  & $502$ & --- \\
    & \texttt{EWM-coverage-surrogate} & $1.5\times10^{-3}$ & $1.4\times10^{-3}$ & $5.6\times10^{-3}$ & $10/10$ & $129$ & $344$ \\
\bottomrule
\end{tabular}
\caption{International real business cycle, robustness and accuracy at $N{=}2$ and $N{=}4$
(country-modular architecture, ten seeds, $3500$ episodes, $p_d{=}0.001$). The columns are: $N$, the
number of countries; Arm, the solver configuration; on-path $R^{H}$, the exact held-out residual on
the on-path region; disaster $R^{H}$, its median over the disaster-region held-out set; disaster
$R^{H}_{p99}$, the 99th percentile of that disaster residual; verified, the verified-stationarity
rate; $B_{\mathrm{policy}}$, the exact-evaluation budget inside policy training (millions); and
$B_{\mathrm{world}}$, the separable surrogate-training budget paid outside the policy gradient loop,
with dashes for the arms that use no surrogate. All
errors are the exact $R^{H}$ of \eqref{eqn:residual} on held-out states, median over all ten seeds
(including
unconverged ones, so the pathwise baseline's residual is a median over runs that complete but do
not verify); the disaster columns are over the disaster-region held-out set. ``verified'' is the
verified-stationarity rate, the share of seeds whose policy is time-invariant ($\sup$-norm change
$<10^{-3}$ on a fixed held-out evaluation set sampled from the coverage distribution) with disaster-region $R^{H}$ below the convergence tolerance
(Appendix~\ref{app:protocol}); $B_{\mathrm{policy}}$ is the exact-evaluation budget inside policy
training (millions). The pathwise baseline never converges ($0/10$ at both $N$); coverage repairs
accuracy but is the most expensive; EWM matches it at the lowest $B_{\mathrm{policy}}$ and the
highest verified rate, plus the separable $B_{\mathrm{world}}$ shown in the last column; its
surrogate audit gives $R^{H}/R^{L}\le1$ in every run (conservative, never exploited; observed, not
guaranteed).}\label{tab:irbc}
\end{table}

All reported IRBC numbers are exact held-out $R^{H}$ from the ab initio ten-seed protocol of
Section~\ref{sub:protocol} (full provenance in Appendix~\ref{app:protocol}).

\begin{figure}[t!]\centering
\includegraphics[width=\textwidth]{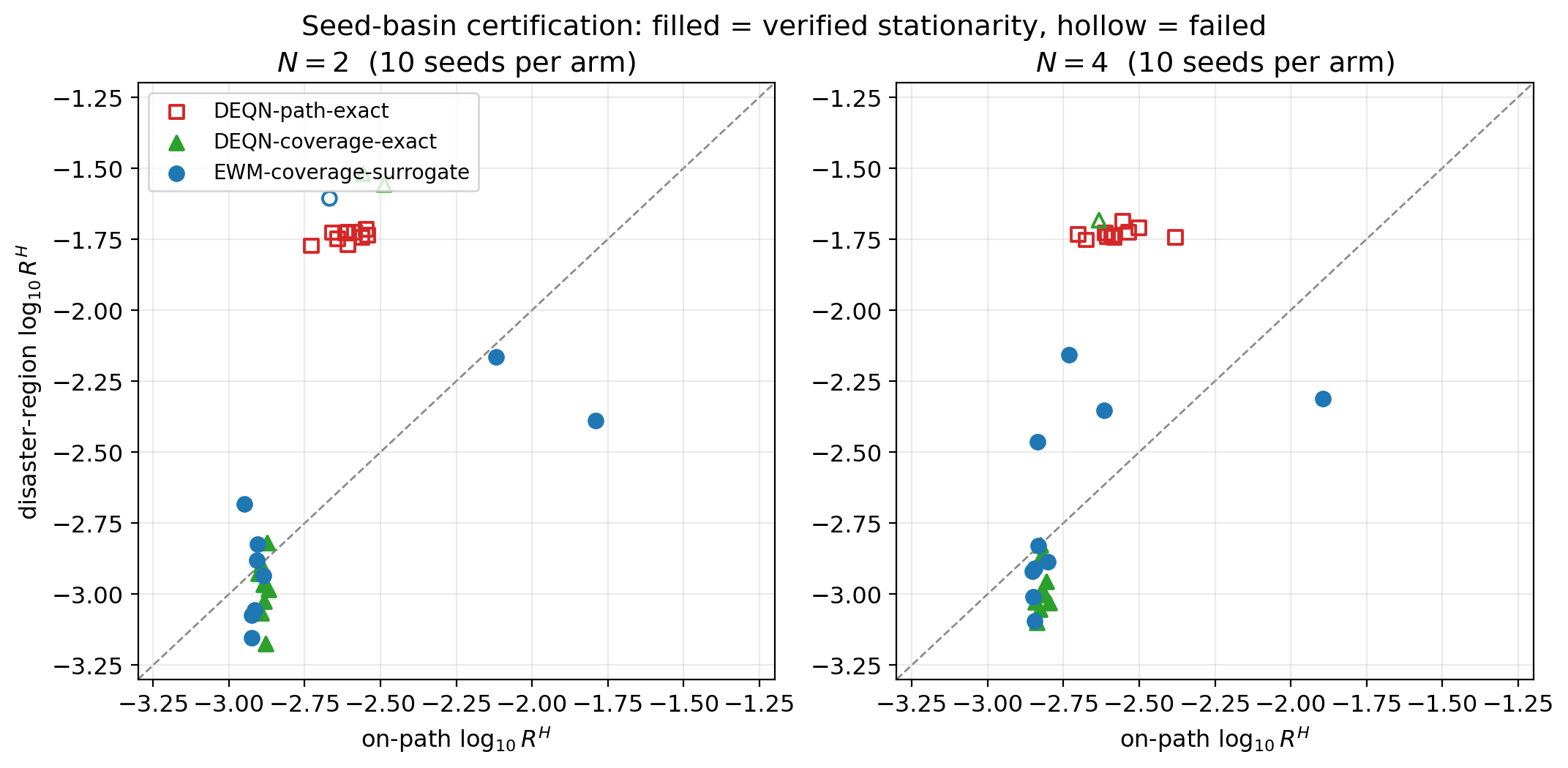}
\caption{Seed-basin certification at $N{=}2$ and $N{=}4$ (ten seeds per arm). Each point is one
seed: the on-path exact residual against the disaster-region exact residual, with the $45^\circ$
line shown and filled markers for seeds that pass verified stationarity. The pathwise baseline sits
an order of magnitude above the diagonal and never verifies; the coverage and surrogate arms
collapse onto the diagonal and verify in eight to ten seeds of ten. Low on-path residual is thus not
evidence of a reliable disaster solution. The pathwise verified rate is $0/10$ at both the
$1.0\times10^{-2}$ and the $1.5\times10^{-2}$ disaster tolerance, so the separation is a property of
the seed distribution, not of the tolerance choice.}\label{fig:irbc_seed_basin}
\end{figure}

\paragraph{The certificate changes the disaster object.}
The reliability gap is not a numerical abstraction; it changes the disaster price an economist reads
off the solution. A policy uncertified in the disaster region can deliver the wrong price of disaster
risk; a policy certified on the disaster coverage does not.

We price the one-period Arrow claim on the disaster state off each policy's own marginal utility.
Every solve whose disaster $R^{H}$ reaches the certified floor prices it in a narrow band
($1.42$--$1.47$ at $N{=}2$, $1.42$--$1.44$ at $N{=}4$), a consensus that holds across both coverage
technologies. The pathwise baseline, which reaches the floor in none of ten seeds, prices below the
band in every seed. The disaster price is therefore trustworthy exactly when the disaster residual is
certified low, which is the object the off-path audit delivers and a pathwise solve cannot;
Appendix~\ref{app:irbc_supp} gives the full pricing exhibit in Figure~\ref{fig:irbc_pricing}, which
plots the disaster Arrow-claim price against the held-out disaster residual, one point per seed and
arm.

The disaster price inherits whatever certificate the solve carries. What EWM adds is making that
certificate affordable and seed-robust: it reaches the certified floor at the lowest in-loop cost and
keeps reaching it as the state dimension grows, where the surrogate-free coverage control certifies
only at several times the in-loop cost (Appendix~\ref{app:highn}) and no external reference solve
exists.

\paragraph{Amortization scales the certificate.}
The continuation surrogate removes the repeated exact continuation evaluations from the policy loop,
which is what keeps the solve affordable as the model grows. State, policy, residual, and coverage
all scale linearly in $N$, and so does the exact continuation: the disaster-crossed monomial rule for
\eqref{eqn:irbc_Q} uses $4(N{+}1)$ exact next-state evaluations per queried state, the dominant
dimension-dependent exact cost paid inside the policy loop. The audited surrogate replaces these
repeated evaluations by one learned query, so the number of exact continuation evaluations avoided
per policy step is the node count $4(N{+}1)$ itself: it grows linearly in $N$ under the monomial rule
(the $\Theta(N)$ amortization of Proposition~\ref{prop:amort}), and faster for any richer quadrature
rule used to define the numerical model. At $N{=}32$ this is
$132$ fewer exact continuation evaluations per step, just over two orders of magnitude, realized as a
$43\times$ to $70\times$ wall-clock saving from $N{=}2$ to $N{=}32$. Table~\ref{tab:irbc_scaling} in
Appendix~\ref{app:irbc_supp} reports the full scaling, with one row per $N$ and the per-query node
count, parameter counts, and wall-clock speed-up in its columns.

At large state dimension EWM still reaches rational-expectations accuracy, and it is the only arm
that does so from every reported seed and at the lowest in-loop budget. We push
\texttt{EWM-coverage-surrogate} to $N{=}8$,
$16$, and $32$ (state dimension $17$, $33$, and $65$): at all three the held-out disaster-region
residual reaches $R^{H}\approx3\times10^{-3}$ and every seed passes verified stationarity (two seeds
at $N{=}8$, five each at $N{=}16$ and $N{=}32$), at a near-constant modular parameter count
($\approx2.9\times10^4$). Accuracy is flat as the state dimension doubles twice,
$17\!\to\!33\!\to\!65$ (median disaster $R^{H}$
$3.1\times10^{-3}\!\to\!3.3\times10^{-3}\!\to\!3.0\times10^{-3}$), the signature of a solver whose
cost, not whose accuracy, absorbs the dimension. The two controls do not keep up: the pathwise
baseline retains its self-confirming signature at every dimension (on-path
$R^{H}\approx3\times10^{-3}$ but disaster-tail $\approx1.8\times10^{-2}$, verified $0/1$), and the
surrogate-free coverage control, while it reaches the same floor at every dimension, does so at
roughly four times EWM's in-loop budget, and not from every seed at the intermediate $N{=}16$ ($3/5$,
against $3/3$ at $N{=}8$ and $N{=}32$). EWM is thus the
only configuration that reaches rational-expectations accuracy from every seed into a
$65$-dimensional state where no independent reference solve is feasible, accuracy certified by the
held-out exact $R^{H}$, verified stationarity, and the surrogate audit (Table~\ref{tab:irbc_highn_main};
the full three-arm certification is in Appendix~\ref{app:highn}, Table~\ref{tab:irbc_highn}).
These high-dimensional models carry the economics of interest: questions of international risk sharing
and how a rare disaster in one region spills across many are precisely the cases a pathwise solver
leaves uncertified, and EWM delivers a certified global solve of them.

\begin{table}[t!]\centering\small
\setlength{\tabcolsep}{6pt}
\begin{tabular}{lccccc}
\toprule
$N$ & state dim & verified & disaster $R^{H}$ & $B_{\mathrm{policy}}$ & exact-node saving \\
\midrule
$8$  & $17$ & $2/2$ & $3.1\times10^{-3}$ & $232$ & $36\times$  \\
$16$ & $33$ & $5/5$ & $3.3\times10^{-3}$ & $439$ & $68\times$  \\
$32$ & $65$ & $5/5$ & $3.0\times10^{-3}$ & $851$ & $132\times$ \\
\bottomrule
\end{tabular}
\caption{High-dimensional certification of \texttt{EWM-coverage-surrogate} at $N{=}8$, $16$, and
$32$ (state dimension $17$, $33$, $65$; country-modular architecture, $3500$ episodes). The columns
are: $N$, the number of countries; state dim, the state dimension; verified, the
verified-stationarity count of seeds; disaster $R^{H}$, the held-out disaster-region exact residual;
$B_{\mathrm{policy}}$, the in-loop policy exact-evaluation budget; and exact-node saving, the
per-query continuation evaluations $4(N{+}1)$ the audited surrogate eliminates. Every seed
passes verified stationarity (two seeds at $N{=}8$, five each at $N{=}16$ and $N{=}32$) and the
held-out disaster-region $R^{H}$ stays at its floor as the state dimension doubles twice; the
exact-node saving is the per-query continuation evaluations $4(N{+}1)$ the audited surrogate
eliminates. At these dimensions no external reference solve is feasible, so rational expectations is
certified by the exact held-out $R^{H}$, verified stationarity, and the surrogate audit, not against
an outside solution. The full three-arm comparison (with the surrogate-free and pathwise controls and
the separable $B_{\mathrm{world}}$) is in Appendix~\ref{app:highn},
Table~\ref{tab:irbc_highn}.}\label{tab:irbc_highn_main}
\end{table}

\subsection{Heterogeneous agents: the Bewley economy as a distributional microscope}\label{sub:bewley}

Heterogeneous-agent models are the workhorse of modern macroeconomics for questions about
inequality and about how the distribution of wealth both shapes and responds to aggregate shocks.
They are also the hardest case for any rational-expectations solver, because what a household must
anticipate is not a handful of prices but the behavior of the entire wealth distribution, and that
distribution shifts when the aggregate regime changes. This subsection is a microscope on one new
object, the learned summary of that distribution. It makes the object legible and exactly
certified, the way Brock--Mirman makes coverage visible: among policies that all pass the
same exact equilibrium audit, it asks which summary recovers the correct cross-section. A poorly
chosen summary leaves the solver in a self-confirming equilibrium now in the space of
distributions, accurate on the populations it has already seen and miscalibrated on the ones a
regime change produces, the same off-path failure this paper studies, carried into the
cross-section.

\paragraph{Relation to the rational-expectations burden in heterogeneous-agent models.}
We do not dispute the informational burden that rational expectations places on these economies:
the aggregate state is the cross-sectional distribution, and prices need not be Markov in any
low-dimensional aggregate alone \citep{Moll2026EJ}. EWM accepts this diagnosis; its contribution is
computational. The policy conditions on a low-dimensional learned perception of the distribution,
while the population is still advanced by the exact structural law of motion and the candidate
solution is audited against the exact full-distribution rational-expectations residual. The encoder
is thus a modeler's representation device, not a behavioral claim about what households observe or
compute.

The reading is sharpest in mean-field terms. The continuation $U(x,\mu)$, defined for every
population $\mu$, is the master equation
\citep{lasry2007mean,cardaliaguet2019master,gu2024masterequations}, the rational-expectations
object: the household anticipates the equilibrium response to any population, not only the one
realized on the equilibrium path. A simulation-based solver enforces consistency only along the
realized path $\{\mu_t^\star\}$, certifying the master equation on the ergodic manifold and leaving
it under-determined off it, the distributional analogue of the self-confirming gap. Here the off-path
object is a genuine economic primitive, the value at counterfactual populations the master equation
already defines, so coverage is the difference between solving the right functional equation and the
wrong one. Coverage over off-equilibrium populations, stressed post-shock cross-sections built by the
exact distributional transition, is the world model imagining off that manifold, and the
off-equilibrium residual again measures the gap. A strategic mean-field-game version with partial
observability is left as a sequel.

This is where the world model's second service, perception, becomes decisive
(Section~\ref{sub:econ}). On the previous models the policy reads a short list of numbers and the
only approximation is the size of the network, the function-space sieve. The heterogeneous-agent
model forces a second choice: its aggregate state is the cross-sectional wealth distribution $\mu$
itself, and before the policy can act the modeler must pick a finite-dimensional summary $h(\mu)$,
the distributional sieve. This is the central computational difficulty in solving
heterogeneous-agent macroeconomics with aggregate risk under rational expectations
\citep{Moll2026EJ}: the continuation is a functional of the whole population, the master-equation
object above, and that population must be encoded to make the problem tractable. The Krusell--Smith
resolution \citep{krusell1998} is a hand-crafted sieve, a short list of moments; EWM instead
instantiates the learned encoder $h_\psi(\mu)$ of Section~\ref{sub:econ}
\citep{MaesLeLidec2026LeWM}, trained jointly with the policy, while leaving the structural object
untouched; exactly where the encoder enters, and what stays exact, is made precise below.
Replacing the hand-chosen moments with a learned, permutation-invariant summary is itself an
established idea, the generalized moments of \citet{han2023deepham} and the symmetry-exploiting
encoders of \citet{kahou2021exploiting}; what is new here is the off-path discipline imposed on it,
the exact full-distribution residual enforced on the off-equilibrium populations the simulated path
never reaches.\footnote{EWM is a modeler's computational technology that operationalizes the
rational-expectations benchmark for this model class; it does not claim that households compute
$h_\psi$, exactly as in the relation to rational expectations set out in Section~\ref{sec:intro}.}

\paragraph{The model.} We solve the published continuum-of-agents economy of
\citet{azinovicDEEPEQUILIBRIUMNETS2022} (their Appendix~A.5.1), a Bewley--Aiyagari economy
\citep{bewley1986stationary,Aiyagari} in which aggregate risk takes the form of two uncertainty
regimes rather than a rare disaster, so a regime change reshapes the whole wealth distribution, and
with it who is exposed to the shock, without any rare-disaster device. A continuum of
ex-ante-identical households trades a single bond $a$ in unit net supply, subject to the borrowing
constraint $a'\ge a_{\min}$, at the market-clearing price $p$. A household with bond holding $a$ draws
an idiosyncratic labor shock $e$, while the economy draws an aggregate shock $z$ (a finite-state
Markov chain indexing the two uncertainty regimes), so labor income is $w(z,e)$ and consumption
$c=w(z,e)+a-p\,a'$. Preferences are recursive Epstein--Zin (relative risk aversion $\sigma{=}8$,
intertemporal elasticity $1/\rho{=}0.5$, discount $\beta{=}0.95$). With $x=(a,e,z,\mu)$ the individual
state, $x'$ its successor, and $\chi$ the risk-adjusted certainty equivalent, the recursive value and
the household's optimality conditions are
\citep[A.19--A.25 of][]{azinovicDEEPEQUILIBRIUMNETS2022}
\begin{align}
V(x) &= \max_{a'\ge a_{\min}}\Big((1-\beta)\,c^{1-\rho}+\beta\,\chi\big(V(x')\big)^{1-\rho}\Big)^{\frac{1}{1-\rho}},
\qquad \chi\big(V(x')\big)=\E\big[\,V(x')^{1-\sigma}\,\big]^{\frac{1}{1-\sigma}},\label{eqn:bewley_value}\\
p\,c^{-\rho} &= \beta\,Q^{H}(x)+\nu,\qquad 0\le \nu \;\perp\; (a'-a_{\min})\ge 0,\label{eqn:bewley_resid}
\end{align}
the Euler equation, whose continuation $Q^{H}(x)=\E\big[c'^{-\rho}\big(V(x')/\chi(V(x'))\big)^{\rho-\sigma}\big]$
carries the recursive-utility adjustment $(V'/\chi)^{\rho-\sigma}$, together with the borrowing
complementarity (multiplier $\nu\ge0$). The economy closes as a functional rational-expectations
equilibrium \citep[FREE, Definition~A.2 of][]{azinovicDEEPEQUILIBRIUMNETS2022}: equilibrium functions
$\{V,a',\nu,p\}$ of the full state at which \eqref{eqn:bewley_value}--\eqref{eqn:bewley_resid} hold and
the bond market clears under the exact distributional transition
\citep[A.26--A.34 of][]{azinovicDEEPEQUILIBRIUMNETS2022},
\begin{equation}
\textstyle\int a'(x)\,d\mu = 1,\qquad \mu' = \Phi(\mu,a',z),
\label{eqn:bewley_free}
\end{equation}
where $\Phi$ is the exact Young push-forward of the population, never learned and never forecast. The
reported error $R^{H}$ stacks the Bellman, Euler, market-clearing, and complementarity errors into the
weighted-sum loss \eqref{eqn:Rsum} of their squares, exactly the DEQN loss of \citet[their
A.41]{azinovicDEEPEQUILIBRIUMNETS2022} (the complementarity in Fischer--Burmeister form,
Section~\ref{sub:worked}); $Q^{H}$ is evaluated exactly by enumeration under $\Phi$. The equilibrium
has no closed form, and $\mu$ is solved for numerically as a fine histogram of $\ledgerBewleyFullMuDim$
bins. As everywhere in this paper, $R^{H}$ is the only equilibrium error reported; the off-equilibrium
coverage states are regime-transient cross-sections, push-forwards of the ergodic population under a
regime change, the off-manifold imagination of the mean-field reading.

\paragraph{What is approximated, and what EWM changes.}
The DEQN solver approximates the four FREE functions by two networks
\citep[A.35--A.36]{azinovicDEEPEQUILIBRIUMNETS2022}, the price and the stacked
value--policy--multiplier map,
\begin{equation}
\hat p(z,\mu)\ \approx\ p(z,\mu),
\qquad
[\hat V,\hat a',\hat\nu](a,e,z,\mu)\ \approx\ [V,a',\nu](a,e,z,\mu),
\label{eqn:bewley_approx}
\end{equation}
each conditioning on an aggregate state that contains the entire cross-sectional distribution
$\mu$, in practice the raw $\ledgerBewleyFullMuDim$-bin histogram. That argument is the only place the
distribution enters the solve, and it is the only place EWM changes anything: it substitutes the
learned, permutation-invariant encoder $h_\psi(\mu)$ for the raw $\mu$,
\begin{equation}
\hat p(z,h_\psi(\mu)),
\qquad
[\hat V,\hat a',\hat\nu](a,e,z,h_\psi(\mu)).
\label{eqn:bewley_approx_ewm}
\end{equation}
Everything structural is untouched (Figure~\ref{fig:bewley_pivot}): the residual
\eqref{eqn:bewley_resid} is still imposed on the full $\mu$, the continuation $Q^{H}$ is still
evaluated exactly under the exact Young push-forward $\Phi$, and the exact $R^{H}$ on the full
population is still the only reported error. The encoder is a substitution inside the argument
of the approximated functions, not a change to the equations they must satisfy; it is the one place
the JEPA-trained world model enters this economy.

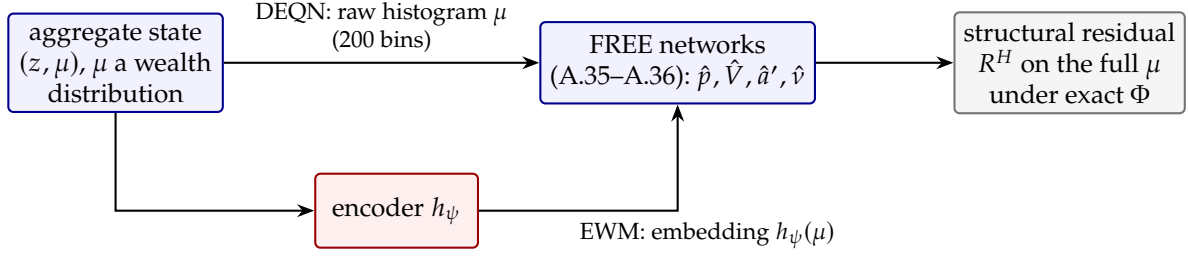
\begin{figure}[t]\centering
\resizebox{0.95\textwidth}{!}{
\begin{tikzpicture}[>=Stealth, font=\small,
  box/.style={rectangle, rounded corners=2pt, draw=blue!55!black, thick, fill=blue!5,
     minimum width=2.9cm, minimum height=1.1cm, align=center},
  enc/.style={rectangle, rounded corners=2pt, draw=red!60!black, thick, fill=red!6,
     minimum width=2.2cm, minimum height=1.0cm, align=center},
  res/.style={rectangle, rounded corners=2pt, draw=black!55, thick, fill=black!4,
     minimum width=3.1cm, minimum height=1.1cm, align=center},
  lbl/.style={font=\footnotesize, align=center},
  ar/.style={->, thick}]
\node[box] (agg) at (0,0)    {aggregate state\\$(z,\mu)$,\ $\mu$ a wealth\\distribution};
\node[box] (net) at (7.6,0)  {FREE networks\\(A.35--A.36):\ $\hat p,\hat V,\hat a',\hat\nu$};
\node[res] (res) at (12.9,0) {structural residual\\$R^{H}$ on the full $\mu$\\under exact $\Phi$};
\node[enc] (enc) at (3.8,-2.0) {encoder $h_\psi$};
\draw[ar] (agg) -- node[lbl, above]{DEQN: raw histogram $\mu$\\($\ledgerBewleyFullMuDim$ bins)} (net);
\draw[ar] (agg.south) |- (enc.west);
\draw[ar] (enc.east) -| node[lbl, below, pos=0.5, xshift=4mm]{EWM: embedding $h_\psi(\mu)$} (net.south);
\draw[ar] (net) -- (res);
\end{tikzpicture}}
\caption{What is approximated, and where the encoder enters. The Bewley solve approximates the
equilibrium (FREE) functions by networks (A.35--A.36 of
\citealp{azinovicDEEPEQUILIBRIUMNETS2022}) that condition on the aggregate state $(z,\mu)$. DEQN
feeds the raw wealth histogram $\mu$ to those networks; EWM instead feeds the learned embedding
$h_\psi(\mu)$. Only the argument is compressed: the structural residual $R^{H}$ is still
imposed on the full distribution $\mu$, and the transition $\Phi$ is exact and never learned.}
\label{fig:bewley_pivot}
\end{figure}

\paragraph{Reproduction anchor.}
Before restricting perception we confirm the substrate. Run as a plain DEQN on the raw histogram,
EWM reproduces the published Appendix-A.5 solve of \citet{azinovicDEEPEQUILIBRIUMNETS2022}, its exact
Euler, Bellman, and market-clearing residuals matching the published bounds to the same order. The
economy and the exact transition $\Phi$ are faithful before the encoder is introduced.

\paragraph{Bounded perception, full-distribution audit.}
The household policy is computationally restricted to read a compact summary $h(\mu)$, while
the modeler audits the resulting policy against the exact full-distribution rational-expectations
residual. This is a restricted-perception policy class audited against rational expectations, not a
bounded-rationality equilibrium concept: the reference distribution never enters training, and the
only equilibrium error reported is the exact $R^{H}$ on the full $\mu$ under the exact, never-learned $\Phi$. The
obstacle is not the economics but the representation of the state: the continuation $Q^{H}$ in
\eqref{eqn:bewley_resid} is exact and $\Phi$ is known in closed form, but the policy must condition on
the entire population $\mu$, a point in a $\ledgerBewleyFullMuDim$-dimensional (and in general
arbitrarily high-dimensional) simplex no network can read afresh at every residual evaluation and
still scale. This is the perception service of the world model (Section~\ref{sub:econ}), here
instantiated on the wealth distribution: a LeWorldModel encoder \citep{MaesLeLidec2026LeWM}
$h_\psi(\mu)\in\mathbb{R}^{\ledgerBewleyEWMSieveDim}$ maps $\mu$ to a finite embedding the policy can act
on, one that retains the wealth statistics equilibrium depends on (Figure~\ref{fig:linear_probe}
verifies this with a linear probe), trained jointly with the policy by a distributional JEPA
objective (predict $h_\psi(\mu')$ from $h_\psi(\mu)$ under the exact $\Phi$) with an anti-collapse
regularizer (SIGReg) that keeps the compact world-view from degenerating to an uninformative constant
(mechanics in Appendix~\ref{app:encoder}).

What we use from the world model here is therefore only the encoder, the representation of the
population: the continuation is the exact $Q^{H}$ of \eqref{eqn:bewley_resid}, so no learned value
ever enters the residual, and the encoder only sets which compression of $\mu$ the policy reads. This
is the complement of the IRBC of Section~\ref{sub:irbc}, where the world
model instead supplied the learned continuation $\wm$ that amortized the expectation: of the two
services of Section~\ref{sub:econ}, the business-cycle model exercises integration, the
heterogeneous-agent economy perception. Perception is decisive precisely where the state is itself a
distribution; the world model's perception service turns an otherwise intractable conditioning
problem into a tractable one while the economics, the residual and the transition, stays exact. As in
Section~\ref{sub:econ}, $h_\psi$ plays by discovery the role the Krusell--Smith moments play by hand,
but is discovered from the model's own transition, so it finds the coordinates of $\mu$ the dynamics
actually move; Appendix~\ref{app:encoder} works a reducible case in which the minimal sufficient
statistic is known in closed form and the encoder recovers it.

\paragraph{Three sieves, one certificate.}
This is a perception experiment, not a residual horse race: among policies that all pass the same
exact RE-internal audit at the same budget, the question is which summary of the distribution the
policy reads recovers the correct cross-section.
We compare three choices of the summary $h(\mu)$, holding everything else fixed: the same economy,
the same policy architecture except its aggregate input, the same exact residual
\eqref{eqn:bewley_resid} imposed on the full $\mu$, the same exact distributional transition $\Phi$,
the same off-ergodic coverage configuration, the same class of random initializations, and
evaluation of the converged policy. The arms differ only in the aggregate object the policy reads.
The \texttt{full-mu-raw} arm feeds in the entire histogram, all $\ledgerBewleyFullMuDim$ numbers,
with no summary at all. The \texttt{ks-moments} arm feeds in $\ledgerBewleyKSMomentsDim$ statistics
chosen by the modeler in the tradition of \citet{krusell1998}: mean wealth, the variance of wealth,
the share of households at the borrowing limit, the Gini coefficient, and the top wealth share. The
\texttt{ewm-sieve} arm feeds in a learned embedding of dimension $\ledgerBewleyEWMSieveDim$ the solver
discovers for itself, the world-model encoder $h_\psi(\mu)$ above. The raw histogram is a far
higher-dimensional input than the learned embedding, a compression the encoder pays for once rather than
at every residual evaluation (the basis of the scaling result below).

We fix the definition of ``solved'' before any result. The converged policy passes an exact,
three-part RE-internal certificate if its held-out relative Euler and complementarity residual is
below $\ledgerBewleyCertEulerThreshold$, the stationary-law $L_1$ drift under the exact $\Phi$ is
below $\ledgerBewleyCertDriftThreshold$, and the policy is feasible and non-degenerate (positive
consumption, $\nu\ge0$, and stationary mass at the borrowing constraint below
$\ledgerBewleyCertConstraintMassThreshold$). The thresholds are loose, RE-internal bars taken from
this project's existing conventions, not tuned on the reference grid. The certificate is RE-internal:
it uses the exact residual and the exact law of motion, but it does not use the reference
stationary distribution. The reference distribution enters only as an external validation
diagnostic, the $L_1$ distance of a run's stationary cross-section to an independently computed
reference (throughout, $L_1$ is the total-variation distance normalized to $[0,1]$).

\paragraph{The learned embedding selects the right cross-section among audited policies.}
Table~\ref{tab:bewley} reports the three sieves under this certificate over
$\ledgerBewleyNumSeeds$ random initializations at $\ledgerBewleyBudgetEpisodes$ training episodes.
The experiment separates certification from selection: all three sieves are certified internally by the
same RE-internal audit, but only the learned embedding selects the reference-consistent stationary
cross-section.

\begin{table}[t!]\centering\small
\setlength{\tabcolsep}{6pt}
\begin{tabular}{lcccc}
\toprule
Sieve & dim. & mean $R^{H}$ & constraint mass & $L_1$ to reference \\
\midrule
\texttt{full-mu-raw}
& \ledgerBewleyFullMuDim
& \ledgerBewleyFullMuEulerMean
& \ledgerBewleyFullMuConstraintMass
& \ledgerBewleyFullMuLoneRefMedian{} (\ledgerBewleyFullMuLoneRefMin--\ledgerBewleyFullMuLoneRefMax) \\
\texttt{ks-moments}
& \ledgerBewleyKSMomentsDim
& \ledgerBewleyKSMomentsEulerMean
& \ledgerBewleyKSMomentsConstraintMass
& \ledgerBewleyKSMomentsLoneRefMedian{} (\ledgerBewleyKSMomentsLoneRefMin--\ledgerBewleyKSMomentsLoneRefMax) \\
\texttt{ewm-sieve}
& \ledgerBewleyEWMSieveDim
& \ledgerBewleyEWMSieveEulerMean
& \ledgerBewleyEWMSieveConstraintMass
& \ledgerBewleyEWMSieveLoneRefMedian{} (\ledgerBewleyEWMSieveLoneRefMin--\ledgerBewleyEWMSieveLoneRefMax) \\
\bottomrule
\end{tabular}
\caption{Three distributional sieves under the same exact RE-internal certificate. The columns are:
Sieve, the policy-input sieve; dim., its input dimension; mean $R^{H}$, the held-out exact
Euler and complementarity residual; constraint mass, the share of mass at the borrowing constraint;
and $L_1$ to reference, the total-variation distance of the stationary cross-section to an
independently computed reference (median, with min and max in parentheses). The policy input
changes across rows; the economy, exact residual \eqref{eqn:bewley_resid}, exact transition,
coverage design, and training budget are held fixed. The certificate (held-out
Euler/complementarity residual, stationary-law drift under the exact transition, and
feasibility/non-collapse) is passed by all three arms at this budget, so it is \emph{constant by
design} and is reported once in the text rather than as a column; it is not a discriminator here.
The final column is not part of the certificate: it is an external validation diagnostic
against an independently computed reference stationary distribution, and it is where the sieves
separate.}\label{tab:bewley}
\end{table}

The numbers bear this out. On the external validation diagnostic, the $L_1$ distance of the stationary
distribution to an independently computed reference, the \texttt{ewm-sieve} embedding sits closest, the
\texttt{ks-moments} sieve markedly farther (about $\ledgerBewleyKSToEWMLoneRatio$ times the
encoder's distance) even though its mean Euler residual is comparable to or lower than the
encoder's, and the \texttt{full-mu-raw} arm, which carries all the information, collapses to
essentially the same wrong cross-section as the hand moments
($L_1\approx\ledgerBewleyRawCollapseLoneRef$). The implication is sharp: a low residual is not
sufficient to identify the correct stationary distribution. Passing the equilibrium audit does not by
itself buy the correct population; the summary the policy reads is what selects, among audited
solutions, the one that recovers it, and among the three tested representations only the learned
embedding does so. The raw-histogram outcome suggests that the relevant limitation is not the
information set itself but the numerical conditioning of the high-dimensional representation, an input
the policy cannot resolve at every residual evaluation.

The learned embedding does not win by being the only arm that passes the equilibrium certificate; all
three do. It wins because, among policies that pass the same RE-internal audit, it is the one that
recovers the economically relevant cross-section. The hand moments are the cleanest illustration: a
comparable or lower Euler residual yet a cross-section far from the reference, the distributional
signature of a self-confirming solution (low on-path error, wrong population). The raw histogram
shows that information content alone is not enough; conditioning is. The learned embedding is a
perception device, a compact, self-discovered summary the policy can resolve, not a learned
simulator: the transition $\Phi$ and the residual $R^{H}$ remain exact, and the reference
distribution remains validation-only. That a learned summary should outperform a fixed moment list
is itself consistent with the generalized moments of \citet{han2023deepham} and the symmetry
encoders of \citet{kahou2021exploiting}; what is specific here is the off-path discipline carried
into distribution space, with the exact full-distribution residual as the judge. The embedding must also be
large enough. We fix its dimension by a capacity sweep, not by tuning to the reference: we raise the
dimension until the policy both passes the residual audit and reproduces the cross-section, and report
the smallest dimension that does so, $\ledgerBewleyEWMSieveDim$.

\paragraph{The embedding carries decision-relevant shape and scales.}
The learned embedding earns its place on two counts. First, it keeps the decision-relevant cross-section: a
linear probe recovers the dispersion of wealth that prices and the continuation load on, the variance
and the Gini coefficient, from the trained embedding (Figure~\ref{fig:linear_probe}), so the policy reads
the shape of the distribution that equilibrium does not pin down, not merely its level.
Second, the learned summary is what lets the solver handle large populations at all. Because the
distribution is compressed into its embedding once, rather than re-read in full at every equilibrium
check, the per-step cost grows only about linearly in the number of asset bins with the encoder
(fitted exponent $\ledgerBewleyEncoderExponent$), against super-linearly with the raw histogram
($\ledgerBewleyHistogramExponent$). In practice the encoder resolves up to $\ledgerBewleyEncoderMax$
bins, an order of magnitude beyond the $\ledgerBewleyHistMemWall$ at which the raw histogram exhausts
memory. This is exactly the cost that has made rational-expectations solutions of these models so
demanding \citep{Moll2026EJ}: the encoder imposes the same exact equilibrium condition on far larger
populations, and it does so without ever learning or approximating the transition law. The fine
resolution is not a numerical luxury. The households a coarse histogram blurs together, those at the
borrowing constraint and in the upper tail, are precisely the ones whose responses drive marginal
propensities to consume and the incidence of aggregate shocks across the distribution.

\begin{figure}[t]\centering
\includegraphics[width=0.82\textwidth]{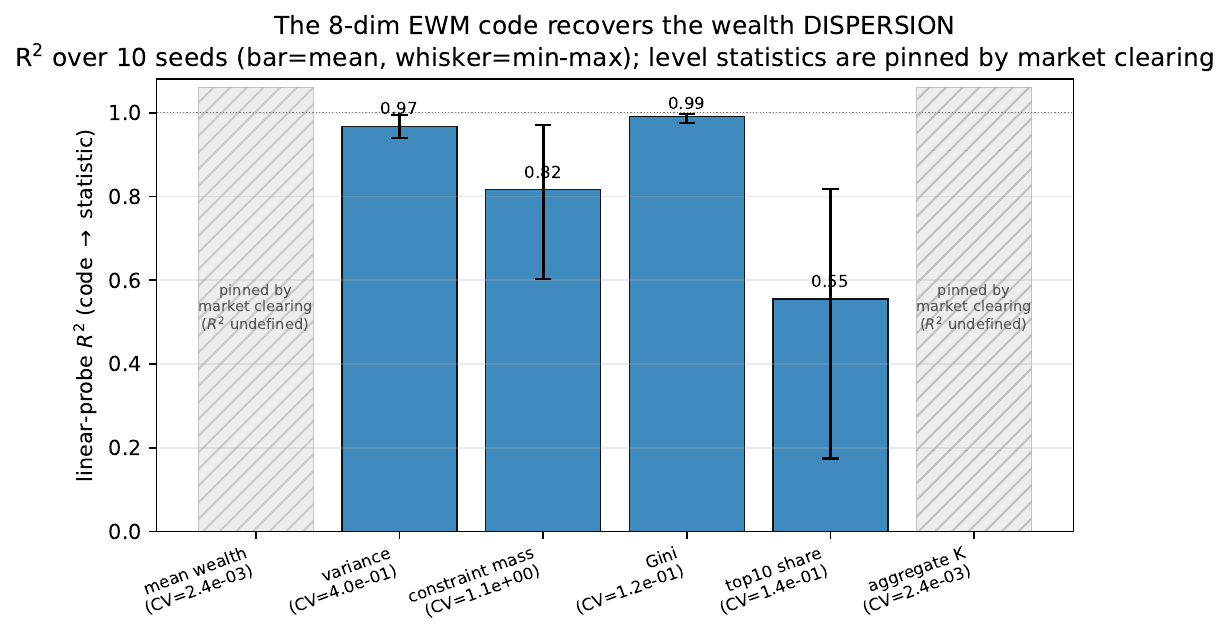}
\caption{The learned embedding keeps the decision-relevant cross-section. From the trained
$\ledgerBewleyEWMSieveDim$-dimensional EWM embedding $z=h_\psi(\mu)$ a linear probe is fit to each wealth
statistic and scored by held-out $R^2$ (bars are the across-seed mean, whiskers the min--max). The
embedding recovers the \emph{dispersion} of wealth that prices and the continuation load on: the variance
($R^2=\ledgerBewleyProbeVarRtwo$) and the Gini coefficient ($R^2=\ledgerBewleyProbeGiniRtwo$), with
the constraint mass and the top-$10\%$ share partially recovered. Mean wealth and aggregate capital
are not probed (hatched): market clearing pins them to the same value across seeds, so their $R^2$
is undefined. The level of the distribution is fixed by equilibrium; what the embedding must carry
is its shape, and it does. An analytically reducible case in which the recovered statistic is
known in closed form is in Appendix~\ref{app:encoder}.}
\label{fig:linear_probe}
\end{figure}

\paragraph{The experiments, together.} Across the three economies the discipline is one and the
same, the exact residual imposed on model-generated coverage and audited off-path, carried from a
scalar capital stock up to a full wealth distribution. Brock--Mirman makes it
visible; the IRBC scales it in aggregate dimension, where it
turns a solver that converges from none of ten random starts into one that converges from nearly all
and stays affordable to a $65$-dimensional state; and the Bewley economy carries it into the space
of distributions, where the world model's perception service supplies the summary and the learned
embedding is the only one that recovers the decision-relevant cross-section the moments and the raw
histogram miss, while scaling the exact audit to populations a histogram cannot hold. The
function-space sieve of the earlier models and the distributional sieve here are both sieves in
Grenander's sense, finite-dimensional approximations to distinct infinite-dimensional ideals, and
EWM supplies a learned instance of each under one unchanged structural residual.

\section{Conclusion}\label{sec:conclusion}

Deep-learning-based solvers have made it possible to compute dynamic stochastic economic models in state
spaces that were previously out of reach. But a global approximator is not, by itself, a global
certificate. A pathwise residual solver checks the model's equilibrium conditions only on the states
generated by its own current policy. It can therefore return a solution that is accurate on its simulated
path while remaining untested after rare shocks, near occasionally binding constraints, in the tails of
the distribution, and along the counterfactual paths for which the model is used. This paper interprets
that failure as \emph{computational self-confirmation}: the solution is confirmed on the data its own
policy generates, but not necessarily where economic analysis later reads it.

Equilibrium World Models address this problem by changing where the model's own equilibrium restrictions
are enforced and by making that broader enforcement computationally feasible. They do not change the
economic model, learn the transition law, or replace rational expectations by a learned forecasting rule.
Instead, they let the solver rehearse ordinary, rare, stressed, locally perturbed, and counterfactual
states generated or advanced through the model's maintained law of motion. The same exact equilibrium
residual is then imposed on this broader coverage region and evaluated on held-out states. Learned world
components enter only as computational devices: a continuation surrogate amortizes expensive expectations,
an action-conditioned continuation carries policy margins when actions move future risks, and a
distributional encoder compresses high-dimensional population states. None of these objects defines
equilibrium; the final certificate remains the model's exact residual.

The theory formalizes this distinction. A purely pathwise solver computes the limiting case of zero
additional coverage, a computational analogue of a self-confirming equilibrium. At finite coverage the
object is a coverage-confirmed fixed point, with the model's restrictions certified on the region the
coverage design deliberately audits. As coverage expands and the approximation, continuation, and
optimization errors vanish, accumulation points satisfy the exact equilibrium conditions on the
coverage-certified region. The result is therefore a conditional certificate, not an unconditional
guarantee: it validates the policy where the held-out exact residual is evaluated and makes no claim for
regions outside the coverage design.

The experiments show why this distinction matters. In the rare-disaster Brock--Mirman laboratory, the gain
comes from enforcing the same equilibrium conditions in disaster and post-disaster states rather than only
in normal times, and coverage reduces disaster-region residuals by about an order of magnitude. In the
endogenous-protection variant the relevant continuation must be conditioned on the action itself, because
the action changes the probability law of future disasters and a state-only continuation averages away the
margin the policy controls. In the international real business cycle model the pathwise solver has low
on-path residuals but fails the broader certification test from all reported random starts, whereas EWMs
converge from nearly all starts, remain certified in the disaster region, and reduce per-query continuation
evaluations by up to two orders of magnitude as the number of countries grows. In the heterogeneous-agent
economy the same principle extends to distributions: the policy reads a low-dimensional learned
representation of the wealth distribution, while the population is advanced by the exact distributional law
and the audit remains the exact full-distribution rational-expectations residual.

The broader lesson is that the states that matter most for macroeconomic and financial policy are often the
states that ordinary simulations visit least. A solver certified only on its own typical path does not
certify the crises, tails, constraints, and counterfactuals economists ask the model to analyze.
Equilibrium World Models make the certificate follow the question. They are disciplined dreamers: the
solver rehearses states beyond its own simulated experience, but every rehearsal stays inside the
maintained structural model and every policy is judged by the model's own equilibrium residual.

The approach also makes its own limits explicit. Coverage design is now part of the computation, and the
scope of the coverage sets the scope of the certificate. Future work should make that design more adaptive,
extend the audit to richer counterfactual objects and to continuous-time environments, and develop sharper
diagnostics for the regions where audit domination by the coverage measure fails. The aim is not to replace
structural economic discipline with machine learning, but to use modern learning methods to enforce that
discipline where it matters most.

\clearpage
\appendix
\section{Proofs and supporting results}\label{app:proofs}

This appendix provides the formal assumptions and the complete proofs for
Section~\ref{sec:theory}. We first collect the assumptions referenced throughout that section. We
then restate and prove the four standard projection and change-of-measure results relocated here for
readability, Propositions~\ref{prop:exist}, \ref{prop:mono}, and \ref{prop:surrgap} and
Lemma~\ref{lem:offpath}. The proofs of the main-text Propositions~\ref{prop:decomp},
\ref{prop:amort}, \ref{prop:fb}, and \ref{prop:girf} and of Theorem~\ref{thm:offpath} follow, and
finally the two main Theorems~\ref{thm:br} and \ref{thm:consist}, whose proofs are deferred from
Section~\ref{sec:theory}.

\subsection{Assumptions}\label{app:assumptions}

\begin{assumption}[Regular structural environment]\label{ass:regular}
$\X$ and $\A$ are compact metric spaces; for each $x$, the admissible set
$\A(x)\subseteq\A$ is non-empty with closed graph; the transition map $\Gamma$ is
continuous, hence uniformly continuous, on the compact set $\X\times\A\times\calE$
(so the induced transition kernel is weakly continuous in $(x,a)$); the kernel $g$
is continuous and bounded, so exact continuations take values in a compact set
$\mathcal{Q}\subset\R^{d_Q}$, with $Q^\pi(x)\in\mathcal{Q}$. Fix moreover a compact convex set
$\mathcal{Q}_0\subset\R^{d_Q}$ containing $\mathcal{Q}$ and the origin, and constrain every
surrogate and every sieve element used below to take values in $\mathcal{Q}_0$. The projected
continuations are finite-dimensional sieve elements and need not take values in the smaller exact
range $\mathcal{Q}$, so the residual must remain defined where the projection lands. The residual
$\rwm:\X\times\A\times\mathcal{Q}_0\to\R^r$ is continuous in all arguments and
$L_R$-Lipschitz in its continuation argument $Q\in\mathcal{Q}_0$. Being continuous on the compact
product domain $\X\times\A\times\mathcal{Q}_0$, $\rwm$ is in particular bounded,
$\|\rwm\|_\infty<\infty$. When the feasibility penalty of the policy objective is active
($\omega>0$), $\calK:\X\times\A\to\R_{+}$ is likewise continuous; in the runs of
Section~\ref{sec:exp} the penalty is off, $\omega\calK\equiv0$, and the clause is vacuous. The extension is automatic for the models we solve: the structural
residuals of Section~\ref{sec:exp} are one algebraic expression valid on all of $\R^{d_Q}$, so
$L_R$ and $\|\rwm\|_\infty$ are simply read off on the larger compact $\mathcal{Q}_0$.
\end{assumption}

\noindent Two clarifications make these primitives match the implementation. First, $\calE$ is
the finite set of quadrature nodes (Gauss--Hermite or monomial) on which the expectation is
actually evaluated, not the full support of the underlying innovation; it is therefore compact
by construction, and the underlying shock may be Gaussian or otherwise unbounded without
disturbing any compactness used below. The continuation operator
$\Qpi(x)=\sum_k w_k\,g\bigl(y_k,\nnp(y_k)\bigr)$ with $y_k=\Gamma(x,\nnp(x),\varepsilon_k)$ is then a finite weighted
sum of continuous functions, hence continuous in $(x,a)$ and weakly continuous as a kernel,
which is all the limit arguments require. The rule is an expectation rule, with nonnegative
weights summing to one ($w_k\ge0$, $\sum_k w_k=1$, as for the Gauss--Hermite rule used here), so
$\Qpi$ is a convex combination of values of $g$ and inherits the bound $\|\Qpi\|_\infty\le\|g\|_\infty$
that the self-map argument of Theorem~\ref{thm:br} uses; a signed cubature rule would replace this
constant by $(\sum_k|w_k|)\,\|g\|_\infty$, still finite, leaving that argument intact. Second, where the residual involves marginal utility
(the Brock--Mirman Euler $1/c-\beta Q$), boundedness of $\rwm$ needs consumption bounded away
from zero, $c\ge\underline c>0$, not merely $c>0$. This is not an extra assumption: a continuous
policy with softplus-positive consumption on the compact domain $\X$ attains a strictly positive
minimum, so $\underline c>0$ follows from the maintained compactness, continuity, and positivity,
and $1/c\le 1/\underline c$ is bounded; the implementation additionally enforces a floor
$c\ge\underline c$ directly in the architecture.

\noindent In the structural models of Section~\ref{sec:exp} the continuation enters the
residual affinely, so $L_R$ is not merely finite but exact, read off the equations rather
than estimated. In the marginal-cost form, $\rwm$ is linear in $Q$ with
$\partial\rwm/\partial Q=-\beta$ (the Brock--Mirman Euler $1/c-\beta Q$ and the international
real business cycle Euler \eqref{eqn:irbc_euler}), so $L_R=\beta$, the discount factor and a
contraction $L_R<1$. In the consumption-multiplied form $c(\beta Q+\mu)-1$ (the rare-disaster
Brock--Mirman family of Section~\ref{sub:bmdc} and Appendix~\ref{sub:bmdcp}),
$\partial\rwm/\partial Q=\beta c$, so $L_R=\beta\,\bar c$ with $\bar c$ the consumption bound
on the evaluation domain, and the protection block of \eqref{eq:alpha_foc} has the tighter
$L_R=\beta p_0$. The continuation-amplification constant in Proposition~\ref{prop:decomp},
Lemma~\ref{lem:offpath}, Proposition~\ref{prop:surrgap}, and Theorem~\ref{thm:offpath} is
therefore the discount factor (times a bounded consumption factor in the multiplicative case),
an exact constant, not a fitted one.

\begin{assumption}[Nested perception classes and coverage]\label{ass:nested}
The construction advances two distinct levers, which we index by two distinct symbols
throughout: the coverage reach $\kappa$, the support and mass of the coverage measure
$\muK$, and the perception capacity $m$, the width of the perception class $\calF_m$. The
perception classes $\{\calF_m\}_{m\ge1}$ are nested in capacity
($\calF_m\subseteq\calF_{m'}$ for $m\le m'$), each a finite-dimensional
closed convex set of continuous, $\mathcal{Q}_0$-valued functions on the state space (a linear,
spline, or polynomial sieve intersected with the pointwise constraint $q(x)\in\mathcal{Q}_0$, so
$\calF_m\subset C(\X;\mathcal{Q}_0)$), identified with its image in $L^2(\muK;\R^{d_Q})$ and
containing the origin; the trained network is the practical nonconvex parametrization of this
idealized class. We evaluate each class element at its continuous representative, so the pointwise
value $\widehat W(x)$ inside the residual is well-defined and lies in $\mathcal{Q}_0$, where
$\rwm$ is defined. At each fixed coverage measure $\muK$ we assume either that the finite sieve is
$L^2(\muK)$-separated (equivalently, the Gram matrix of the chosen finite basis on
$\operatorname{supp}\muK$ is positive definite), or that the projection operator is equipped with a
fixed continuous-representative tie-breaker; thus the $L^2(\muK)$ projection used below is a
single-valued map into continuous representatives. Let $S_0$ denote the reference set from which
the coverage homotopy starts: for a fixed baseline this is $\operatorname{supp}\muerg$, while along
a moving homotopy one may take $S_0=\overline{\bigcup_j\operatorname{supp}\mu^{\pi^{\kappa_j}}}$.
Along the homotopy both levers advance, stage $j$ carrying reach $\kappa_j$ and capacity $m_j$; as
reach grows the coverage measures converge weakly, $\muK\Rightarrow\mu_\infty$, to a limiting
coverage measure $\mu_\infty$. Let $\calR:=\operatorname{supp}\mu_\infty\subseteq\X$ be the
coverage-certified region: the closure of the domain the coverage schedule certifies in the limit,
the ergodic support $\operatorname{supp}\muerg$ together with the bounded structural perturbations
and stressed/disaster seeds the schedule imposes, each rolled forward finitely under $\Gamma$.
Full support, $\operatorname{supp}\mu_\infty=\calR$, holds by this definition, so the consistency
hypothesis (condition~(iv) of Theorem~\ref{thm:consist}) is the achievable statement that the
coverage design is rich enough that its limit places mass throughout the region it is meant to
certify, rather than the unachievable demand that the coverage measure fill the maximal forward
orbit from $S_0$ under all admissible actions $a_t\in\A(x_t)$, $\varepsilon_{t+1}\in\calE$
(for a policy-specific reading, the orbit under $a_t=\pi^\infty(x_t)$). Enriching the coverage
design enlarges $\calR$ toward that maximal $\Gamma$-reachable set; the theorem certifies the exact
residual on whichever $\calR$ the schedule reaches, not on states no coverage schedule visits. Capacity growth is used in Theorem~\ref{thm:consist} only through the realized perception
condition $\eta_\kappa(\theta^\kappa)\to0$; a primitive sufficient condition is density of
$\bigcup_m\calF_m$ in the relevant continuation family
$\{Q^\pi:\pi\text{ admissible in the compact policy class}\}$ under the moving evaluation
measures, not density in all of $L^2(\mu_\infty;\R^{d_Q})$, which would be incompatible with the
compact $\mathcal{Q}_0$-valued range. The existence and projection results
(Propositions~\ref{prop:exist}, \ref{prop:mono} and Theorem~\ref{thm:br}) hold the coverage
measure $\muK$ and the capacity $m$ fixed at a reference stage, so that the inner product and the
class are fixed.
\end{assumption}

\begin{assumption}[Unique invariant measure]\label{ass:invariant}
For every admissible policy $\pi$, the chain $x'=\Gamma(x,\pi(x),\varepsilon')$
has a unique invariant measure $\mu^\pi$ (a Doeblin minorization on the induced
kernel suffices \citep[Ch.~11]{stokeylucas1989}), so the ergodic training measure
$\muerg=\mu^{\nnp}$ of \eqref{eqn:pathloss} is well defined. (The consistency argument does not
separately require the ergodic measure to be dominated by $\mu_\infty$: full support of
$\mu_\infty$ on $\calR$ already certifies the limiting policy on the whole reachable region,
which contains the realized path.)
\end{assumption}

\begin{assumption}[Audit domination]\label{ass:audit}
Every audit measure $\eta$ on which the off-path residual is evaluated or bounded is
absolutely continuous with respect to the coverage measure, $\eta\ll\muK$, with
Radon--Nikodym derivative bounded uniformly across admissible $\pi$ by a constant,
$d\eta/d\muK\le B$. This is the density bound invoked in Lemma~\ref{lem:offpath} and
Theorem~\ref{thm:offpath}; it is a separate object from the weak-convergence coverage
condition (condition~(iv) of Theorem~\ref{thm:consist}, supplied by
Assumption~\ref{ass:nested}), and the two are not interchangeable.
\end{assumption}

\noindent Audit domination is a substantive restriction, not a technicality. A single
deterministic impulse-response path or disaster branch is a measure-zero set and is
not absolutely continuous with respect to $\muK$, so $B=\infty$ and the bound of
Theorem~\ref{thm:offpath} is vacuous on it. The same caution applies to a finite
deterministic grid or finite held-out sample: as an atomic empirical measure it is singular
with respect to a continuous coverage measure, so it is dominated only when it is itself drawn
from a $\muK$-absolutely-continuous law, or when $\muK$ assigns positive mass to the audited
atoms. In practice the off-path residual is reported on a tube or distribution around such a
path, sampled from the coverage support, which restores domination. In the experiments the held-out audit sets
(the disaster-region grid, the impulse-response tube) are finite Monte Carlo samples from a
population audit distribution that is, by construction, absolutely continuous with respect to
$\muK$ (it is a tube or region law supported where $\muK$ has mass); domination is a property of
that population law, not of the finite sample, and under it the reported residuals fall under
Theorem~\ref{thm:offpath}. Where domination genuinely fails, a Wasserstein control rather than
a density ratio is the appropriate substitute, which we leave open.

\subsection{Supporting projection results}\label{app:supporting}

\begin{proposition}[Existence and uniqueness of the projected world model]\label{prop:exist}
Under Assumptions~\ref{ass:regular}--\ref{ass:nested}, for every admissible policy
$\pi$ and every $\kappa$ the continuation $\Qpi$ is bounded and lies in
$L^2(\muK;\R^{d_Q})$, and there exists a unique $L^2(\muK)$ projection, represented by the continuous representative
specified in Assumption~\ref{ass:nested},
\[
\hat q^{\pi}_\kappa:=\Pi^{L^2(\muK)}_{\calF_m}\Qpi
=\argmin_{q\in\calF_m}\E_{x\sim\muK}\bigl[\|q(x)-\Qpi(x)\|^2\bigr],
\]
delivered by the Hilbert projection theorem \citep[Theorem~3.16]{bauschkecombettes2011} on the
closed convex set $\calF_m\subseteq L^2(\muK;\R^{d_Q})$ (Assumption~\ref{ass:nested}).
\end{proposition}

\begin{proof}
By Assumption~\ref{ass:regular} the kernel $g$ is bounded, so $\Qpi$ is bounded and
lies in $L^2(\muK;\R^{d_Q})$. By Assumption~\ref{ass:nested} the class $\calF_m$ is a
non-empty closed convex subset of the Hilbert space $L^2(\muK;\R^{d_Q})$, so the Hilbert
projection theorem delivers a unique $L^2(\muK)$ minimizer; the representative convention in
Assumption~\ref{ass:nested} turns this equivalence class into the continuous function
$\hat q^{\pi}_\kappa=\Pi^{L^2(\muK)}_{\calF_m}\Qpi$ used in pointwise residual evaluations.
\end{proof}

\begin{proposition}[Monotonicity at a fixed measure]\label{prop:mono}
Fix an evaluation measure $\mu$ and a policy $\pi$. Under
Assumption~\ref{ass:nested}, $m\le m'$ implies
$\inf_{q\in\calF_{m'}}\E_\mu[\|q-\Qpi\|^2]
\le\inf_{q\in\calF_m}\E_\mu[\|q-\Qpi\|^2]$: enlarging the perception class
weakly reduces the in-class approximation error at a fixed measure. When
the coverage measure itself varies with $\kappa$, this is the class-approximation
component only; the change of measure is governed separately by the coverage gap
of Proposition~\ref{prop:decomp}.
\end{proposition}

\begin{proof}
For $m\le m'$, Assumption~\ref{ass:nested} gives
$\calF_m\subseteq\calF_{m'}$, so the infimum of
$\E_\mu\|q-\Qpi\|^2$ over the larger class $\calF_{m'}$ is no larger than over
$\calF_m$.
\end{proof}

\begin{proof}[Proof of Proposition~\ref{prop:decomp}]
Write $f(x):=\|\rwm(x,\nnp(x),\Qpi(x))\|^2$, bounded by $\|\rwm\|_\infty^2$. By the
total-variation convention of Section~\ref{sec:theory},
\[
\E_{\mu}[f]\;\le\;\E_{\muK}[f]+2\|\rwm\|_\infty^2\,\TV(\mu,\muK)
=\E_{\muK}[f]+\mathrm{CoverageErr}.
\]
On $\muK$, split the residual at the current policy into three pieces,
\[
\rwm(x,\nnp,\Qpi)
=\underbrace{[\rwm(x,\nnp,\Qpi)-\rwm(x,\nnp,\hat q^{\nnp}_\kappa)]}_{A}
+\underbrace{[\rwm(x,\nnp,\hat q^{\nnp}_\kappa)-\rwm(x,\nnp,\wm)]}_{B}
+\underbrace{\rwm(x,\nnp,\wm)}_{C},
\]
and use $\|A+B+C\|^2\le 3(\|A\|^2+\|B\|^2+\|C\|^2)$. By the $L_R$-Lipschitz property
of $\rwm$ in its continuation argument (Assumption~\ref{ass:regular}),
$\|A\|^2\le L_R^2\|\Qpi-\hat q^{\nnp}_\kappa\|^2$ and
$\|B\|^2\le L_R^2\|\hat q^{\nnp}_\kappa-\wm\|^2$. Taking $\E_{\muK}$ and using that
$\hat q^{\nnp}_\kappa$ is the $L^2(\muK)$-projection of $\Qpi$ (so
$\E_{\muK}\|\Qpi-\hat q^{\nnp}_\kappa\|^2=\eta_\kappa(\theta)$),
\[
\E_{\muK}[f]\le 3\bigl(\mathrm{OptErr}+L_R^2\eta_\kappa(\theta)
+L_R^2\E_{\muK}\|\hat q^{\nnp}_\kappa-\wm\|^2\bigr)
=3(\mathrm{OptErr}+\mathrm{ApproxErr}_\kappa+\mathrm{SurrFitErr}).
\]
Combining the two displays gives the bound.
\end{proof}

\begin{lemma}[Surrogate-induced off-ergodic discrepancy]\label{lem:offpath}
Let $\eta$ be a measure generated by an impulse-response or large-shock
experiment with $\eta\ll\muK$ and $d\eta/d\muK\le B$. Then
\[
\E_\eta\!\bigl[\|\rwm(x,\nnp(x),\wm(x))-\rwm(x,\nnp(x),\Qpi(x))\|^2\bigr]
\;\le\;
L_R^2\,B\,\E_{\muK}\!\bigl[\|\wm(x)-\Qpi(x)\|^2\bigr].
\]
\end{lemma}

\begin{proof}
Pointwise, $\|\rwm(x,\nnp,\wm)-\rwm(x,\nnp,\Qpi)\|^2\le L_R^2\|\wm-\Qpi\|^2$ by the
Lipschitz property. Integrating against $\eta$ and changing measure by the density
bound $d\eta/d\muK\le B$,
\[
\E_\eta\!\bigl[\|\rwm(x,\nnp,\wm)-\rwm(x,\nnp,\Qpi)\|^2\bigr]
\le L_R^2\,\E_{\muK}\!\Bigl[\tfrac{d\eta}{d\muK}\,\|\wm-\Qpi\|^2\Bigr]
\le L_R^2 B\,\E_{\muK}\|\wm-\Qpi\|^2.
\]
\end{proof}

\begin{proposition}[Surrogate--exact residual gap]\label{prop:surrgap}
Under Assumption~\ref{ass:regular} ($\rwm$ is $L_R$-Lipschitz in its continuation
argument), $\|R^{H}(x)-R^{L}(x)\|\le L_R\,\|\wm(x)-\Qpi(x)\|$ pointwise, and hence on
any measure $\mu$
\[
\E_\mu[\|R^{H}\|^2]\;\le\;2\,\E_\mu[\|R^{L}\|^2]\;+\;2L_R^2\,\E_\mu[\|\wm-\Qpi\|^2].
\]
\end{proposition}

\begin{proof}
Pointwise $\|R^{H}-R^{L}\|\le L_R\|\wm-\Qpi\|$ by the Lipschitz property; with
$\|R^{H}\|^2\le 2\|R^{L}\|^2+2\|R^{H}-R^{L}\|^2$ and integrating against $\mu$ gives
the stated bound.
\end{proof}

\begin{proposition}[Amortization scaling of the continuation]\label{prop:amort}
Suppose the exact continuation $\Qpi(x)=\E[g(x',\nnp(x'))\mid x]$ is evaluated by the
fixed quadrature rule with $n_q(N)$ nodes, each node costing one exact transition and one
policy forward pass, while a surrogate query $\wm(x)$ is a single forward pass whose cost
is independent of $n_q(N)$. The comparison is in exact-evaluation counts; any scaling of the
network forward pass with the input dimension is common to the surrogate and to the policy
evaluations inside the quadrature, and is not part of the $n_q(N)$ amortization factor. Then the per-state cost of evaluating the continuation exactly,
relative to one surrogate query, is $\Theta(n_q(N))$: the exact rule touches all $n_q(N)$
nodes, whereas the surrogate replaces the whole expectation with one forward pass. For the
monomial rule used in the IRBC, $n_q(N)=4(N+1)=\Theta(N)$, so the amortization available per
continuation evaluation grows linearly in the shock/country dimension $N$.
\end{proposition}

\begin{proof}
An accounting identity. By hypothesis one exact continuation query evaluates the fixed
quadrature rule at $n_q(N)$ nodes, each node costing one exact transition and one policy forward
pass, so its cost is $\Theta(n_q(N))$ exact evaluations; one surrogate query is a single forward
pass, cost $\Theta(1)$ independent of $n_q(N)$. The ratio is $\Theta(n_q(N))$. For the monomial
rule used in the IRBC, $n_q(N)=4(N+1)=\Theta(N)$, giving the stated linear scaling.
\end{proof}

\begin{proposition}[Complementarity and occasionally-binding constraints]\label{prop:fb}
Write the model's KKT conditions in Fischer--Burmeister form
\citep{fischer1992special},
$\phi(a,b)=a+b-\sqrt{a^2+b^2}=0$, which encodes $a\ge0$, $b\ge0$, $ab=0$. The map
$\phi$ is globally Lipschitz with constant at most $1+\sqrt2$, so the
complementarity block of $\rwm$ is Lipschitz in its own argument $(a,b)$. This
constant is distinct from the continuation-argument constant $L_R$ that governs the
continuation-argument bounds of Proposition~\ref{prop:decomp}, Lemma~\ref{lem:offpath},
and Proposition~\ref{prop:surrgap} (equal to the discount factor $\beta$ in the marginal-cost
form, $\beta\bar c$ otherwise; see the remark after Assumption~\ref{ass:regular}). The complementarity block
remains continuous and Lipschitz in its own arguments across the kink; regularity of the
continuation map $\Qpi$ does not follow from this Lipschitz property but from the
maintained assumptions on the transition, the policy, and the continuation primitives,
which is what keeps those projection and consistency arguments valid when the constraint
binds. Occasionally-binding constraints, the irreversible-investment kink
of the IRBC and the borrowing constraint of the heterogeneous-agent model, thus
enter the exact residual without smoothing and are covered by the same $L^2$ theory;
the kink affects differentiability of the residual, not the validity of the
projection and consistency arguments. The Lipschitz property gives a forward bound
(a small constraint error implies a small Fischer--Burmeister residual), not its
inverse. We do not establish a finite-error inverse map from Fischer--Burmeister
residuals to separate primal feasibility, dual feasibility, and complementary-slackness
violations; these are therefore distinct diagnostics, not recoverable from the stacked
residual alone, and a complete constraint audit reports them separately rather than reading
them off $R^{H}$.
\end{proposition}

\begin{proof}
Off the origin the gradient of $\phi(a,b)=a+b-\sqrt{a^2+b^2}$ is
$\bigl(1-a/r,\,1-b/r\bigr)$ with $r=\sqrt{a^2+b^2}$, whose Euclidean norm is at most
$1+\sqrt2$ (writing $a=r\cos\vartheta$, $b=r\sin\vartheta$ gives
$\|\nabla\phi\|^2=3-2\sqrt2\,\sin(\vartheta+\tfrac\pi4)\le 3+2\sqrt2=(1+\sqrt2)^2$). At the
origin $\phi$ is non-differentiable, but it is locally Lipschitz on $\R^2$ and
differentiable off that single point; since $\R^2$ is convex, the supremum of the
gradient norm over the full-measure differentiable set bounds the global Lipschitz
constant, so $\phi$ is globally Lipschitz with constant $1+\sqrt2$.
\end{proof}

\begin{proposition}[Generalized impulse-response stability]\label{prop:girf}
Fix a realization of the shock sequence and let $\{x_h\}_{h=0}^{H}$ be the conditional
impulse-response path it induces under a policy $\pi$, and $\{x_h^\star\}$ the
corresponding rational-expectations path from the same initial state and the same shock
realization; the generalized impulse response \citep{koop1996impulse} averages such conditional
paths over shock draws, and since the bound below holds for every realization whose path stays
on the audited tube, it passes to that average. Suppose a direct
audit verifies that the one-step exact residual is bounded along the path,
$\sup_{0\le h<H}\|R^{H}(x_h)\|\le\delta$, on the impulse-response tube $\calT_H$ around the path
(this pointwise bound is the held-out quantity we report; Theorem~\ref{thm:offpath} controls the
expected residual $\E_\eta\|R^{H}\|^2$ over a tube $\eta$ dominated by the coverage
measure, from which the pointwise statement follows only with high probability, via Markov's
inequality together with a union bound over the $H$ horizons when the supremum $\sup_{0\le h<H}$ is
asserted, not deterministically), and that the policy-induced
one-step map is stable with modulus $\lambda\ge0$, in the sense that
$\|x_{h+1}-x_{h+1}^\star\|\le\lambda\,\|x_h-x_h^\star\|+c\,\|R^{H}(x_h)\|$ for a
constant $c$. Then
\[
\|x_H-x_H^\star\|\;\le\;c\,\delta\sum_{h=0}^{H-1}\lambda^{h}
=\begin{cases}
c\,\delta\,\dfrac{1-\lambda^{H}}{1-\lambda}, & \lambda\neq1,\\[1.2ex]
c\,\delta\,H, & \lambda=1.
\end{cases}
\]
\end{proposition}

\begin{proof}
With $x_0=x_0^\star$, unrolling
$\|x_{h+1}-x_{h+1}^\star\|\le\lambda\|x_h-x_h^\star\|+c\|R^{H}(x_h)\|$ and
$\|R^{H}(x_h)\|\le\delta$ gives
$\|x_H-x_H^\star\|\le c\delta\sum_{h=0}^{H-1}\lambda^{h}$, which evaluates to the
geometric sum in the statement.
\end{proof}

\begin{proposition}[Identification under measure-moving actions]\label{prop:actcond}
Suppose an action coordinate $a$ enters the continuation only through the transition
measure: write $\Qpi(x,a):=\int g(x',\nnp(x'))\,P(dx'\mid x,a)$ for the continuation as a
function of a freely varying action, so that the state-only continuation of the main text is
its evaluation at the policy's own action, $\Qpi(x)=\Qpi(x,\nnp(x))$, with $P$ depending on $a$
through the known structural law $\Gamma$ (a hazard, a regime-switching probability), and the
integrand $g$ not depending on $a$. It is the two-argument object $\Qpi(x,a)$ that is
differentiated below. Assume the map $a\mapsto P(\cdot\mid x,a)$ is differentiable in total
variation: there is a signed measure $\partial_a P(\cdot\mid x,a)$ with
$\|P(\cdot\mid x,a+h)-P(\cdot\mid x,a)-h\,\partial_a P(\cdot\mid x,a)\|_{\mathrm{TV}}=o(h)$ and
$\|\partial_a P(\cdot\mid x,a)\|_{\mathrm{TV}}<\infty$ uniformly over $\X\times\A$, with
$M_a:=\sup_{(x,a)\in\X\times\A}\|\partial_a P(\cdot\mid x,a)\|_{\mathrm{TV}}<\infty$, where
$\|\cdot\|_{\mathrm{TV}}$ is the variation norm of a signed measure. This hypothesis covers a
transition with a density $p(x'\mid x,a)$ differentiable in $a$
($\partial_a P=\partial_a p\,dx'$, $\|\partial_a P\|_{\mathrm{TV}}=\int|\partial_a p|\,dx'$) and,
equally, finite-state and regime-switching laws whose mixing weights move with $a$ but which admit
no Lebesgue density; it is the natural form for the disaster-hazard model of
Appendix~\ref{sub:bmdcp}. Differentiation then passes through the integral by the duality
$|\int f\,d(\partial_a P)|\le\|f\|_\infty\,\|\partial_a P\|_{\mathrm{TV}}$ for bounded $f$. A
state-only surrogate $\wm(x)\approx\Qpi(x)$, evaluated at the
policy's own action, has integrated $a$ out before differentiation, so
$\partial\wm/\partial a\equiv0$: the continuation component of the $a$-gradient is absent from the
surrogate residual (current-period payoff terms may still depend on $a$). An action-conditioned surrogate
$\widehat Q(x,a)=\int \widehat v(x')\,P(dx'\mid x,a)$ that carries the exact dependence of
$P$ on $a$ has
$\partial\widehat Q/\partial a=\int \widehat v(x')\,\partial_a P(dx'\mid x,a)$, so it
represents the marginal-benefit term that the state-only object cannot. With the exact
integrand $g$ in place of the learned $\widehat v$ this derivative is exact; with a learned
$\widehat v$ the derivative error is controlled at the value level,
\[
\Bigl|\partial_a\widehat Q(x,a)-\partial_a\Qpi(x,a)\Bigr|
=\Bigl|\int(\widehat v-g)(x')\,\partial_a P(dx'\mid x,a)\Bigr|
\le\|\widehat v-g\|_\infty\,\|\partial_a P(\cdot\mid x,a)\|_{\mathrm{TV}},
\]
where the integrand $g$ is shorthand for $g(x',\nnp(x'))$. The required surrogate accuracy is
thus an $O(1)$ value level amplified only by the uniformly bounded mass $M_a$, so
$|\partial_a\widehat Q(x,a)-\partial_a\Qpi(x,a)|\le M_a\,\|\widehat v-g\|_\infty$ uniformly over
$\X\times\A$, not a direct action-derivative accuracy; a value-level error does not blow up at
the derivative level.
\end{proposition}

\begin{proof}
Total-variation differentiability lets differentiation pass through the integral: for bounded
measurable $f$, $\partial_a\!\int f\,dP(\cdot\mid x,a)=\int f\,d(\partial_a P(\cdot\mid x,a))$,
the remainder bounded by $\|f\|_\infty$ times the $o(h)$ total-variation remainder. A state-only
$\wm(x)$ does not depend on $a$, so $\partial_a\wm\equiv0$. For the action-conditioned object,
$\partial_a\widehat Q=\int\widehat v\,d(\partial_a P)$; subtracting
$\partial_a\Qpi=\int g\,d(\partial_a P)$ and using the duality
$|\int(\widehat v-g)\,d(\partial_a P)|\le\|\widehat v-g\|_\infty\,\|\partial_a P\|_{\mathrm{TV}}$
gives the stated value-level bound. When $P$ has a density, $\partial_a P$ has density
$\partial_a p$ and $\|\partial_a P\|_{\mathrm{TV}}=\int|\partial_a p|\,dx'$, recovering the
absolutely continuous case.
\end{proof}

\begin{proof}[Proof of Theorem~\ref{thm:offpath}]
By Assumption~\ref{ass:regular}, $\rwm$ is $L_R$-Lipschitz in its continuation
argument, so $\|R^{H}(x)-R^{L}(x)\|\le L_R\|\wm(x)-\Qpi(x)\|$ at every $x$; with the
elementary inequality $\|a\|^2\le 2\|b\|^2+2\|a-b\|^2$ (taking $a=R^{H}$, $b=R^{L}$) this
gives the pointwise bound $\|R^{H}\|^2\le 2\|R^{L}\|^2+2L_R^2\|\wm-\Qpi\|^2$. The
right-hand side is nonnegative, so taking $\E_\eta$ and using $d\eta/d\muK\le B$
(Assumption~\ref{ass:audit}) to pass from $\eta$ to $\muK$ on the right-hand side
yields the claim.
\end{proof}

\begin{proof}[Proof of Theorem~\ref{thm:br}]
By Assumption~\ref{ass:regular} the kernel $g$ is continuous and bounded, so for
every admissible policy $\pi$ the continuation $Q^{\pi}(x)=\E[g(x',\pi(x'))\mid x]$
is bounded and lies in $L^2(\muK;\R^{d_Q})$. By Assumption~\ref{ass:nested} the class
$\calF_m$ is a non-empty closed convex subset of $L^2(\muK;\R^{d_Q})$, so the
Hilbert projection theorem delivers a unique minimizer of
$\widehat W\mapsto\|\widehat W-Q^{\pi}\|_{L^2(\muK)}$ over $\calF_m$, namely
$\Pi^{L^2(\muK)}_{\calF_m}(Q^{\pi})$. Hence the world-model arm
\eqref{eqn:Lworld} is solved exactly at a policy $\pi$ if and only if
\begin{equation}
\widehat W=\Pi^{L^2(\muK)}_{\calF_m}\!\big(Q^{\pi}\big).
\tag{$\ast$}\label{eqn:proj_id}
\end{equation}
For the policy arm, under Assumption~\ref{ass:regular} the objective
$(\widehat W,\theta)\mapsto\E_{\muK}[\|\rwm(x,\nnp(x),\widehat W(x))\|^2+\omega\calK(x,\nnp(x))]$ is
jointly continuous in its two arguments: $\widehat W$ enters continuously through the
$L_R$-Lipschitz continuation argument of $\rwm$, $\theta$ through the continuous parametrization
$\theta\mapsto\nnp$, and the penalty term through the continuity of $\calK$ (vacuous in our runs,
where $\omega\calK\equiv0$), all on the compact $\X\times\Theta$ with the policy-parameter set
$\Theta$ compact, hence Hausdorff. Berge's maximum theorem, applied with $\widehat W$ as the parameter, then
makes the argmin correspondence $\Theta_\kappa(\cdot)$ non-empty, compact-valued and upper
hemicontinuous in $\widehat W$, and the Kuratowski--Ryll-Nardzewski theorem supplies a measurable
selector $\theta(\cdot)$ \citep[Theorems~17.31 and 18.13]{aliprantisborder2006}.
Thus $\theta(\widehat W)$ solves the surrogate-continuation policy arm
\eqref{eqn:Lpolicy} given $\widehat W$.

\textbf{Existence.} Suppose $\Theta_\kappa$ is single-valued; then, $\Theta$ being compact
Hausdorff, the single-valued upper hemicontinuous selector $\theta(\cdot)$ is continuous. Then $T_\kappa=\Pi^{L^2(\muK)}_{\calF_m}\circ Q^{\pi_{\theta(\cdot)}}$
is a composition of continuous maps: $\widehat W\mapsto\theta(\widehat W)$ is
continuous, $\theta\mapsto\nnp$ is continuous in the supremum norm by the
parametrization, and since $g$ is continuous and bounded, dominated convergence makes
$\pi\mapsto Q^{\pi}=\E[g(x',\pi(x'))\mid x]$ continuous from the supremum norm into
$L^2(\muK)$; here no class-wide modulus is needed, the single fixed parametrization
$\theta\mapsto\nnp$ on the compact policy-parameter set $\Theta$ already supplies it, so
$\theta\mapsto \Qpi$ is continuous, and the metric projection onto the
closed convex set $\calF_m$ is non-expansive, hence continuous. Moreover
$T_\kappa$ maps the closed convex set
$\calB:=\{\widehat W\in\calF_m:\|\widehat W\|_{L^2(\muK)}\le\|g\|_\infty\}$ into
itself, since $\|\Pi_{\calF_m}(Q^{\pi})\|_{L^2(\muK)}\le\|Q^{\pi}\|_{L^2(\muK)}
\le\|g\|_\infty$ for every admissible $\pi$ (projection onto a closed convex set
containing the origin, Assumption~\ref{ass:nested}, is norm-non-increasing). For a finite-dimensional closed convex
sieve class $\calF_m$ (for instance a linear sieve, spline, or polynomial class spanning a
fixed finite-dimensional subspace of $L^2(\muK)$), $\calB$ is compact and convex and
Brouwer's theorem \citep{aliprantisborder2006} yields a fixed point. The neural networks of the
implementation are the nonconvex parametrization of such a class (the remark below): the network
manifold is finitely parametrized but is not a convex subset of $L^2(\muK)$, so the theorem is a
statement about the idealized sieve object the network approximates, not about the network class
literally. For an infinite-dimensional $\calF_m$ the same conclusion holds under
the additional hypothesis that $T_\kappa(\calB)$ is relatively compact, for instance when the metric
projection lands in a fixed finite-dimensional subclass, or when $\calF_m$ embeds compactly so
that the smoothing of $Q^{\pi}$ by the bounded continuous kernel $g$ is a compact operation, via
Schauder's theorem; this extension is not needed for any result in the paper.

\textbf{Fixed point $\Rightarrow$ $\kappa$-EWM.} Suppose
$\widehat W_\star=T_\kappa(\widehat W_\star)
=\Pi^{L^2(\muK)}_{\calF_m}(Q^{\pi_{\theta(\widehat W_\star)}})$, and write
$\pi_\star=\pi_{\theta(\widehat W_\star)}$. Then
$\theta(\widehat W_\star)\in\Theta_\kappa(\widehat W_\star)$ solves the policy arm
given $\widehat W_\star$, and by \eqref{eqn:proj_id} the identity
$\widehat W_\star=\Pi_{\calF_m}(Q^{\pi_\star})$ solves the world arm given
$\pi_\star$. The pair $(\pi_\star,\widehat W_\star)$ solves both arms exactly and is
therefore an exact $\kappa$-EWM (Definition~\ref{def:kewm}).

\textbf{$\kappa$-EWM $\Rightarrow$ fixed point.} Let $(\pi,\widehat W)$ be an exact
$\kappa$-EWM. Exactness of the world arm and \eqref{eqn:proj_id} force
$\widehat W=\Pi_{\calF_m}(Q^{\pi})$; exactness of the policy arm puts the
parameter of $\pi$ in $\Theta_\kappa(\widehat W)$, so $\pi=\pi_{\theta(\widehat W)}$
for the chosen selector. Substituting,
$\widehat W=\Pi_{\calF_m}(Q^{\pi_{\theta(\widehat W)}})=T_\kappa(\widehat W)$,
a fixed point.

\textbf{Interpretation.} At such a fixed point two conditions hold jointly on
$\muK$: by \eqref{eqn:proj_id} the perception $\widehat W_\star$ is the best
approximation in the class $\calF_m$ to the continuation $Q^{\pi_\star}$ that
the equilibrium itself generates, and by the policy arm $\pi_\star$ is a best
response to $\widehat W_\star$. This is the structure of a coverage-confirmed
fixed point on $\muK$: beliefs are best-in-class on the measure they are evaluated
against, behavior is consistent with beliefs, and the residual is left unconstrained off
$\muK$. When $\muK=\muerg$ and the perception gap is closed there ($\widehat W_\star=Q^{\pi_\star}$
$\muerg$-a.s.), this is the classical self-confirming equilibrium on the policy's own
ergodic measure (in the sense of \citealp{fudenberg1993selfconfirming} and
\citealp[Ch.~6]{sargent1999conquest}), the belief being confirmed on the events the policy
visits; if the perception gap does not close it is the restricted-perceptions counterpart. For
$\muK\neq\muerg$ it
is the same fixed-point structure on the enlarged, modeler-designed measure, a
coverage-confirmed fixed point. What binds
is coverage, not the perception class, which sets it apart from the restricted-perceptions
equilibrium of \citet{BranchEvans2006RPE}.
\end{proof}

\begin{proof}[Proof of Theorem~\ref{thm:consist}]
By (v) and the Arzel\`a--Ascoli theorem \citep[Theorem~7.25]{rudin1976} the sequence $(\pi^\kappa)$ is relatively
compact in $C(\X;\A)$, so it has a subsequence $\pi^{\kappa_j}\to\pi^\infty$
converging uniformly on $\X$ to an admissible policy $\pi^\infty$ (admissible because
$\A(x)$ has closed graph, Assumption~\ref{ass:regular}, and the convergence is
uniform); fix one such subsequence.

\textbf{Step 1 (the exact residual vanishes on the training measures).} Apply the
residual decomposition of Proposition~\ref{prop:decomp} to the pair
$(\pi^{\kappa_j},\wm^{\kappa_j})$ with the test measure taken to be the training
measure itself, $\mu=\mu_{\kappa_j}$, so that
$\mathrm{CoverageErr}=2\|\rwm\|_\infty^2\,\TV(\mu_{\kappa_j},\mu_{\kappa_j})=0$. Here
$\hat q^{\pi^{\kappa_j}}_{\kappa_j}$ is the $L^2(\mu_{\kappa_j})$-projection of the
current policy's continuation $Q^{\pi^{\kappa_j}}$ and
$\mathrm{OptErr}_{\kappa_j}$ is evaluated at $(\pi^{\kappa_j},\wm^{\kappa_j})$, so
conditions (i)--(iii) apply to the right-hand side verbatim:
\[
\E_{\mu_{\kappa_j}}\!\bigl[\|\rwm(x,\pi^{\kappa_j}(x),Q^{\pi^{\kappa_j}}(x))\|^2\bigr]
\le 3\Bigl(\mathrm{OptErr}_{\kappa_j}+L_R^2\,\eta_{\kappa_j}(\theta^{\kappa_j})
+L_R^2\,\E_{\mu_{\kappa_j}}\bigl[\|\hat q^{\pi^{\kappa_j}}_{\kappa_j}-\wm^{\kappa_j}\|^2\bigr]\Bigr).
\]
The three terms on the right vanish as $j\to\infty$ by conditions (iii), (i) and
(ii) respectively. Writing
$g_j(x):=\|\rwm(x,\pi^{\kappa_j}(x),Q^{\pi^{\kappa_j}}(x))\|^2$,
\begin{equation}
\int g_j\,d\mu_{\kappa_j}\;\longrightarrow\;0. \tag{$\dagger$}\label{eqn:thm2dagger}
\end{equation}

\textbf{Step 2 (the residual is uniformly continuous in the policy).} By
Assumption~\ref{ass:regular}, $\Gamma$ is uniformly continuous and $g$ is uniformly
continuous and bounded on the compact $\X\times\A\times\calE$, with moduli
$\omega_\Gamma,\omega_g$; by (v) the admissible policies share a modulus of
continuity $\omega_\pi$. Fix admissible $\pi_1,\pi_2$ and abbreviate
$y_i=\Gamma(x,\pi_i(x),\varepsilon)$ and $G_i=g(y_i,\pi_i(y_i))$. Then pointwise in
$(x,\varepsilon)$,
\begin{align*}
\|y_1-y_2\| &\le \omega_\Gamma\bigl(\|\pi_1-\pi_2\|_\infty\bigr),\\
\|\pi_1(y_1)-\pi_2(y_2)\|
&\le \|\pi_1(y_1)-\pi_1(y_2)\|+\|\pi_1(y_2)-\pi_2(y_2)\|
\le \omega_\pi\bigl(\|y_1-y_2\|\bigr)+\|\pi_1-\pi_2\|_\infty,\\
\|G_1-G_2\|
&\le \omega_g\Bigl(\|y_1-y_2\|+\|\pi_1(y_1)-\pi_2(y_2)\|\Bigr),
\end{align*}
where the second line uses the common policy modulus (v) at the outer
composition $\pi(\Gamma(\cdot))$. Substituting the first two bounds into the third and using the finite quadrature rule gives
\[
\sup_{x\in\X}\|Q^{\pi_1}(x)-Q^{\pi_2}(x)\|
\le \sum_k w_k\sup_{x\in\X}\|G_{1k}(x)-G_{2k}(x)\|
\le\omega\bigl(\|\pi_1-\pi_2\|_\infty\bigr)
\]
for an explicit modulus $\omega$ assembled from $\omega_g,\omega_\Gamma,\omega_\pi$, with
$\omega(t)\to0$ as $t\to0$ and independent of $x$. Hence
$\pi^{\kappa_j}\to\pi^\infty$ uniformly implies $Q^{\pi^{\kappa_j}}\to Q^{\pi^\infty}$
uniformly; and as $\rwm$ is uniformly continuous on the compact
$\X\times\A\times\mathcal{Q}_0$ and every exact continuation lies in
$\mathcal{Q}\subseteq\mathcal{Q}_0$ (the larger $\mathcal{Q}_0$ is required only for the projected
and surrogate continuations used elsewhere, not for the exact continuations evaluated here), it
follows that
$g_j\to g_\infty:=\|\rwm(\cdot,\pi^\infty(\cdot),Q^{\pi^\infty}(\cdot))\|^2$
uniformly on $\X$. Every $g_j$ and $g_\infty$ is continuous and bounded by
$\|\rwm\|_\infty^2$.

\textbf{Step 3 (pass to the limit).} For the fixed bounded continuous function
$g_\infty$ (continuous as the uniform limit of the continuous $g_j$, Step~2),
\[
\Bigl|\int g_\infty\,d\mu_\infty-\int g_j\,d\mu_{\kappa_j}\Bigr|
\;\le\;
\Bigl|\int g_\infty\,d\mu_\infty-\int g_\infty\,d\mu_{\kappa_j}\Bigr|
\;+\;\|g_\infty-g_j\|_\infty .
\]
The first term tends to $0$ by the weak convergence $\mu_{\kappa_j}\Rightarrow
\mu_\infty$ of (iv), applied to the single bounded continuous $g_\infty$; the
second tends to $0$ by Step~2. With \eqref{eqn:thm2dagger},
\[
\int g_\infty\,d\mu_\infty=\lim_{j\to\infty}\int g_j\,d\mu_{\kappa_j}=0 .
\]

\textbf{Step 4 (conclude).} $g_\infty\ge0$ is continuous and integrates to zero
against $\mu_\infty$, so $g_\infty=0$ $\mu_\infty$-almost surely, i.e.\
$\rwm(x,\pi^\infty(x),Q^{\pi^\infty}(x))=0$ for $\mu_\infty$-a.e.\ $x$. Since
$g_\infty$ is continuous and $\mu_\infty$ has full support on the coverage-certified region $\calR$
by the full-support property of Assumption~\ref{ass:nested}, its zero set is all of $\calR$; hence $\pi^\infty$ satisfies the exact
equilibrium condition everywhere on the certified region the coverage schedule reaches.
\end{proof}

\section{Experimental protocol and architecture}\label{app:protocol}

The arms, accuracy measure, and compute axis summarized in
Section~\ref{sub:protocol} are specified in full here, together with the metrics and
the proposition-to-experiment map that the theory of Section~\ref{sec:theory}
obliges.

We separate three questions that are easy to conflate. \emph{Verification} asks whether the code
evaluates the intended residual: we evaluate the structural residual through the same code path EWM
uses in training at the closed-form Brock--Mirman policy, where it must vanish up to quadrature
precision, and we reproduce the published reference solution. \emph{Validation} asks whether the
learned policy satisfies the exact equilibrium conditions on states it never trained on: the held-out
exact residual $R^{H}$, and, where a closed-form or independent reference exists, the supremum-norm
error against it. \emph{Uncertainty quantification} asks how stable the result is under the randomness
of training: every reported number is a distribution across at least ten independent seeds, summarized
by its median, interquartile range, and worst case. The arms, accuracy measure, and compute axis below
address the three separately.

To identify the mechanism we report a three-rung ladder that isolates coverage from
the surrogate, with one refinement and a benign control:
\begin{enumerate}[leftmargin=1.7em,nosep]
\item \texttt{DEQN-path-exact}: the pathwise baseline, exact residual on the
policy's own simulated path (no coverage, no surrogate). The coverage gap.
\item \texttt{DEQN-coverage-exact}: the exact residual on the enlarged coverage
measure $\muK$, no surrogate. This rung isolates whether coverage alone
suffices, and is the control without which the mechanism cannot be identified.
\item \texttt{EWM-coverage-surrogate}: coverage plus the audited continuation
surrogate, the full construction. The fix.
\end{enumerate}
These three arms are reported at $N{=}2$ and $N{=}4$ in Table~\ref{tab:irbc}, each from a
single consistent ten-seed run, and the scaling of the amortization index $A(\tau)$ with $N$
in Table~\ref{tab:irbc_scaling}. The middle rung, \texttt{DEQN-coverage-exact}, is the
surrogate-free control that isolates whether the learned continuation is decisive. The IRBC coverage
measure is specified in Appendix~\ref{app:irbc_cov_spec}; the warm-started sensitivity exercise of
Figure~\ref{fig:irbc_kappa} instead varies the surrogate width and exact-target anchor share at fixed
coverage reach. The smooth Brock--Mirman model (Section~\ref{sub:worked}) serves as the
benign control that closes the design. The placebo-coverage control, in which coverage states
are drawn from an arbitrary hypercube rather than the model's forward map, isolates whether the
gain comes from the exact transition rather than from sampling more states. Finer ablation arms, a
surrogate-only arm (\texttt{DEQN-path-surrogate}), a routed arm (\texttt{EWM-routed}), and an
exact-audit arm (\texttt{EWM-exact-audit}), isolate secondary effects and are not reported here.

\paragraph{Metrics.} Average Euler errors are not the headline; the economically
relevant quantities live in the tail and off the path. We report the $p_{95}$ and
$p_{99}$ held-out Euler residuals; the disaster- and impulse-response (GIRF)
residuals along structurally generated counterfactual paths
(Proposition~\ref{prop:girf}); residuals at constraint boundaries, conditional on the
constraint binding; the failed-seed frequency (a run is failed if its
held-out $R^{H}$ has not fallen below the convergence tolerance by the episode budget); the
verified-stationarity rate, the fraction of seeds whose trained policy is invariant
on a fixed held-out evaluation set sampled from the coverage distribution, $\sup_{x}\|\pi_{t+\Delta}(x)-\pi_t(x)\|<10^{-3}$ (the supremum taken over that finite held-out set, which beyond a few dimensions is a Monte Carlo sample, not a dense grid)
over $\Delta$ further re-simulations and the disaster-region mean $R^{H}$ below the convergence
tolerance of $1.5\times10^{-2}$; and the surrogate gap $R^{H}-R^{L}$ that audits
the learned continuation. Compute is reported as the number of expensive
exact-quadrature evaluations, separated to avoid hidden cost into $B_{\mathrm{policy}}$
(exact calls inside policy training), $B_{\mathrm{world}}$ (exact calls used to train the
continuation surrogate), and $B_{\mathrm{audit}}$ (exact held-out evaluations), with
total $B_{\mathrm{total}}=B_{\mathrm{policy}}+B_{\mathrm{world}}+B_{\mathrm{audit}}$,
alongside wall-clock; ``compute parity'' means matched $B_{\mathrm{total}}$.

\paragraph{Diagnostics motivated by the theory.} The propositions of
Section~\ref{sec:theory} motivate a battery of diagnostics, reported per model where feasible:
(O1)~the held-out $R^{H}$ attributed to the optimization, perception-class, surrogate-fit, and
coverage terms of Proposition~\ref{prop:decomp}, showing $\mathrm{CoverageErr}$ dominates off-path
for \texttt{DEQN-path-exact}; (O2)~a coverage sweep of off-path $R^{H}$ against the reach $\kappa$,
testing the decrease Theorem~\ref{thm:consist} predicts in the vanishing-error limit; (O3)~the
coverage gap $\TV(\mu,\muK)$ (or density ratio $B$) at each stage, tracking that decrease; (O4)~the
off-ergodic bound of Theorem~\ref{thm:offpath}, off-path $\E_\eta\|R^{H}\|^2$ against both terms
$2B\,\E_{\muK}\|R^{L}\|^2$ and $2B\,L_R^2\,\E_{\muK}\|\wm-\Qpi\|^2$, not the continuation term
alone; (O5)~the surrogate gap $R^{H}-R^{L}$ and continuation error $\|\wm-\Qpi\|$ on held-out states
(Proposition~\ref{prop:surrgap}), with the routing threshold and routed fraction; (O6)~the
generalized impulse-response residual by horizon (Proposition~\ref{prop:girf}) with an estimate of
the stability modulus $\lambda$; (O7)~the seed-basin scatter of on-path against off-path $R^{H}$;
and (O8), as accuracy references, the analytic Brock--Mirman solution and a distributional accuracy
metric for the Bewley economy. The per-arm results for the headline model, and the coverage
configuration that realizes $\kappa$, are reported with that model in Tables~\ref{tab:irbc}
and~\ref{tab:irbc_scaling} of Section~\ref{sub:irbc}.

\paragraph{Reproducibility and provenance.} Every IRBC number reported in
Section~\ref{sub:irbc} is the exact held-out $R^{H}$ from completed ten-seed
runs at $N{=}2$ and $N{=}4$; the surrogate residual
$R^{L}$ enters only the audit. The solver is ab initio, random initialization with the exact
structural residual as the only objective, no external/reference warm start and no reference
solution as a label. Every figure and table regenerates from the logged per-seed CSV, with seeds, budgets,
and coverage weights recorded per row under version control.\footnote{All numerical results
reported in this paper are computed in double-precision (float64) arithmetic.}

\subsection{Networks, training, and architecture}\label{app:arch}\label{app:hyper}

Table~\ref{tab:wm_vs_ewm} contrasts an artificial world model with an equilibrium world model, expanding the brief comparison of Section~\ref{sub:whatis}.

\begin{table}[t!]\centering\small
\begin{tabular}{p{0.24\textwidth}p{0.33\textwidth}p{0.33\textwidth}}
\toprule
 & \textbf{World model (AI)} & \textbf{Equilibrium world model} \\
 & a Dreamer agent learning a game's physics & this paper \\
\midrule
What is modeled
 & the environment's dynamics: the next observation or state
 & the continuation $\Qpi$: next period's marginal utilities, prices, and multipliers \\[3pt]
The transition law
 & learned from observed play, and approximate
 & the map $\Gamma$ is exact, known, and never learned; it is evaluated at the policy's own action $a=\nnp(x)$, the unknown we solve for \\[3pt]
What imagination buys
 & rollouts in a learned latent train the agent on states its play never reached
 & the exact residual is imposed at off-path states the policy's own path never reaches \\[3pt]
Training signal
 & predict observed transitions; reward for the policy
 & quadrature on $\Gamma$ for the surrogate; the exact residual, unsupervised, for the policy \\[3pt]
Discipline on imagined states
 & predictive fit only
 & feasibility, market clearing, optimality, and complementarity, since the states come from the exact $\Gamma$ \\[3pt]
Goal
 & a high-reward control policy
 & an equilibrium policy with small exact residual everywhere the economy can reach \\
\bottomrule
\end{tabular}
\caption{From a world model to an equilibrium world model. The columns are: a left axis naming the
aspect compared; World model (AI), the answer for a Dreamer agent learning a game's physics; and
Equilibrium world model, the answer for this paper. Both rehearse in imagined
states to learn about regions ordinary experience never reaches. The difference is one
of discipline: an artificial world model learns the dynamics and is bound only by
predictive fit, whereas an equilibrium world model takes the transition map $\Gamma$ as exact
and known, learns the equilibrium policy and the continuation that enters its exact residual
but never the dynamics,\protect\footnotemark{} and admits a
rehearsed state only if it satisfies the model's own feasibility, market-clearing, and
optimality conditions.}
\label{tab:wm_vs_ewm}
\end{table}
\footnotetext{The map $\Gamma$ is written in closed form for both blocks of the state: the
exogenous block follows a known law independent of the policy (for example an AR(1)
productivity shock), and the endogenous block a known accounting identity in the action (for
example the capital law $K'=(1-\delta)K+I$). Early in training the endogenous next state that
$\Gamma$ implies can be far from its equilibrium value, but the only approximate input is the
action $a=\nnp(x)$, never $\Gamma$ itself: this is an unconverged policy, not an approximated
transition, and it self-corrects as $\nnp$ converges. A learned world model, by contrast, has
the dynamics themselves wrong, and only more data corrects them.}

\paragraph{Stop-gradient and Polyak-target conventions.} Two conventions used in
Algorithm~\ref{alg:ewm} and the world arm of Section~\ref{sub:discipline}. The stop-gradient operator
$\sg(\cdot)$ is the identity in the forward pass but has Jacobian zero in the backward pass, so
applying it to the regression target lets the surrogate $\wm$ chase a detached estimate of the
continuation rather than co-adapt with it. Where it stabilizes training, the bootstrapped target
$Q^{H}(x;\bar\theta)$ is read at a Polyak average $\bar\theta_t=(1-\tau)\bar\theta_{t-1}+\tau\theta_t$
(small $\tau\in(0,1)$) of the policy parameters, not the surrogate's: because the target is itself a
function of the policy at the next-period nodes, reading it at a slow $\bar\theta$ decouples it from
the live update that chases it and damps the high-frequency oscillations of the bootstrapped fixed
point, the standard target-network device of deep reinforcement learning
\citep{polyak1992acceleration,lillicrap2016continuous}. The same $\bar\theta$ holds the coverage
measure $\muK(\bar\theta)$ fixed within a stage, and the outer loop imposes self-consistency
$\bar\theta=\theta(\widehat W_\star)$ in the theory of Section~\ref{sec:theory}.

For reproducibility, Table~\ref{tab:hyper} collects the network architecture (width, depth,
activation, output transform), the optimizer and learning-rate schedule, the minibatch size
and episode budget, and the exact-quadrature rule used to evaluate $R^{H}$, for the two solved
models. Both arms share the policy network; EWM adds the continuation surrogate $\wm$.
Figure~\ref{fig:architecture} pictures the two for the Brock--Mirman setting side by side: DEQN as a
single policy network trained by one equilibrium loss, and EWM as that same policy plus the
surrogate, each trained by its own loss, with the surrogate entering the policy residual under a
stop-gradient. The coverage schedule that realizes the reach $\kappa$ is given with the construction
in Section~\ref{sub:irbc}.

\begin{figure}[t!]
\centering
\resizebox{\textwidth}{!}{
\begin{tikzpicture}[
  font=\small,
  >={Stealth[length=2.4mm]},
  inn/.style={circle, draw=black!55, fill=black!6, minimum size=5mm, inner sep=0pt, font=\footnotesize},
  outn/.style={circle, draw=black!60, fill=white, minimum size=5.5mm, inner sep=0pt, font=\footnotesize},
  mlpP/.style={rectangle, rounded corners=3pt, draw=blue!55!black, fill=blue!8,
    minimum width=1.9cm, minimum height=1.0cm, align=center, font=\footnotesize},
  mlpW/.style={rectangle, rounded corners=3pt, draw=teal!55!black, fill=teal!14,
    minimum width=1.9cm, minimum height=1.0cm, align=center, font=\footnotesize},
  loss/.style={rectangle, rounded corners=2pt, draw=black!60, fill=black!5, align=center,
    font=\footnotesize, minimum width=2.5cm, minimum height=1.0cm, inner sep=3pt},
  quad/.style={rectangle, rounded corners=2pt, draw=black!50, dashed, fill=black!3,
    align=center, font=\scriptsize, minimum height=0.85cm, inner sep=3pt},
  flow/.style={-{Stealth[length=2.6mm]}, line width=1pt, black!82},
  grad/.style={->, very thick, red!55!black},
  aux/.style={->, dashed, black!60}
]
\node[font=\small, blue!60!black] at (2.7,5.9) {DEQN: one network, one loss};
\node[inn] (dk) at (0,4.3) {$k$};
\node[inn] (dz) at (0,3.5) {$z$};
\node[mlpP] (dm) at (1.7,3.9) {MLP $\nnp$\\[1pt]$32\times2$, $\tanh$};
\node[outn] (da) at (3.7,3.9) {$a$};
\draw[flow] (dk) -- (dm); \draw[flow] (dz) -- (dm); \draw[flow] (dm) -- (da);
\node[loss] (dl) at (6.0,3.9) {equilibrium loss\\$\|\rwm(x,a,\Qpi)\|^2$};
\draw[flow] (da) -- (dl);
\node[quad] (dq) at (3.7,1.8) {$\Qpi$: exact\\quadrature on $\Gamma$};
\draw[aux] (dq) -- (dl) node[pos=0.45, right, font=\scriptsize] {continuation};
\draw[grad] (dl.north) .. controls +(0,0.55) and +(0,0.55) .. (dm.north)
   node[pos=0.5, above, font=\scriptsize, red!55!black] {update $\theta$};
\draw[black!20, dashed] (7.5,-1.7) -- (7.5,6.2);
\begin{scope}[shift={(8.0,0)}]
\node[font=\small, red!58!black] at (3.1,5.9) {EWM: two networks, two losses};
\node[inn] (pk) at (0,4.3) {$k$};
\node[inn] (pz) at (0,3.5) {$z$};
\node[mlpP] (pm) at (1.7,3.9) {policy\\MLP $\nnp$};
\node[outn] (pa) at (3.7,3.9) {$a$};
\draw[flow] (pk) -- (pm); \draw[flow] (pz) -- (pm); \draw[flow] (pm) -- (pa);
\node[loss] (pl) at (6.0,3.9) {policy loss\\$\|\rwm(x,a,\wm)\|^2$};
\draw[flow] (pa) -- (pl);
\draw[grad] (pl.north) .. controls +(0,0.55) and +(0,0.55) .. (pm.north)
   node[pos=0.5, above, font=\scriptsize, red!55!black] {update $\theta$};
\node[inn] (sk) at (0,1.4) {$k$};
\node[inn] (sz) at (0,0.6) {$z$};
\node[mlpW] (sm) at (1.7,1.0) {surrogate\\MLP $\wm$};
\node[outn] (sq) at (3.7,1.0) {$\wm$};
\draw[flow] (sk) -- (sm); \draw[flow] (sz) -- (sm); \draw[flow] (sm) -- (sq);
\node[loss] (sl) at (6.0,1.0) {world loss\\$\|\wm-\Qpi\|^2$};
\draw[flow] (sq) -- (sl);
\draw[grad] (sl.south) .. controls +(0,-0.6) and +(0,-0.6) .. (sm.south)
   node[pos=0.5, below, font=\scriptsize, red!55!black] {update $\psi$};
\draw[aux] (sq.north) .. controls +(0,1.3) and +(-1.0,-1.3) .. (pl.south);
\node[font=\scriptsize, align=left, anchor=west] at (5.55,2.25)
   {$\wm$ stands in for\\$\Qpi$ (stop-grad)};
\node[quad] (sq2) at (3.7,-1.0) {$\Qpi$: exact quadrature\\on sparse anchors};
\draw[aux] (sq2) -- (sl) node[pos=0.5, right, font=\scriptsize] {regression target};
\end{scope}
\end{tikzpicture}}
\caption{Network architecture and training, Brock--Mirman setting (Table~\ref{tab:hyper} lists
the hyperparameters per model). Left, DEQN: a single policy $\nnp$ maps the
state $x=(k,z)$ through a multilayer perceptron (MLP; two $\tanh$ layers of width $32$) to the action $a$; the
continuation $\Qpi$ in the residual is computed by exact Gauss--Hermite quadrature at every
state, and the one equilibrium loss $\|\rwm(x,a,\Qpi)\|^2$ updates the policy weights $\theta$.
Right, EWM: two networks trained by two separate losses. The policy $\nnp$ is updated
by the same equilibrium residual, but with the continuation supplied by a second network, the
surrogate $\wm$, entered under a stop-gradient so the policy loss does not train $\wm$. The
surrogate is updated by its own regression loss $\|\wm-\Qpi\|^2$ against the exact continuation
evaluated by quadrature on a sparse anchor set (a random subsample of the coverage batch); the
held-out exact residual $R^{H}$ then audits it. The structural
model, transition $\Gamma$, and residual $\rwm$ are identical across the two arms; EWM differs
only in adding $\wm$ (of capacity $m$, $\tanh$) and in the coverage measure on which the
residual is imposed. The rare-disaster variant of Section~\ref{sub:bmdc} adds the indicator $d$
to the inputs.}
\label{fig:architecture}
\end{figure}
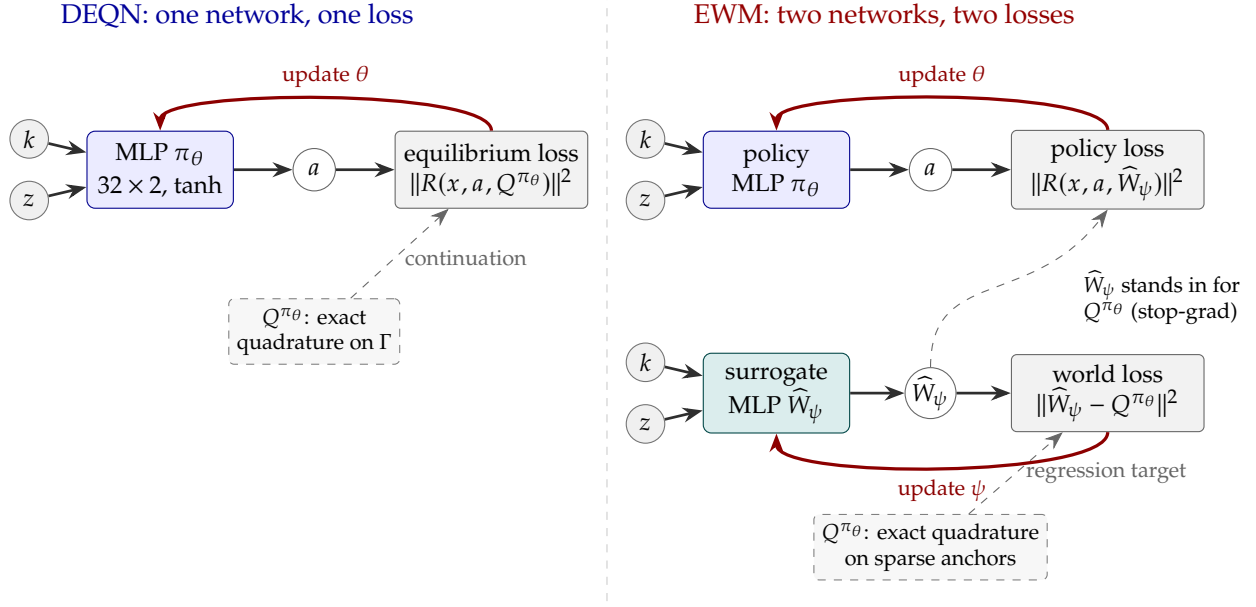

\begin{table}[t]\centering\small
\setlength{\tabcolsep}{5pt}

\begin{tabular}{lll}
\toprule
 & Brock--Mirman & IRBC ($N$ countries) \\
\midrule
Policy net                  & $32\times2$ & $96\times2$ \\
Surrogate $\wm$ net         & $m\times2$, $m\in\{16,32,64\}$ & $96\times2$ \\
Activation (both nets)      & $\tanh$ & $\tanh$ \\
Policy output transform     & sigmoid ($c$), softplus ($\mu$) & softplus / identity \\
Surrogate output transform  & softplus (positive $\Qpi$) & softplus (positive $\Qpi$) \\
Policy optimizer / LR        & Adam, constant $10^{-3}$ & Adam, cosine $1.2\times10^{-3}$, grad-clip $8$ \\
Surrogate optimizer / LR     & Adam, constant $10^{-3}$ (as policy) & Adam, constant $10^{-3}$ \\
Minibatch / episodes        & $256$ / $3000$ & $256$ / $3500$ \\
Quadrature (rule / nodes) & Gauss--Hermite, $5$ & monomial, $4(N{+}1)$ \\
State dimension           & $2$ ($k,z$); $3$ with $d$ & $2N{+}1$ \\
\bottomrule
\end{tabular}
\caption{Network architecture and training hyperparameters for the two solved models. The columns
are: a left axis listing each architecture or training setting; Brock--Mirman, its value in the
Brock--Mirman model; and IRBC ($N$ countries), its value in the international real business cycle
model. The
first two rows are hidden-layer width\,$\times$\,number of hidden layers; both arms (DEQN, EWM)
share the policy, and EWM adds the continuation surrogate $\wm$ of capacity $m$, a
separate feed-forward network of the listed width that outputs a positive continuation
$\Qpi$ through a softplus. The surrogate shares the policy's activation, minibatch, and
episode budget, and is fit to exact quadrature targets once per episode, after which the
policy takes its inner gradient steps. Its optimizer matches the policy's in the Brock--Mirman
setting (Adam at a constant $10^{-3}$); in the IRBC the policy alone uses the cosine
learning-rate schedule and gradient-norm clipping listed, while the surrogate is trained at a
constant $10^{-3}$. For the IRBC scaling runs the dense policy is replaced by a shared
per-country trunk of the same size, and the quadrature is the disaster-crossed monomial rule
($2(N{+}1)$ monomial nodes over the $N{+}1$ continuous shocks, crossed with the two disaster states,
$=4(N{+}1)$).
The Bewley experiment (Section~\ref{sub:bewley}) keeps a two-layer width-$64$ policy with Adam at a
constant $10^{-3}$ over $8000$ steps, but its learned object is the distributional encoder
$h_\psi$ in place of a continuation surrogate (one hidden layer of width $64$, $\tanh$, output
dimension $d_{\mathrm{latent}}{=}4$), trained jointly with the policy; the continuation there is
evaluated exactly by enumeration over the $2\times2$ aggregate-by-idiosyncratic next states
rather than by quadrature.}\label{tab:hyper}
\end{table}

\begin{table}[t]\centering\small
\setlength{\tabcolsep}{6pt}
\begin{tabular}{lll}
\toprule
Parameter & Brock--Mirman & IRBC ($N$ countries) \\
\midrule
Discount factor $\beta$                          & $0.95$              & $0.99$ \\
Capital share                                    & $\alpha{=}0.36$     & $\zeta{=}0.36$ \\
Depreciation $\delta$                            & $0.10$              & $0.01\!\to\!0.06$ (disaster) \\
EIS $\gamma$                                      & log ($\gamma{=}1$)  & $\gamma_j\in\{0.25,1.0\}$ \\
Adjustment cost $\phi$                           & --- (irreversible only) & $0.50$ \\
TFP persistence $\rho$                            & $0.90$              & $0.95$ \\
TFP innovation s.d.                              & $\sigma{=}0.04$     & $(\sigma_1,\sigma_2){=}(0.004,0.020)$ \\
Aggregate-shock s.d. $\sigma_a$                  & ---                 & $0.006$ \\
Disaster probability $p_d$                       & $0.002$             & $0.001$ \\
Disaster persistence $\Pr(d'{=}1\mid d{=}1)$     & $0.60$              & $0.85$ \\
Disaster TFP collapse                            & $b{=}0.45$          & $0.30$ \\
\bottomrule
\end{tabular}
\caption{Calibration of the two solved models: the Brock--Mirman model with a rare disaster
(Section~\ref{sub:bmdc}) and the international real business cycle model with rare disasters
(Section~\ref{sub:irbc}). The columns are: Parameter, the calibration parameter; Brock--Mirman, its
value in the Brock--Mirman model; and IRBC ($N$ countries), its value in the international real
business cycle model. Brock--Mirman has a single good and log utility, so the cross-country
volatilities and the convex adjustment cost are not applicable; its only friction beyond the disaster
is the irreversibility constraint $i\ge0$. The disaster scales productivity by $(1-b\,d)$ in
Brock--Mirman and by $(1-0.30\,d)$ in the IRBC. The endogenous-protection variant of
Appendix~\ref{sub:bmdcp} keeps this Brock--Mirman calibration but makes the disaster rate a choice,
$p_d(\alpha)=p_0(1-\alpha)$ with $p_0=0.01$ and a protection cost $g(\alpha)=\tfrac12\phi\alpha^2$,
$\phi=0.10$.}
\label{tab:calib}
\end{table}

\section{Brock--Mirman details and auxiliary examples}\label{app:bm}

This appendix collects the material behind the Brock--Mirman worked example of
Section~\ref{sub:bmdc}. Appendix~\ref{app:bmdc_spec} gives the full specification of the
rare-disaster model, its on-path sampling, coverage measure, calibration, networks, and the
three solver arms. Appendix~\ref{sub:bmdcp} then develops a self-contained extension in which
an endogenous protection choice moves the disaster probability itself, the smallest setting in
which action-conditioning becomes load-bearing.

\subsection{Brock--Mirman specification}\label{app:bmdc_spec}

This appendix specifies the Brock--Mirman model with a rare disaster of
Section~\ref{sub:bmdc}. Appendix~\ref{app:bmdc_sampling} gives the on-path sampling, the coverage
measure, the calibration (Table~\ref{tab:calib}), the networks, and the three solver arms;
Appendix~\ref{app:bmdc_cov} explains how the coverage parameters are set.

\subsubsection{Sampling and the coverage measure}\label{app:bmdc_sampling}
The on-path simulation runs $256$ tracks under the current policy, initialized once at the
deterministic steady state $(k_{ss},z{=}0,d{=}0)$ and then carried forward across episodes: each
episode advances the tracks $48$ further steps through the exact $\Gamma$ (continuous AR(1)
productivity and the two-state disaster chain) from where the previous episode left them, never
resetting to the steady state, so the on-path batch tracks the policy's own evolving ergodic law
$\muerg$. Coverage adds two path-derived components. The stress component
takes a fraction $0.40$ of those path states as seeds, forces the disaster on ($d{:=}1$),
displaces capital off its ergodic level by $k\mapsto k\,e^{\,U(-0.6,\,0.3)}$, and rolls each
seed three steps through the exact $\Gamma$ (four states per seed); the local component
perturbs every path state by $k\mapsto k\,e^{\,U(-0.1,\,0.1)}$ at fixed $(z,d)$. Each
component is repaired to the feasible region $k\ge10^{-5}$ before the residual is evaluated
(a divergence guard resets any track with $c<10^{-5}$ or $k>50\,k_{ss}$ back to steady
state); the repair bounds the simulation only, never the residual. The continuation
$\Qpi(x)=\E[(1/c')(\mathrm{MPK}'+1-\delta)]$, with
$\mathrm{MPK}'=\alpha e^{z'}(1-bd')(k')^{\alpha-1}$ the next-period marginal product
of capital, is evaluated by a $5$-node Gauss--Hermite rule over $z'$
crossed with the two-state disaster chain; calibration (Table~\ref{tab:calib}) $\alpha{=}0.36,\ \beta{=}0.95,\
\delta{=}0.10,\ \rho{=}0.90,\ \sigma{=}0.04,\ p_d{=}0.002,\ p_{dd}{=}0.60,\ b{=}0.45$. Policy
and surrogate are each two-hidden-layer ($32$, $\tanh$) networks, trained with Adam at
learning rate $10^{-3}$; this is Algorithm~\ref{alg:ewm} with a single stage and $E{=}6$
inner steps. The three arms differ only as in Table~\ref{tab:bmdc}: \texttt{path-exact}
trains on $\muerg$ with exact quadrature, \texttt{coverage-exact} on $\muerg+$coverage with
exact quadrature at every step, and EWM forms the exact targets once per episode, fits $\wm$,
then trains the policy against $\wm$. We report a single calibration over ten seeds, run to
convergence ($3000$ episodes, each reported policy a Polyak tail-average over the final $500$);
errors are the held-out exact residual ($R^{H}$ on a frozen
grid: the normal region $d{=}0$ near the ergodic set, the disaster region $d{=}1$ at
drawn-down capital), since the model has no closed form at $\delta<1$.

\subsubsection{How the coverage parameters are set}\label{app:bmdc_cov}
The stress and local fractions, the displacement scales, and the roll
horizon are design parameters of the coverage measure, fixed by a coverage criterion (that
the rare and post-disaster region carry enough mass to be sampled at all, with the local
perturbation a small fraction of the ergodic spread so perturbed states stay admissible under
$\Gamma$), and never searched against the reported held-out $R^{H}$, which would be tuning on
the evaluation metric. Increasing any of them widens the reach $\kappa$. This is distinct from the
IRBC warm-started exercise in Figure~\ref{fig:irbc_kappa}, which holds coverage reach fixed and
increases surrogate width and exact-target anchor share.

\subsection{Endogenous protection: action-conditioning when the action moves the measure}\label{sub:bmdcp}

The rare-disaster Brock--Mirman economy of Section~\ref{sub:bmdc} is modified in one
respect. In the normal regime the agent can choose protection $\alpha\in[0,1)$, which lowers the
disaster-entry probability from $p_0$ to
$p_d(\alpha)=p_0(1-\alpha)$ at resource cost
$g(\alpha)=\tfrac12\phi\alpha^2$. The state remains $x=(k,z,d)$, the controls are
$(k',\alpha)$, and investment irreversibility is imposed with the same
Fischer--Burmeister condition as before. The calibration keeps the disaster rare but avoidable:
$p_0=0.01$, $b=0.45$, $p_{dd}=0.6$, and $\phi=0.10$. It implies an interior normal-regime
protection policy that rises with capital, approximately from $0.08$ to $0.34$; at the
representative normal state used for the comparison, the VFI benchmark is
$\alpha^\star=0.188$. In the disaster regime protection has no current effect on the transition,
so the optimum is $\alpha^\star(d{=}1)=0$.

The experiment isolates a margin that is absent from the exogenous-disaster case. Investment
changes the next physical state; protection changes the probability measure over regimes. The
protection first-order condition is
\begin{equation}
\underbrace{u'(c)\,g'(\alpha)}_{\text{marginal cost today}}
\;=\;
\underbrace{\beta\,p_0\,\big(\E[V(x',d'{=}0)]-\E[V(x',d'{=}1)]\big)}_{\displaystyle \beta\,p_0\,\Delta V\ \text{: marginal benefit of a lower hazard}},
\label{eq:alpha_foc}
\end{equation}
where $\Delta V$ is the continuation-value gap between the normal and disaster regimes. This is the
object the solver must represent. It is not pinned down by the Euler equation, which carries
marginal utilities rather than the value level of avoiding a regime. A pathwise Euler residual can
therefore be small while the protection margin remains unidentified.

The distinction between a state-only continuation and an action-conditioned continuation is exact in
this example. A state-only surrogate of the parameterized-expectations form,
$\wm(x)\approx\Qpi(x,\pi(x))$, has already evaluated the continuation at the policy's own action.
Consequently
\begin{equation}
\frac{\partial \wm(x)}{\partial\alpha}\equiv 0.
\label{eq:stateonly_dies}
\end{equation}
It perceives no benefit from lowering the hazard and collapses the protection FOC to
$u'(c)g'(\alpha)=0$. By contrast, an action-conditioned object keeps the regime probabilities under
the control of the current action:
\begin{equation}
\widehat Q(x,\alpha)=\E_{z'}\!\left[\sum_{d'}P(d'\mid d,\alpha)\,\widehat v(k',z',d')\right],
\qquad P(d'{=}1\mid d{=}0,\alpha)=p_0(1-\alpha).
\end{equation}
Since the hazard is structural and known,
\begin{equation}
\frac{\partial \widehat Q}{\partial\alpha}
=\E_{z'}\!\left[\sum_{d'}\frac{\partial P(d'\mid d,\alpha)}{\partial\alpha}\,\widehat v(k',z',d')\right]
=-p_0\big(\widehat v(\cdot,1)-\widehat v(\cdot,0)\big)=p_0\,\Delta\widehat V\;\neq 0.
\label{eq:ac_lives}
\end{equation}
The surrogate is therefore required only to learn the regime-resolved value levels. Its marginal
benefit error is $p_0(\Delta\widehat V-\Delta V)$, bounded by
$2p_0\|\widehat v-V\|_\infty$ in this two-regime case; this is the value-level control formalized
in Proposition~\ref{prop:actcond}. No derivative of a learned value function with respect to
$\alpha$ is needed.

All arms use the exact transition $\Gamma$ and are trained on the same coverage measure, that of
Appendix~\ref{app:bmdc_sampling} under calibration~B (Table~\ref{tab:calib}), so the comparison
isolates action-conditioning rather than coverage. At a state $x=(k,z,d)$, consumption
is $c=e^{z}(1-bd)k^\zeta+(1-\delta)k-k'-g(\alpha)$. The residual consists of the consumption
Euler equation, the Fischer--Burmeister irreversibility condition, and the protection FOC
\eqref{eq:alpha_foc}. Expectations over $z'$ use a seven-node Gauss--Hermite rule, crossed with the
two-state disaster chain. The learned EWM objects are an action-conditioned marginal-utility
continuation $\widehat Q\approx\Qpi$ and a bootstrapped regime-resolved value
$\widehat v\approx V$; both are audited against exact quadrature, and neither replaces the structural
law $\Gamma$. Runs use ten seeds per arm and are scored with the exact-quadrature held-out
protection-FOC residual, together with the VFI gap $|\alpha-\alpha^\star|$ as a small-model
cross-check only. The VFI solution is never used as a training target, label, or warm start, and the
full setup is replicable in \texttt{experiments/bm\_dcp}.

The four arms form an identification ladder. \texttt{deqn-plain} minimizes only the Euler and
irreversibility residuals, so the value of avoiding the disaster is outside its objective and
$\alpha$ is unidentified. \texttt{state-only} adds a learned continuation but averages over the
regime transition before differentiating, so \eqref{eq:stateonly_dies} forces the perceived benefit
of protection to zero. \texttt{exact} uses a regime-resolved value and exact quadrature at every
inner policy step. \texttt{action-conditioned} uses the same regime-resolved margin but amortizes it
through \eqref{eq:ac_lives}, computing exact targets once per episode and using the cheap surrogate
inside the policy steps.

\begin{figure}[t!]\centering
\includegraphics[width=0.92\textwidth]{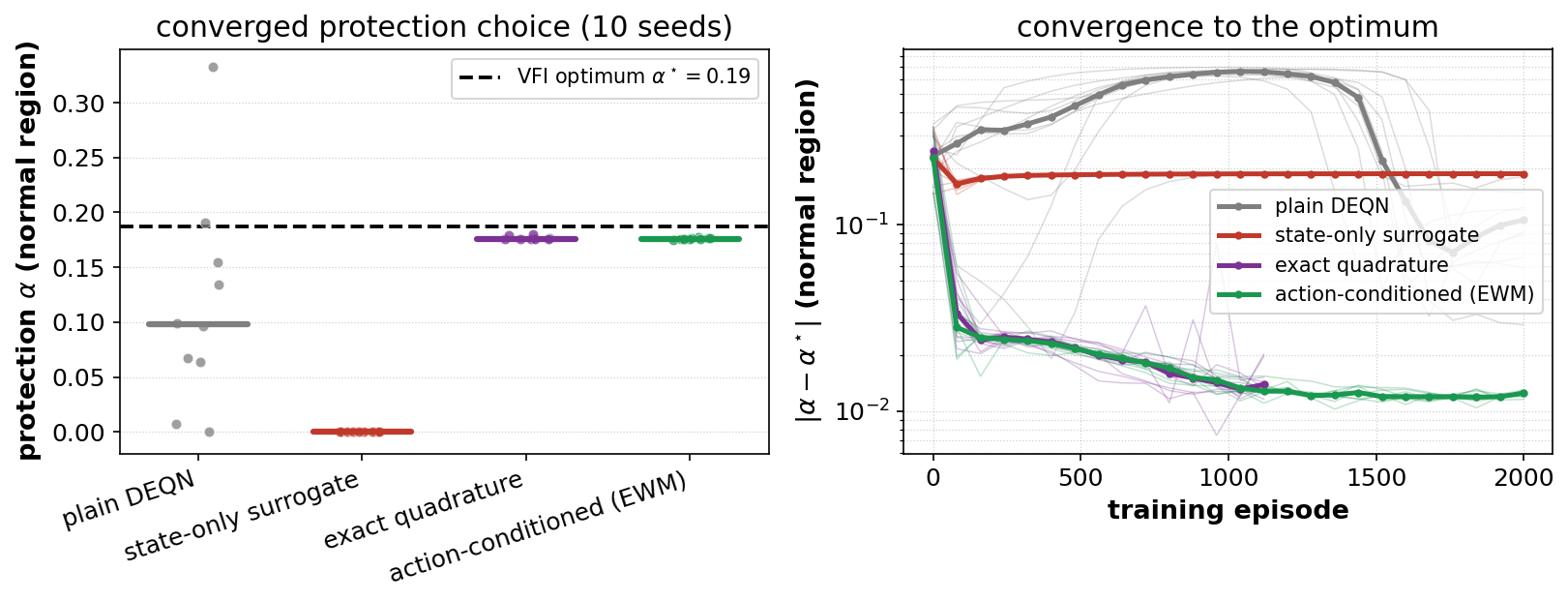}
\caption{Endogenous protection. Left: converged normal-regime protection by arm, with the VFI
benchmark $\alpha^\star=0.188$ dashed. The regime-resolving arms, \texttt{exact} and
\texttt{action-conditioned}, recover the benchmark; \texttt{state-only} collapses to zero because
\eqref{eq:stateonly_dies} removes the regime value gap; \texttt{deqn-plain} is unidentified across
seeds. Right: $|\alpha-\alpha^\star|$ over training. The action-conditioned surrogate reaches the
accuracy of exact quadrature at $4.7\times$ lower exact-evaluation cost (Table~\ref{tab:bmdcp}). Ten
seeds; the VFI solution is an evaluation-only reference.}
\label{fig:bmdcp}
\end{figure}

Table~\ref{tab:bmdcp} gives the numerical ladder. The regime-resolving arms recover the optimal
normal-regime protection: both \texttt{exact} and \texttt{action-conditioned} choose
$\alpha=0.176$, within $0.012$ of the VFI benchmark, and have held-out protection-FOC residuals of
order $10^{-5}$. The state-only arm chooses no protection, exactly as predicted by
\eqref{eq:stateonly_dies}. The plain DEQN has no disciplined protection margin: across ten seeds its
normal-region choice spans $[0,0.33]$, and one seed buys protection inside the disaster where the
optimum is zero. The action-conditioned surrogate matches exact quadrature on the policy margin while
reducing exact evaluations from $11{,}673$ to $2{,}477$ million, a factor of $4.7$.

A final point concerns certification. The state-only surrogate has a
very small residual in its own misspecified terms ($2.9\times10^{-6}$), because after averaging over
regimes it is self-consistent that no protection is valuable. That residual is not a valid
certificate. Only the exact-quadrature protection FOC, evaluated with the regime value gap, exposes
the error. Thus the appendix separates two requirements: when an action moves only the next state,
coverage of the off-ergodic region suffices, the same $\SCE\!\to\!\RE$ device as in
Section~\ref{sub:coverage}, and a state-only world model is enough; when the action moves the
transition measure, the world model must also be action-conditioned. Seed-level robustness of this
ladder at scale is established on the international model of Section~\ref{sub:irbc}.

\begin{table}[t!]\centering\small
\resizebox{\textwidth}{!}{
\begin{tabular}{llccccc}
\toprule
Arm & continuation it carries & $\alpha$(nor) & $|\Delta\alpha|_{\mathrm{nor}}$ &
$|\Delta\alpha|_{\mathrm{dis}}$ & FOC$^{H}$(nor) & $B_{\mathrm{policy}}$ \\
\midrule
\texttt{deqn-plain}        & none (Euler+FB only)               & $[0.00,0.33]$ & $[0.03,0.19]$ & $[0.00,0.41]$ & --- & $3441$ \\
\texttt{state-only}        & state-only continuation $\wm(x)$   & $0.000$ & $0.188$ & $0.000$ & --- & $1239$ \\
\texttt{exact}             & regime value, exact quadrature     & $0.176$ & $0.012$ & $0.001$ & $7.6\times10^{-5}$ & $11673$ \\
\texttt{action-conditioned}& regime-resolved value $\widehat v$ & $0.176$ & $0.012$ & $0.000$ & $3.6\times10^{-5}$ & $2477$ \\
\bottomrule
\end{tabular}}
\caption{Endogenous-protection Brock--Mirman, calibration B: $p_0=0.01$, $\phi=0.10$, $b=0.45$,
$p_{dd}=0.6$. The normal-region VFI benchmark is $\alpha^\star=0.188$ and the disaster-region
benchmark is zero. $|\Delta\alpha|$ denotes the distance to the VFI benchmark; FOC$^{H}$ is the
reference-free exact-quadrature protection-FOC residual on held-out states; and
$B_{\mathrm{policy}}$ is the exact-evaluation budget in millions. For all arms except
\texttt{deqn-plain}, entries are ten-seed medians; \texttt{deqn-plain} reports across-seed ranges because its objective does not
identify $\alpha$. Coverage is held fixed across arms, so the table isolates the value of resolving
the regime gap and action-conditioning the continuation.}\label{tab:bmdcp}
\end{table}

\section{IRBC model details}\label{app:irbc_details}

This appendix collects the implementation detail behind the international real business cycle
experiment of Section~\ref{sub:irbc}: the replicable coverage-measure specification, the
framework-to-IRBC object map, the high-dimensional existence solves at $N{=}8,16,32$, and the
supplementary scaling and pricing exhibits.

\subsection{IRBC coverage-measure specification}\label{app:irbc_cov_spec}

This is the replicable per-episode specification of $\muK$ for the
\texttt{EWM-coverage-surrogate} arm of Section~\ref{sub:irbc}. The path batch is $64$ tracks
$\times\,64$ periods, simulated from seeds $k_j\sim\mathrm{LogUnif}[0.70,1.35]$,
$z_j\sim\mathcal N\!\big(0,(0.8\,\sigma_j/\sqrt{1-\rho^2})^2\big)$, $d{=}0$, rolled under the
exact transition (continuous AR(1) TFP and the two-state disaster chain with the calibration
below). The stress component adds a fraction $\rho_2/\rho_1{=}0.3$ of those states with $d$
set to $1$; the local component adds a fraction $\rho_3/\rho_1{=}0.2$ with
$k_j\mapsto k_j\,e^{0.1\,\eta_j}$ and $z_j\mapsto z_j+0.1\,\sigma_j\,\eta'_j$,
$\eta_j,\eta'_j\sim\mathcal N(0,1)$ i.i.d. Each component is repaired (clipped) to the feasible
box $k_j\in[0.1,10]$, $z_j\in[-0.8,0.8]$ ($\approx\pm4$ ergodic s.d.) before the residual is
evaluated. The residual is always the structural IRBC residual evaluated at the repaired admissible
state; no clipping penalty or auxiliary residual is added. Calibration (Table~\ref{tab:calib}): $\beta{=}0.99$, $\zeta{=}0.36$,
adjustment cost $\phi{=}0.50$, EIS $\gamma_j\in\{0.25,1.0\}$, TFP persistence $\rho{=}0.95$,
idiosyncratic volatilities $(\sigma_1,\sigma_2){=}(0.004,0.020)$, aggregate $\sigma_a{=}0.006$,
disaster $p_d{=}0.001$, $\Pr(d'{=}1\mid d{=}1){=}0.85$, TFP collapse
$a_j{=}e^{z_j}(1-0.30\,d)$, depreciation $\delta(0){=}0.01\!\to\!\delta(1){=}0.06$. The
\texttt{DEQN-coverage-exact} arm uses the same construction with stress/local fractions
$0.5/0.25$. It should therefore be read as a conservative surrogate-free coverage control: it gives
the exact-residual solver at least as much off-path mass as the EWM arm's $0.3/0.2$, not an easier
support. The reach $\kappa$ is the mass and radius of this off-ergodic support and is held
fixed across the warm-started IRBC stages; the fractions $0.1\!\to\!0.2\!\to\!0.4$ advanced over
those stages are the surrogate's exact-target anchor share, a surrogate-fidelity dial of the
world component, not a widening of $\kappa$. These fractions, the
perturbation scale, and the stage schedule are design parameters of the coverage measure, not
optimized against the reported held-out $R^{H}$: tuning on the evaluation metric is disallowed,
so we never search over them for the best disaster number. They are fixed instead by a coverage
criterion, that the rare and post-shock regions carry enough mass to be sampled at all, with the
local perturbation kept a small fraction of the ergodic spread so the perturbed states remain
admissible under $\Gamma$. Increasing the stress and local fractions or the perturbation scale
is exactly increasing the reach $\kappa$. The main comparison is a property of the on-path
formulation, not of one
optimizer setting: a constant learning rate in place of the cosine schedule leaves the
self-confirming gap and its repair unchanged.

\subsection{Framework-to-IRBC mapping}\label{app:irbc_map_sec}

Table~\ref{tab:irbc_map} reads each framework object of Section~\ref{sec:framework} onto its
international real business cycle instance, as summarized in Section~\ref{sub:irbc}.

\begin{table}[t!]\centering\small

\begin{tabular}{lll}
\toprule
Framework object & Symbol & IRBC instance \\
\midrule
State        & $x$                          & capital, log-TFP, disaster flag $(k_{1:N},z_{1:N},d)$ \\
Choice       & $a=\nnp(x)$                  & $(k'_j,i_j,\mu_j,\lambda)$; $c_j=(\lambda/\tau_j)^{-\gamma_j}$ \\
Transition   & $x'=\Gamma(x,a,\varepsilon')$ & $k$-part $=k'$;\ $z'_j=\rho z_j+\sigma_j\varepsilon'_j+\sigma_a\varepsilon'_a$;\ $d'$ 2-state Markov \\
Continuation & $\Qpi_j(x)$                   & Euler expectation \eqref{eqn:irbc_Q}; stabilized by $\wm$ \\
Residual     & $\rwm(x,a,Q)$                & Euler${}_j$, clearing, FB \eqref{eqn:irbc_euler}--\eqref{eqn:irbc_fb} \\
\addlinespace
Training measure (DEQN) & $\muerg$          & the policy's own simulated path \\
Training measure (EWM)  & $\muK$            & forced-disaster + local perturb, rolled by exact $\Gamma$ \\
\bottomrule
\end{tabular}
\caption{The framework objects of Section~\ref{sec:framework} and their international
real business cycle instances. The columns are: Framework object, the abstract object of
Section~\ref{sec:framework}; Symbol, its notation; and IRBC instance, its concrete counterpart in
the international real business cycle model. As in Table~\ref{tab:bm_map}, the methods share every row
but the training measure: DEQN imposes the residual on the policy's own ergodic path;
EWM imposes the same residual on the coverage measure $\muK$, built by forcing the
disaster ($d:=1$) and locally perturbing visited states, then rolling forward under the
exact transition $\Gamma$.}\label{tab:irbc_map}
\end{table}

\subsection{High-dimensional existence at \texorpdfstring{$N{=}8$, $16$, and $32$}{N=8, 16, 32}}\label{app:highn}

This appendix gives the full high-dimensional certification solves summarized in
Section~\ref{sub:irbc}. EWM is run with five seeds at $N{=}16$ and $N{=}32$ and two at $N{=}8$,
the surrogate-free coverage arm with three at $N{=}8$, five at $N{=}16$, and three at $N{=}32$,
and the pathwise
baseline with one seed per $N$; this is a certification-at-scale probe rather than the ten-seed
robustness grid of $N\le4$. The $N{=}8$ rows here are a
full solve, distinct from the five-seed $N{=}8$ amortization-timing benchmark of
Table~\ref{tab:irbc_scaling}, which isolates $A(\tau)$ on a fixed reference policy rather than
solving the model. No independent reference solve is feasible at the upper dimensions
($33$--$65$), so rational expectations is certified here by the held-out exact $R^{H}$ reaching
its floor, verified stationarity, and the surrogate audit ($R^{H}/R^{L}\approx1$, with
$R^{H}-R^{L}$ of the same order as the residual itself), not against an external solution. A
pre-training recertification of the monomial $R^{H}$ against a $4096$-node scrambled-Sobol QMC
rule agrees to $\lvert\Delta R^{H}\rvert\approx10^{-6}\ll4\times10^{-2}$, confirming the high-$N$
errors are not a quadrature artifact.

On the surrogate's overhead: the cost the curse of dimensionality inflates is the in-loop policy
budget $B_{\mathrm{policy}}$, paid at every gradient step, and it is smaller than the failing
baseline's at every dimension here (Table~\ref{tab:irbc_highn}). The surrogate-training budget
$B_{\mathrm{world}}$ is larger and grows with $N$, but it is a different kind of cost: incurred
off the policy loop, refit by alternation as the policy moves, and reusable and parallelizable,
it buys the dimension-robust inner loop rather than competing with it step by step. At the low dimensions where exact evaluation is
cheap ($N\le4$), even the combined $B_{\mathrm{policy}}+B_{\mathrm{world}}$ stays at or below the
surrogate-free coverage arm's total exact budget (Section~\ref{sub:irbc}), so the
training does not erase the saving. At the upper dimensions the surrogate-free coverage arm
remains trainable and certifies at all three dimensions (three of three seeds at $N{=}8$, three
of five at $N{=}16$, three of three at $N{=}32$), but pays roughly four times the surrogate's
in-loop exact budget at a higher total cost (at $N{=}32$, $B_{\mathrm{policy}}$ of $3{,}312$
against $851$ million; at $N{=}16$, $1{,}706$ against $439$), so at the upper dimensions the
surrogate buys amortization rather than feasibility, the reverse of what an existence-only
reading would suggest.

\begin{table}[t!]\centering\small

\setlength{\tabcolsep}{4.5pt}
\begin{tabular}{llcccccc}
\toprule
$N$ & Arm & on-path $R^{H}$ & disaster $R^{H}$ & disaster $R^{H}_{p99}$ & verified
& $B_{\mathrm{policy}}$ & $B_{\mathrm{world}}$ \\
\midrule
$8$  & \texttt{EWM-coverage-surrogate} & $4.3\times10^{-3}$ & $3.1\times10^{-3}$ & $9.4\times10^{-3}$ & $2/2$ & $232$  & $619$  \\
     & \texttt{DEQN-coverage-exact}    & $1.7\times10^{-3}$ & $1.0\times10^{-3}$ & $4.1\times10^{-3}$ & $3/3$ & $903$  & ---    \\
     & \texttt{DEQN-path-exact}        & $3.2\times10^{-3}$ & $1.8\times10^{-2}$ & $3.0\times10^{-2}$ & $0/1$ & $516$  & ---    \\
\addlinespace
$16$ & \texttt{EWM-coverage-surrogate} & $1.7\times10^{-3}$ & $3.3\times10^{-3}$ & $1.3\times10^{-2}$ & $5/5$ & $439$  & $1170$ \\
     & \texttt{DEQN-coverage-exact}    & $1.9\times10^{-3}$ & $1.7\times10^{-3}$ & $3.7\times10^{-3}$ & $3/5$ & $1706$ & ---    \\
     & \texttt{DEQN-path-exact}        & $3.8\times10^{-3}$ & $1.9\times10^{-2}$ & $2.8\times10^{-2}$ & $0/1$ & $975$  & ---    \\
\addlinespace
$32$ & \texttt{EWM-coverage-surrogate} & $1.7\times10^{-3}$ & $3.0\times10^{-3}$ & $5.9\times10^{-3}$ & $5/5$ & $851$  & $2270$ \\
     & \texttt{DEQN-coverage-exact}    & $1.8\times10^{-3}$ & $1.1\times10^{-3}$ & $5.9\times10^{-3}$ & $3/3$ & $3312$ & ---    \\
     & \texttt{DEQN-path-exact}        & $3.3\times10^{-3}$ & $1.8\times10^{-2}$ & $2.5\times10^{-2}$ & $0/1$ & $1892$ & ---    \\
\bottomrule
\end{tabular}
\caption{High-dimensional certification at $N{=}8$, $16$, and $32$ (state dimension
$17$, $33$, and $65$, country-modular architecture, $3500$ episodes, $p_d{=}0.001$). The columns are:
$N$, the number of countries; Arm, the solver configuration; on-path $R^{H}$, the exact held-out
residual on the on-path region; disaster $R^{H}$, its median over the disaster-region held-out set;
disaster $R^{H}_{p99}$, the 99th percentile of that disaster residual; verified, the
verified-stationarity count of seeds; $B_{\mathrm{policy}}$, the in-loop policy exact-evaluation
budget; and $B_{\mathrm{world}}$, the separable surrogate-training budget. All errors are
the exact $R^{H}$ on held-out states (median over seeds);
disaster columns are over the disaster-region set. \texttt{EWM-coverage-surrogate} reaches the
disaster-tail floor and verified stationarity at all three dimensions (five seeds at $N{=}16$ and
$N{=}32$, two at $N{=}8$) at an in-loop policy budget $B_{\mathrm{policy}}$ smaller than the failing
baseline's, plus a separable surrogate-training budget $B_{\mathrm{world}}$. The surrogate-free
\texttt{DEQN-coverage-exact} arm also certifies at all three dimensions ($3/3$ at $N{=}8$,
$3/5$ at $N{=}16$, $3/3$ at $N{=}32$) but at roughly four times the EWM in-loop budget
($B_{\mathrm{policy}}$ of $3{,}312$ against $851$ million at $N{=}32$) and a higher total exact
budget, so at the headline dimension the surrogate buys amortization rather than feasibility; the
pathwise baseline keeps the self-confirming signature (low on-path, high disaster, $0/1$ verified).
Coverage arms use three to five seeds per dimension and the baseline one (an existence rather than
ten-seed robustness probe at $N\le4$); budgets in millions of exact
evaluations.}\label{tab:irbc_highn}
\end{table}

\subsection{Supplementary IRBC exhibits}\label{app:irbc_supp}

This appendix collects the supporting IRBC exhibits kept out of the crisp main-text
presentation of Section~\ref{sub:irbc}: the amortization scaling table, the disaster-pricing
certificate, and the surrogate-capacity homotopy.

\begin{table}[t!]\centering\small

\setlength{\tabcolsep}{8pt}
\begin{tabular}{ccccc}
\toprule
$N$ & state $2N{+}1$ & nodes $4(N{+}1)$ & params dense / modular & $A(\tau)$ speedup \\
\midrule
$2$  & $5$  & $12$  & $10{,}373$ / $29{,}291$ & $43\times$ \\
$4$  & $9$  & $20$  & $11{,}145$ / $29{,}299$ & $53\times$ \\
$8$  & $17$ & $36$  & $12{,}689$ / $29{,}315$ & $61\times$ \\
$16$ & $33$ & $68$  & $15{,}777$ / $29{,}347$ & $63\times$ \\
$32$ & $65$ & $132$ & $21{,}953$ / $29{,}411$ & $70\times$ \\
\bottomrule
\end{tabular}
\caption{Architecture and amortization scaling of the international real business cycle model under
the country-modular architecture, one row per country count $N$ from $2$ to $32$. The columns are:
the number of countries $N$; the state dimension $2N{+}1$; the per-query exact continuation node
count $4(N{+}1)$ that the audited surrogate eliminates; the policy parameter count under a dense
versus the modular architecture; and $A(\tau)$, the wall-clock speed-up the surrogate buys. The
first four columns are exact. $A(\tau)$ is the median time of the exact high-fidelity scrambled-Sobol
QMC continuation divided by that of the audited surrogate $\wm$, at a fixed $\approx0.4\%$
continuation accuracy on a fixed reference policy over five seeds per $N$; it is hardware-dependent
and corroborates the node-count saving. The per-query node count $4(N{+}1)$ is the dominant
dimension-dependent in-loop exact cost and is exactly what the surrogate removes: it grows linearly
in $N$ ($\Theta(N)$, Proposition~\ref{prop:amort}, doubling with each country-doubling) under the
monomial rule, and faster for a richer quadrature rule, while the
modular parameter count stays near-constant as the state dimension grows sixteen-fold.}\label{tab:irbc_scaling}
\end{table}

\begin{figure}[t!]\centering
\includegraphics[width=\textwidth]{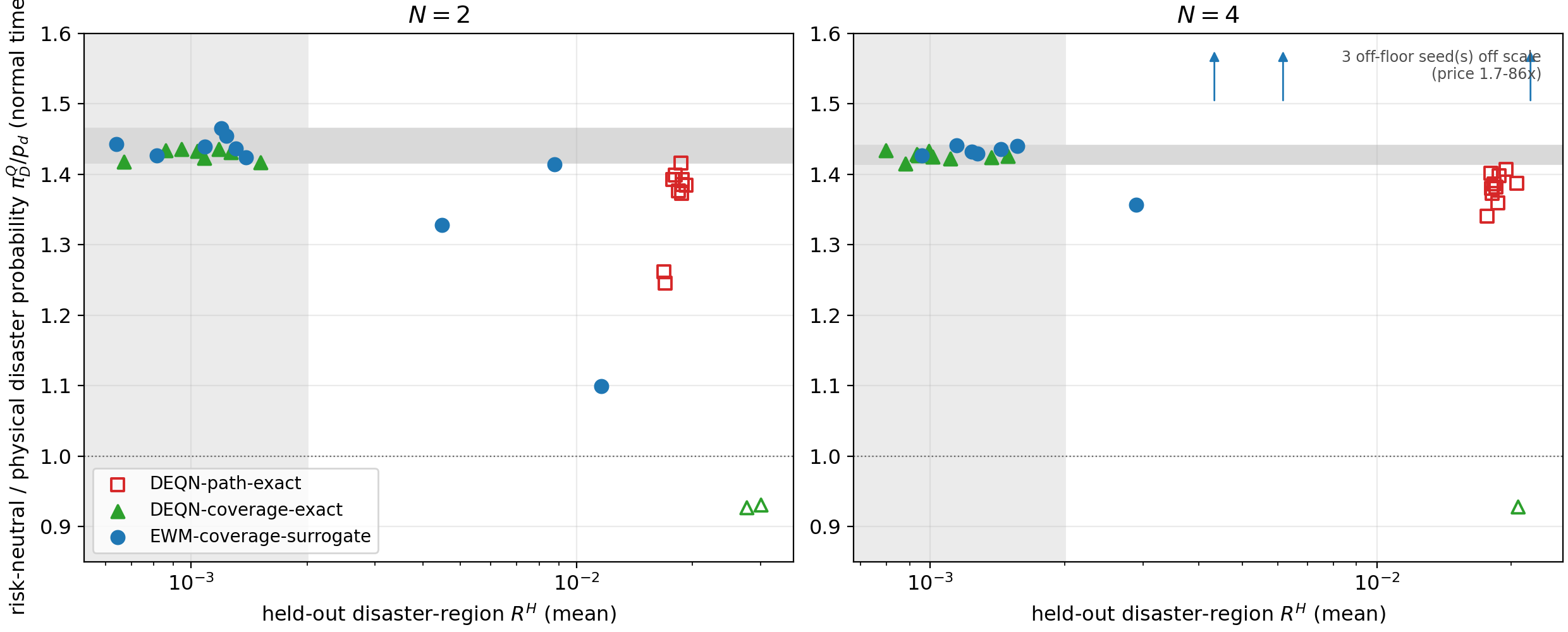}
\caption{Normal-times price of the one-period disaster Arrow claim (implied risk-neutral disaster
probability per unit physical probability) against the held-out disaster residual, $N{=}2$ and
$N{=}4$, one point per seed and arm. Shaded $x$: the certified residual floor
($R^{H}<2\times10^{-3}$); shaded $y$: the price consensus of all floor-reaching seeds across both
coverage technologies. Every floor-reaching seed prices in the band; the pathwise arm reaches the
floor in none of ten seeds and prices below it; off-floor seeds scatter with their residual, in either
direction, which is the mechanism by which an uncertified disaster residual produces an untrustworthy
disaster price.}\label{fig:irbc_pricing}
\end{figure}

\subsubsection{The surrogate-capacity homotopy: the $\widehat\eta_\kappa$--$R^{H}$ descent}
Figure~\ref{fig:irbc_kappa} runs the warm-started sieve at $N{=}2$ along the
surrogate-capacity axis: each stage grows the surrogate width $m$ ($16\!\to\!32\!\to\!64$,
the world-model capacity of Definition~\ref{def:kewm}) and its exact-target anchor fraction
($0.1\!\to\!0.2\!\to\!0.4$, the share of coverage states on which the exact continuation is
evaluated to fit $\wm$), each stage distilled from the previous converged one, with the
coverage reach held fixed. The realized surrogate gap $\widehat\eta_\kappa=\lVert\wm-\Qpi\rVert$ collapses
monotonically ($1.5\times10^{-2}\!\to\!1.7\times10^{-3}$) and the held-out disaster residual
descends with it ($3.3\times10^{-2}\!\to\!3.7\times10^{-3}$), stably across seeds
($\widehat\eta_\kappa$ is the trained surrogate's fit to the exact continuation; it
upper-bounds the in-class approximation gap $\eta_\kappa$ of Section~\ref{sec:theory} plus the
surrogate-fit error, the $\mathrm{ApproxErr}_\kappa$ and $\mathrm{SurrFitErr}$ terms of
Proposition~\ref{prop:decomp}); the endpoint
($3.7\times10^{-3}$) is of the same order as the converged \texttt{EWM} disaster-residual level of
Table~\ref{tab:irbc}. This schedule advances surrogate fidelity and capacity at fixed
coverage, so it traces the $\mathrm{ApproxErr}_\kappa$ and $\mathrm{SurrFitErr}$ descent of
Proposition~\ref{prop:decomp} rather than a coverage sweep. The Bewley experiment of
Section~\ref{sub:bewley} is complementary: it holds the exact residual, transition, coverage
configuration, and training budget fixed and varies only the distributional summary the policy reads
(raw histogram, hand moments, or learned encoder), so it isolates the world model's perception role.
The two experiments together separate the world model's integration and amortization role on the IRBC
model from its perception role on the heterogeneous-agent model.

\begin{figure}[t!]\centering
\includegraphics[width=0.86\textwidth]{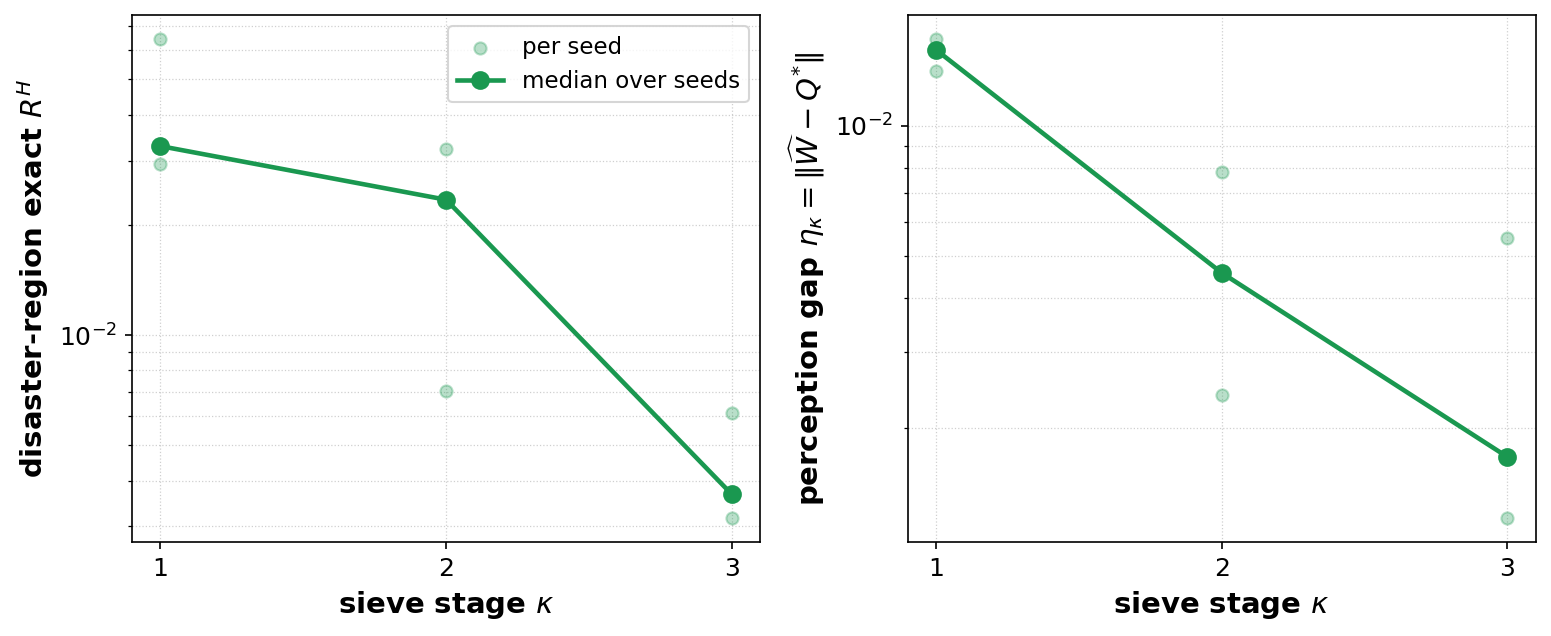}
\caption{The warm-started surrogate-capacity homotopy on the international real business cycle model:
as the surrogate is refined (width $m$ $16\!\to\!64$ and exact-target anchor fraction
$0.1\!\to\!0.4$, at fixed coverage reach, each distilled from the previous converged stage), both
the realized surrogate gap $\widehat\eta_\kappa=\lVert\wm-\Qpi\rVert$ (right) and the held-out
disaster-region exact residual $R^{H}$ (left) descend monotonically and stably across seeds, the
joint surrogate-fidelity-and-capacity descent of the $\mathrm{ApproxErr}_\kappa$/$\mathrm{SurrFitErr}$
terms. The endpoint is of the same order as the converged \texttt{EWM} disaster residual of
Table~\ref{tab:irbc}; the complementary perception experiment, which varies only the distributional
summary at fixed coverage, is the Bewley experiment (Section~\ref{sub:bewley}). Three seeds, $3000$
episodes.}
\label{fig:irbc_kappa}
\end{figure}

\section{Heterogeneous-agent (Bewley) model details}\label{app:encoder}

Section~\ref{sub:bewley} hands the policy a learned summary of the wealth distribution in place of
the distribution itself. This appendix explains what that summary (the \emph{encoder}) is and how it
is trained (the \emph{JEPA} objective), and then shows, in a case with a closed-form answer, that
what it learns to keep is exactly the information the equilibrium needs: on a genuinely
finite-dimensional population it recovers the minimal sufficient statistic, of which the
Krusell--Smith moment list is the hand-chosen special case.

\subsection{The distributional encoder and the JEPA objective}\label{app:encoder_what}

\subsubsection{The encoder}
An \emph{encoder} is a learned map $h_\psi$ that compresses a high-dimensional object into a
low-dimensional representation intended to preserve the components relevant for prediction and
decision making. Here the large object is the
cross-sectional wealth--employment distribution $\mu$, a histogram over hundreds of cells; the
encoder compresses it to a short \emph{code} $z=h_\psi(\mu)\in\mathbb{R}^{d}$ with $d$ of order a
handful, and every decision object reads $z$ in place of $\mu$. This is the same move as
summarizing a distribution by its mean and variance, except that the summary is learned from the
model's own dynamics rather than fixed in advance, and $h_\psi$ is built to be
\emph{permutation-invariant}: it sees a population, not a labelled list of households.
Figure~\ref{fig:linear_probe} shows that the short embedding keeps the wealth statistics the equilibrium
actually reads.

\subsubsection{The JEPA objective}
The encoder is trained by a \emph{Joint Embedding Predictive
Architecture} (JEPA), a way of learning a world model introduced by \citet{LeCun2022AMI} and made
to train stably end to end by the LeWorldModel of \citet{MaesLeLidec2026LeWM}, the architecture we
adapt. A JEPA has two networks: the encoder, which maps the current state to its code, and a
\emph{predictor} $P_\phi$, which forecasts the next embedding from the current one. The two are
trained together so that the predicted next embedding matches the true next embedding, the one obtained by
advancing the state one step and re-encoding. The objective predicts only the compact embedding, never
the full object: a forecaster made to reproduce every cell of next period's distribution would
spend its capacity on detail no decision uses, whereas predicting the embedding forces the encoder to
keep exactly the features needed to see one step ahead and to drop the rest.

\subsubsection{The anti-collapse regularizer}
Predictability alone admits a degenerate solution: the encoder may map every state
to the same constant embedding, which is perfectly predictable but uninformative for the
policy, because it cannot distinguish economically different populations. This failure is called
\emph{collapse}. We rule out this degeneracy with a single anti-collapse regularizer, SIGReg
\citep{BalestrieroLeCun2025SIGReg}, which requires the embeddings to spread out and fill an isotropic Gaussian. It
checks this cheaply, by projecting the embeddings onto many random directions and applying a
one-dimensional normality test to each. In the population ideal, matching all one-dimensional
projections characterizes the full joint distribution by the Cram\'er--Wold device, and SIGReg
implements a finite random-projection approximation to this criterion. The whole objective is then two terms,
predict the next embedding plus keep the embeddings spread, with no stop-gradient, target network, or
pretrained encoder \citep{MaesLeLidec2026LeWM}.

\subsubsection{The economic adaptation}
One thing changes when this recipe enters an equilibrium
model, and it is the disciplining difference of the paper. In a vision world model the next state
is unknown and the transition must itself be learned from data; here the ``observation'' is the
population $\mu_t$ and its one-step law is the \emph{exact, known} push-forward $\Phi$. The target
next embedding is therefore $h_\psi(\Phi(\mu_t))$, formed by advancing the true dynamics and
re-encoding, so $\Phi$ is never learned and never approximated: only the \emph{summary} is learned,
while the dynamics that move it stay exact. The rest of this appendix makes precise what that learned summary
turns out to be.

\subsection{What the encoder learns}\label{app:encoder_learns}

\subsubsection{Setup}
Let the cross-sectional distributions that the dynamics ever visit lie in a
$d$-parameter exponential family,
\begin{equation}
\mu_\theta(x) = h(x)\,\exp\!\big(\theta^\top T(x) - A(\theta)\big),
\qquad \theta\in\Theta\subseteq\mathbb{R}^{d},
\label{eqn:expfam}
\end{equation}
with sufficient statistic $T:\mathcal{X}\to\mathbb{R}^{d}$ (for a Gaussian, $T(x)=(x,x^2)$ and
$d=2$, the mean and variance). A \emph{sufficient statistic} is the short list of numbers that
carries everything about $\mu$ the decision uses; once it is known, the rest of the distribution is
irrelevant detail. Two facts make this the right idealization of the Bewley state.
First, every aggregate the residual \eqref{eqn:bewley_resid} needs is a moment of $\mu$: capital
$K=\int a\,d\mu$, hence prices $r,w$, and therefore the continuation $Q^{H}(\cdot,\mu)$ depend on
$\mu$ only through $\mathbb{E}_{\mu}[T]$. Second, $\theta\mapsto\mathbb{E}_{\mu_\theta}[T]$ is a
bijection on the interior of $\Theta$ (a standard property of minimal exponential families). So
the map
\begin{equation}
e^\star:\ \mu_\theta\ \longmapsto\ \mathbb{E}_{\mu_\theta}[T]\in\mathbb{R}^{d}
\label{eqn:ideal_enc}
\end{equation}
is a $d$-dimensional encoder that discards nothing the solver uses: the policy and the
residual factor through $e^\star(\mu)$ exactly. The raw $M\times2$ histogram of
Section~\ref{sub:bewley} is, on such a population, a redundant coordinate system for a
$d$-dimensional object.

\subsubsection{Krusell--Smith as a hand-chosen encoder}
The moment sieve fixes $T$ in advance
(mean of assets, variance, a constraint mass, a top share). When the model's true sufficient
statistic is spanned by that list, the hand encoder equals $e^\star$ up to a reparameterization
and the solver is correct; this is the regime of approximate aggregation, in which the mean of
the wealth distribution very nearly suffices \citep{krusell1998}. When it is not, when a
binding constraint and a regime change make a feature of $\mu$ outside the chosen list matter for
the continuation, the hand list omits the relevant coordinate, the policy cannot tell the
populations apart along it, and the solver settles on a locally consistent but wrong stationary
distribution, the cross-section failure documented in Section~\ref{sub:bewley} (the moment sieve's
stationary distribution sits far from the reference).

\subsubsection{What the JEPA objective recovers}
The encoder is never told $T$; it has to discover it.
Suppose it is trained as above, to zero prediction error with non-degenerate (uncollapsed) embeddings.
At such a minimizer $h_\psi$ is an invertible relabelling of the coarsest statistic $S(\mu)$ for
which $\big(S(\mu_t)\big)_t$ is itself a Markov chain under $\Phi$, that is, of the minimal
sufficient statistic of the distributional dynamics. On the family \eqref{eqn:expfam} that
statistic is $e^\star(\mu)$, up to coordinates that are not excited by the transition $\Phi$.
So the learned embedding recovers $e^\star$ up to a smooth change of coordinates, which
is all the policy needs. One qualifier carries to the general case: the objective keeps whatever
makes the embedding predictable, which can be slightly larger than what the residual reads. A coordinate
the dynamics stir but the residual never loads on is dynamically relevant yet decision-irrelevant;
on the family \eqref{eqn:expfam} the two coincide.
Figure~\ref{fig:jepa_econ} summarizes the construction and the role of SIGReg.

\begin{figure}[t!]\centering
\resizebox{\textwidth}{!}{
\begin{tikzpicture}[>=Stealth, font=\small,
  dist/.style={rectangle, rounded corners=2pt, draw=blue!55!black, thick, fill=blue!5,
     minimum width=2.7cm, minimum height=1.05cm, align=center},
  code/.style={circle, draw=red!60!black, thick, fill=red!6, minimum size=0.95cm},
  op/.style={font=\footnotesize\itshape, align=center},
  ar/.style={->, thick},
  jloss/.style={<->, thick, dashed, red!55!black}]

\node[dist] (mu)  at (0,0)      {$\mu_t$\\[-1pt]\scriptsize wealth distribution};
\node[code] (zt)  at (3.8,0)    {$z_t$};
\node[code] (zh)  at (8.0,0)    {$\hat z_{t+1}$};
\node[dist] (mu1) at (3.8,-2.1) {$\mu_{t+1}$\\[-1pt]\scriptsize $=\Phi(\mu_t)$};
\node[code] (zt1) at (8.0,-2.1) {$z_{t+1}$};

\draw[ar] (mu)  -- node[op,above]{$h_\psi$\\encoder} (zt);
\draw[ar] (zt)  -- node[op,above]{$P_\phi$\\predictor} (zh);
\draw[ar] (mu)  -- node[op, pos=0.42, left=7pt]{$\Phi$ (exact)} (mu1);
\draw[ar] (mu1) -- node[op,above]{$h_\psi$} (zt1);
\draw[jloss] (zh) -- node[op,right]{JEPA loss\\$\|\hat z_{t+1}-z_{t+1}\|^2$} (zt1);

\begin{scope}[shift={(12.0,-1.05)}]
  \node[op] at (0,1.5) {\textbf{SIGReg}};
  \draw[rounded corners, draw=black!35, fill=black!3] (-1.05,0.35) rectangle (1.05,1.2);
  \fill[red!60!black] (0,0.78) circle (1.7pt);
  \node[font=\scriptsize, text=black!65] at (0,0.12) {collapse $h_\psi\!\equiv\!c$};
  \draw[->, thick, black!55] (0,-0.15) -- (0,-0.95);
  \draw[rounded corners, draw=black!35, fill=black!3] (-1.05,-2.55) rectangle (1.05,-1.2);
  \foreach \p in {(-0.55,-1.45),(0.42,-1.4),(-0.2,-1.9),(0.5,-2.05),(-0.55,-2.15),(0.12,-1.65),(0.02,-2.3),(0.6,-1.75)}
     \fill[red!60!black] \p circle (1.5pt);
  \node[font=\scriptsize, text=black!65] at (0,-2.92) {isotropic Gaussian};
\end{scope}
\end{tikzpicture}}
\caption{The world model's encoder, drawn for the heterogeneous-agent economy. The state is a
cross-sectional wealth distribution $\mu_t$. A learned encoder $h_\psi$ maps it to a
low-dimensional embedding $z_t$, and a predictor $P_\phi$ forecasts next period's embedding
$\hat z_{t+1}$; the target $z_{t+1}$ is obtained by pushing the exact law of motion
$\Phi$ forward to $\mu_{t+1}=\Phi(\mu_t)$ and then encoding, and the JEPA objective trains
$(h_\psi,P_\phi)$ so the two match. Forcing an embedding to predict its own future is what makes $z$
keep exactly what the dynamics need, the minimal sufficient statistic of the distributional
dynamics (on a Gaussian population the mean and variance, $e^\star(\mu)=\mathbb{E}_\mu[T]$).
The SIGReg term (right) spreads the embeddings to fill an isotropic Gaussian, ruling out the
degenerate collapse $h_\psi\equiv c$ in which every population maps to one embedding and the policy
can no longer tell them apart. The transition $\Phi$ is exact and never learned; only the
summary is.}
\label{fig:jepa_econ}
\end{figure}
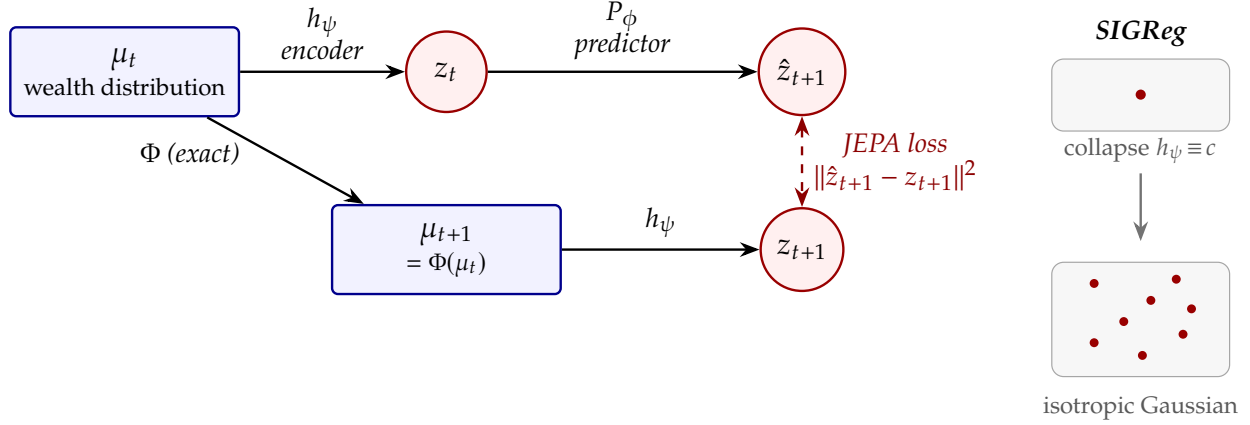

\subsubsection{The linear--Gaussian instance}
When the dynamics are linear the recovered
embedding has a familiar name. Suppose the deviation of $\mu$ from its
mean evolves linearly, $m_{t+1}=Fm_t+\text{(exact shock)}$ with $m_t\in\mathbb{R}^{p}$ the vector
of centered moments and $F$ the linearized push-forward $\Phi$, and that the policy needs the
linear functional $g^\top m_t$. The smallest embedding that keeps $h_\psi$ Markov and preserves
$g^\top m$ is the projection of $m$ onto the $F$-invariant subspace generated by $g$, the
observable/controllable subspace of the pair $(F,g)$; the optimal $d$-dimensional linear encoder
is its leading principal subspace, a principal-component analysis of the dynamics rather
than of a static sample. In words, the encoder runs a principal-component analysis not on a
snapshot of the population but on how the population moves, keeping the few directions that
both stir under $\Phi$ and feed the decision. The Krusell--Smith ``use the mean'' rule is the $d{=}1$ truncation that
keeps only the first such direction; the learned encoder keeps the $d$ directions the transition
actually excites and the residual actually loads on, which is why the
$\ledgerBewleyEWMSieveDim$-dimensional learned embedding of Section~\ref{sub:bewley} recovers the
cross-section that a fixed, hand-picked moment list misses, not by being shorter (it is not) but by
spending its dimensions on the directions that matter rather than on a prechosen menu.

\paragraph{Capacity check.} The embedding dimension is set by a capacity sweep internal to the
training design, not by tuning to the reference distribution: the reported
$\ledgerBewleyEWMSieveDim$-dimensional embedding is the smallest that both passes the held-out
residual audit and recovers the reference cross-section, which is why it is the headline choice in
Section~\ref{sub:bewley}.

\clearpage
\bibliographystyle{apalike}
\bibliography{references}

\end{singlespace}
\end{document}

%% file: bewley_ledger_values.tex
\newcommand{\ledgerBewleyNumSeeds}{10}
\newcommand{\ledgerBewleyBudgetEpisodes}{3000}
\newcommand{\ledgerBewleyFullMuDim}{200}
\newcommand{\ledgerBewleyKSMomentsDim}{5}
\newcommand{\ledgerBewleyEWMSieveDim}{8}
\newcommand{\ledgerBewleyCertEulerThreshold}{0.03}
\newcommand{\ledgerBewleyCertDriftThreshold}{0.05}
\newcommand{\ledgerBewleyCertConstraintMassThreshold}{0.5}

\newcommand{\ledgerBewleyFullMuEulerMean}{$2.6\times10^{-2}$}
\newcommand{\ledgerBewleyFullMuConstraintMass}{$<\!10^{-3}$}
\newcommand{\ledgerBewleyFullMuLoneRefMedian}{0.69}
\newcommand{\ledgerBewleyFullMuLoneRefMin}{0.67}
\newcommand{\ledgerBewleyFullMuLoneRefMax}{0.69}
\newcommand{\ledgerBewleyKSMomentsEulerMean}{$6.5\times10^{-3}$}
\newcommand{\ledgerBewleyKSMomentsConstraintMass}{0.002}
\newcommand{\ledgerBewleyKSMomentsLoneRefMedian}{0.26}
\newcommand{\ledgerBewleyKSMomentsLoneRefMin}{0.25}
\newcommand{\ledgerBewleyKSMomentsLoneRefMax}{0.46}
\newcommand{\ledgerBewleyEWMSieveEulerMean}{$1.6\times10^{-2}$}
\newcommand{\ledgerBewleyEWMSieveConstraintMass}{0.004}
\newcommand{\ledgerBewleyEWMSieveLoneRefMedian}{0.13}
\newcommand{\ledgerBewleyEWMSieveLoneRefMin}{0.11}
\newcommand{\ledgerBewleyEWMSieveLoneRefMax}{0.17}

\newcommand{\ledgerBewleyKSToEWMLoneRatio}{1.9}
\newcommand{\ledgerBewleyRawCollapseLoneRef}{0.69}

\newcommand{\ledgerBewleyProbeGiniRtwo}{0.99}

\newcommand{\ledgerBewleyProbeVarRtwo}{0.97}

\newcommand{\ledgerBewleyEncoderExponent}{0.92}
\newcommand{\ledgerBewleyHistogramExponent}{1.65}
\newcommand{\ledgerBewleyEncoderMax}{250{,}000}
\newcommand{\ledgerBewleyHistMemWall}{15{,}000}


%% file: bewley_ledger_stubs.tex
\providecommand{\ledgerBewleyNumSeeds}{20}
\providecommand{\ledgerBewleyBudgetEpisodes}{3000}
\providecommand{\ledgerBewleyFullMuDim}{200}        
\providecommand{\ledgerBewleyKSMomentsDim}{5}
\providecommand{\ledgerBewleyEWMSieveDim}{8}

\providecommand{\ledgerBewleyCertEulerThreshold}{0.03}
\providecommand{\ledgerBewleyCertDriftThreshold}{0.05}
\providecommand{\ledgerBewleyCertConstraintMassThreshold}{0.5}

\providecommand{\ledgerBewleyFullMuEulerMean}{\ledgerTODO}
\providecommand{\ledgerBewleyKSMomentsEulerMean}{\ledgerTODO}
\providecommand{\ledgerBewleyEWMSieveEulerMean}{\ledgerTODO}

\providecommand{\ledgerBewleyFullMuConstraintMass}{\ledgerTODO}
\providecommand{\ledgerBewleyKSMomentsConstraintMass}{\ledgerTODO}
\providecommand{\ledgerBewleyEWMSieveConstraintMass}{\ledgerTODO}

\providecommand{\ledgerBewleyFullMuLoneRefMedian}{\ledgerTODO}
\providecommand{\ledgerBewleyFullMuLoneRefMin}{\ledgerTODO}
\providecommand{\ledgerBewleyFullMuLoneRefMax}{\ledgerTODO}
\providecommand{\ledgerBewleyKSMomentsLoneRefMedian}{\ledgerTODO}
\providecommand{\ledgerBewleyKSMomentsLoneRefMin}{\ledgerTODO}
\providecommand{\ledgerBewleyKSMomentsLoneRefMax}{\ledgerTODO}
\providecommand{\ledgerBewleyEWMSieveLoneRefMedian}{\ledgerTODO}
\providecommand{\ledgerBewleyEWMSieveLoneRefMin}{\ledgerTODO}
\providecommand{\ledgerBewleyEWMSieveLoneRefMax}{\ledgerTODO}

\providecommand{\ledgerBewleyKSToEWMLoneRatio}{\ledgerTODO}
\providecommand{\ledgerBewleyRawCollapseLoneRef}{\ledgerTODO}

\providecommand{\ledgerBewleyEncoderExponent}{\ledgerTODO}
\providecommand{\ledgerBewleyHistogramExponent}{\ledgerTODO}
\providecommand{\ledgerBewleyHistMemWall}{\ledgerTODO}
\providecommand{\ledgerBewleyEncoderMax}{\ledgerTODO}

\providecommand{\ledgerBewleyProbeGiniRtwo}{0.99}
\providecommand{\ledgerBewleyProbeVarRtwo}{0.97}

%% file: tab_bmdc.tex
\setlength{\tabcolsep}{4pt}%
\begin{tabular}{lccccc}
\toprule
Arm & normal & disaster mean & disaster $p_{95}$ & $B_{\mathrm{policy}}$ & low-error \\
\midrule
\texttt{DEQN-path-exact} & \shortstack{$3.3\times10^{-4}$\\[1pt] {\scriptsize[$3.0\times10^{-4}$,\,$3.4\times10^{-4}$]}} & \shortstack{$5.3\times10^{-3}$\\[1pt] {\scriptsize[$4.8\times10^{-3}$,\,$5.6\times10^{-3}$]}} & \shortstack{$1.7\times10^{-2}$\\[1pt] {\scriptsize[$1.5\times10^{-2}$,\,$1.8\times10^{-2}$]}} & $2212$ & 10/10 \\
\texttt{DEQN-coverage-exact} & \shortstack{$2.5\times10^{-4}$\\[1pt] {\scriptsize[$2.2\times10^{-4}$,\,$3.1\times10^{-4}$]}} & \shortstack{$7.3\times10^{-4}$\\[1pt] {\scriptsize[$6.1\times10^{-4}$,\,$9.2\times10^{-4}$]}} & \shortstack{$1.7\times10^{-3}$\\[1pt] {\scriptsize[$1.2\times10^{-3}$,\,$2.1\times10^{-3}$]}} & $7962$ & 10/10 \\
\texttt{EWM-coverage-surrogate} & \shortstack{$3.9\times10^{-4}$\\[1pt] {\scriptsize[$3.6\times10^{-4}$,\,$4.7\times10^{-4}$]}} & \shortstack{$1.6\times10^{-3}$\\[1pt] {\scriptsize[$1.3\times10^{-3}$,\,$1.9\times10^{-3}$]}} & \shortstack{$3.5\times10^{-3}$\\[1pt] {\scriptsize[$3.2\times10^{-3}$,\,$4.0\times10^{-3}$]}} & $1327$ & 10/10 \\
\bottomrule
\end{tabular}